\documentclass[11pt,letterpapaer]{article}

\usepackage{mdframed}

\usepackage{amsfonts}
\usepackage{amssymb}
\usepackage[english]{babel}
\usepackage{amsmath,amsthm}
\usepackage{mathtools}
\usepackage{epsfig}
\usepackage{euscript}
\usepackage{color}
\usepackage{fancyvrb}
\usepackage{fullpage}
\usepackage{graphicx}
\usepackage{float}
\usepackage{caption}
\usepackage{subcaption}
\usepackage{physics}
\usepackage{bbm}
\usepackage{bbold}
\usepackage{tikz}
\newcommand{\caU}{\mathcal U}
\newcommand{\caG}{\mathcal G}
\newcommand{\caH}{\mathcal H}

\newcommand{\caP}{\mathcal P}

\newcommand{\caE}{\mathcal{E}}
\newcommand{\caX}{\mathcal{X}}

\newcommand{\wc}{{w_c}}
\newcommand{\wq}{{w_q}}
\usepackage{enumitem}

\newcommand{\supp}{\mathrm{supp}}
\newcommand{\bfs}{\mathbf{S}}

\newcommand{\ext}{\mathrm{ext}}
\newcommand{\mqs}{\mathcal{M}}
\newcommand{\dstar}{{d_*}}
\newcommand{\rbleq}{\leq_{\text{rb}}}

\usepackage{thmtools, thm-restate}

\newtheorem{theorem}{Theorem}[section]
\newtheorem{proposition}{Proposition}[section]
\newtheorem{lemma}{Lemma}[section]
\newtheorem{definition}{Definition}[section]

\usepackage[bookmarks=true,colorlinks=true,linkcolor=blue, urlcolor=blue, citecolor=blue]{hyperref}

\usepackage[backend=bibtex,style=numeric-comp,sorting=none]{biblatex}
\usepackage{nicematrix}

\numberwithin{equation}{section}

\title{
LDPC stabilizer codes as gapped quantum phases:\\
stability under graph-local perturbations
}

\author{Wojciech De Roeck${}^{1}$, Vedika Khemani${}^{2}$, Yaodong Li${}^{2}$, Nicholas O'Dea${}^{2}$, Tibor Rakovszky${}^{2,3,4}$ }
\date{\small $^1$ Instituut voor Theoretische Fysica, KU Leuven, Belgium\\%
    $^2$Department of Physics, Stanford University, Stanford, CA 94305 \\
    $^3$ Department of Theoretical Physics, Institute of Physics,
Budapest University of Technology and Economics,
M\H{u}egyetem rkp. 3., H-1111 Budapest, Hungary\\
$^4$ 
HUN-REN-BME Quantum Error Correcting Codes and Non-equilibrium Phases Research Group,
Budapest University of Technology and Economics,
M\H{u}egyetem rkp. 3., H-1111 Budapest, Hungary\\}

\begin{document}
\captionsetup{width=0.8\textwidth}
\captionsetup{font=footnotesize}

\maketitle

\abstract{
We generalize the proof of stability of topological order, due to Bravyi, Hastings and Michalakis, to stabilizer Hamiltonians corresponding to low-density parity check (LDPC) codes without the restriction of geometric locality in Euclidean space. 
We consider Hamiltonians $H_0$ defined by $[[N,K,d]]$ LDPC codes which obey certain topological quantum order conditions: (i) code distance $d \geq c \log(N)$, implying local indistinguishability of ground states, and (ii) a mild condition on local and global compatibility of ground states; these include good quantum LDPC codes, and the toric code on a hyperbolic lattice, among others. 
We  consider stability under weak perturbations that are quasi-local on the interaction graph defined by $H_0$, and which can be represented as sums of bounded-norm terms. As long as the local perturbation strength is smaller than a finite constant, we show that the perturbed Hamiltonian has well-defined spectral bands originating from the $O(1)$ smallest eigenvalues of $H_0$. The band originating from the smallest eigenvalue has $2^K$ states, is separated from the rest of the spectrum by a finite energy gap, and has exponentially narrow bandwidth $\delta = C N e^{-\Theta(d)}$, which is tighter than the best known bounds even in the Euclidean case.  We also obtain that the new ground state subspace is related to the initial code subspace by a quasi-local unitary, allowing one to relate their physical properties. 
Our proof uses an iterative procedure that performs successive rotations to eliminate non-frustration-free terms in the Hamiltonian. Our results extend to quantum Hamiltonians built from classical LDPC codes, which give rise to stable symmetry-breaking phases. These results show that LDPC codes very generally define stable gapped quantum phases, even in the non-Euclidean setting, initiating a systematic study of such phases of matter.

\tableofcontents

 % \newpage
\section{Introduction}

Over the last two decades, ideas from quantum information theory have been transformational in our understanding of quantum phases of matter~\cite{zeng2019quantum}. 
This has been particularly notable in the study of \emph{zero temperature topological order}. Stabilizer codes~\cite{calderbank1996good,steane1996multiple,gottesman1997stabilizer}, the most well-studied examples of error correcting codes, also define gapped and exactly-solvable Hamiltonians which are sums of commuting projectors, such that the ground-state space of the Hamiltonian is the same as the code space. Encoding a number of qubits $K>0$ corresponds to a $2^K$-fold ground-state degeneracy for the corresponding Hamiltonian\footnote{The code rate of a code is $K/N$.};  likewise, features that endow the code with robustness against errors -- such as local indistinguishability of  code states~\cite{knill1997theory} --  also define non-trivial topological order in the ground states.  Kitaev's toric code~\cite{kitaev2003fault} stands as a paradigmatic example of $\mathbb{Z}_2$ topological order~\cite{fradkinshenker1979}, while various other examples such as fracton models~\cite{chamon2005quantum, haah2011local} describe distinct types of quantum order. 

This correspondence between codes and quantum order is powerful because stabilizer Hamiltonians can serve as representative fixed points which capture the properties of an entire \emph{phase} of matter. This was established in the seminal work of Bravyi, Hastings, and Michalakis (BHM) ~\cite{Bravyi_2010} (see also~\cite{Klich_2010, Bravyi_2011, Michalakis_2013}) which showed that zero temperature topological order in gapped, commuting-projector Hamiltonians remains \emph{stable} to \emph{all} sufficiently weak perturbations comprised of a sum of sufficiently local terms (as long as the local perturbation strength is below a constant threshold, which remains finite in the thermodynamic limit). 
More specifically, BHM showed that (i) the gap above the ground state subspace remains finite in the perturbed model, and that (ii) the degeneracy of the ground states is split into a narrow spectral band whose width $\delta$ scales to zero exponentially with increasing system size, so that the ground state degeneracy is robust. Furthermore, (iii) the low energy subspaces of the unperturbed and perturbed Hamiltonians are related by quasi-local unitary transformations, which cannot destroy the  patterns of long-range entanglement characteristic of topologically ordered ground states~\cite{hastings2005quasiadiabatic, chen2010local}. In other words, despite being solvable, stabilizer models are \emph{not} fine-tuned points in parameter space with respect to their physical properties. The results of BHM apply to geometrically local Hamiltonians on $D$-dimensional Euclidean lattices which obey certain
topological quantum order conditions (referred to as TQO-I and TQO-II). Nowadays, the most widely used result on ground state stability is \cite{bravyi2011short,Michalakis_2013,nachtergaele2022quasi3}, which, in contrast to BHM, is not restricted to commuting projector Hamiltonians. That restriction is replaced by local versions of the gap condition and the TQO-conditions. 
However, as it stands also this result is restricted to Euclidean lattices.

Recently, there have been a number of breakthroughs in quantum error correction (QEC) on the topic of low-density parity check (LDPC)  codes defined on \emph{non-Euclidean} expander graphs ~\cite{sipser1996expander,
hastings2021fiber,panteleev2021degenerate,panteleev2021quantum,panteleev2022asymptotically,breuckmann2021balanced,leverrier2022quantum,dinur2023good,lin2022good}. 
The LDPC condition imposes that each stabilizer check acts on only a constant number of qubits, and that every qubit is only part of  a constant number of checks; besides these requirements, completely general interaction graphs are allowed. Thus, the LDPC condition imposes a notion of graph-locality, but otherwise allows for spatially non-local interactions and/or non-Euclidean geometries, which are outside the scope of BHM's results. Of particular interest are ``good" LDPC codes defined on \emph{higher dimensional expanders}~\cite{
hastings2021fiber,panteleev2021degenerate,panteleev2021quantum,panteleev2022asymptotically,breuckmann2021balanced,leverrier2022quantum,dinur2023good,lin2022good},  which exhibit optimal scaling for certain metrics of error correction which can provably not be realized in local, Euclidean geometries~\cite{bpt2010}\footnote{Embedding non-Euclidean LDPC models within Euclidean geometries will necessarily produce non-local and/or non-planar interactions.}. Expanders have the property that the surface area of a region scales proportionally to its volume; notably,  the volume of a ball of radius $r$ on an expander graph can grow \emph{exponentially} with $r$.  In good codes, both $K$ (the number of logical qubits
encoded) and the code distance $d$ (the size of the smallest  operator that distinguishes states in the codespace) are proportional to $N$ (the number of physical
qubits), which is the optimal scaling for the density and robustness of encoded information. Quantum codes that satisfy this property were discovered only recently, after a decades-long quest. 
Moreover, such non-Euclidean geometries are rapidly becoming experimentally accessible in various quantum computing and simulation platforms~\cite{Kollar2019,Periwal2021,Bluvstein2022,LukinLDPC}. 

These developments raise a natural question, central to this work:
\begin{mdframed}
    \emph{Can we associate non-local and/or non-Euclidean  LDPC stabilizer codes with stable phases of quantum matter?}
\end{mdframed}

It is already apparent that non-Euclidean quantum LDPC codes can display novel and interesting properties from the perspective of physics. A subset of the present authors have recently embarked on the program of understanding these codes from a physics perspective, and have shown that these models can be understood as unconventional gauge theories~\cite{rakovszky2023gauge, rakovszky2024product, li2024perturbative}.
In~\cite{freedman2013quantum, Anshu_2023}, it was shown that good quantum LDPC codes can exhibit a particularly striking and strong form of topological order called (called ``No Low-energy Trivial States'' (NLTS)), which crucially relies on the bulk and boundary scaling proportionally on expander graphs. The question remains: are the novel properties of the Hamiltonians defined by these  codes \emph{stable} to weak graph-local perturbations, in analogy with BHM's results for local Euclidean models?  Establishing  stability under perturbations is a necessary first step for associating these codes with stable \emph{phases}.

On one hand, it is physically intuitive to expect all LDPC codes to be stable. They are gapped, and the LDPC condition ensures that some aspects of local physics is retained, including extensivity of energy and generalized Lieb-Robinson bounds on graphs~\cite{hastings2006spectral} (such bounds play an important role in BHM's results). Most importantly, if codes have a code distance $d$ which diverges with $N$, it means that ground states are locally indistinguishable because operators must have support on at least $d$ qubits to distinguish different ground states. Thus, if the perturbation is comprised of graph-local terms of strength $\epsilon <1$, an operator of weight $d$ is only generated at $d$th order in degenerate perturbation theory,  suggesting that ground states get split by an amount $\epsilon^d$ which vanishes as $N\rightarrow\infty$. In line with this intuition, the local indistinguishability condition is precisely the TQO-I condition which enters BHM's rigorous stability proof for local, Euclidean models. A related point is that it is proven that $d \geq c\log(N)$ for a constant $c>0$ is sufficient to ensure that LDPC codes have a finite \emph{decoding threshold}~\cite{Aharonov_1999, pryadko2014} to all local errors, even in non-local and non-Euclidean settings. This is a different notion of stability (and is the property of a non-equilibrium error channel rather than a local Hamiltonian), but this property is also intrinsically tied to local indistinguishability and the topological character of the code space.  In \cite{li2024perturbative}, a subset of us were able to use the existence of a decoding threshold to physically argue for stability of the gap and ground-state degeneracy of LDPC Hamiltonians, but our results were for a limited class of models and perturbations (CSS LDPC codes perturbed by purely $X$ or $Z$ fields), and these arguments do not constitute a rigorous proof. 

Despite these intuitive reasons for optimism, there are various issues which also make the stability analysis  subtle, as recently noted in~\cite{lavasani2024klocal}. Good codes have $K \propto N$, corresponding to an exponentially large ground state degeneracy and an extensive ground state entropy. Thus, within degenerate perturbation theory, even exponentially small matrix elements between these exponentially many states could produce large energy differences, thereby destroying the ground state degeneracy. More physically, stability of the ground state degeneracy would imply a stable violation of the third law of thermodynamics, which is certainly unconventional from the perspective of statistical physics\footnote{This might seem alarming, but we remind the reader that such violations would occur on expander graphs in which bulk and boundary scale proportionally, which is outside the purview of standard treatments in statistical mechanics.}.

In \cite{lavasani2024klocal}, Lavasani \textit{et.\,al.} extended the approach of \cite{Bravyi_2011} to prove the stability of a large class of $k$-local (but geometrically non-local) LDPC codes. Their results significantly extended BHM's work, but fell short of proving stability on  expander graphs in which the volume of a ball of radius $r$ grows exponentially with $r$. This left open the stability question for important examples such as the surface code on hyperbolic tilings~\cite{Breuckmann_2016} and  ``good'' quantum LDPC codes~\cite{dinur2021locally, panteleev2021asymptotically, leverrier2022quantum, dinur2022good}. The reason Ref.~\cite{lavasani2024klocal}'s analysis did not extend to expander graphs was because, following \cite{Bravyi_2011}, the authors used the formalism of quasi-adiabatic continuity in their analysis. This converts local operators to quasi-local operators; However, the bounds on locality (and Lieb-Robinson velocities) were not strong enough to produce operators with exponentially decaying tails to counteract the exponential growth of volumes.

% Recently, Lavasani\,\textit{et.\,al.}~\cite{lavasani2024klocal} extended the approach of \cite{Bravyi_2011} to prove the perturbative stability of a large class of LDPC codes, but still leaving open those living on sufficiently expanding graphs, including the hyperbolic surface code~\cite{Breuckmann_2016} and the class of ``good'' qLDPC codes~\cite{dinur2021locally, panteleev2021asymptotically, leverrier2022quantum, dinur2022good}.

% Meanwhile, such models present new challenges to the existing theoretical framework.
% In particular, their possibly extensive ground state entropy (apparently violating the third law of thermodynamics) 

\subsection{Informal statement of conditions and results}
In this work, we provide a definitive affirmative answer to the perturbative stability of \textit{all} quantum LDPC stabilizer codes with a sufficiently large code distance and under reasonable TQO conditions; these conditions are stated precisely in Secs.~\ref{subsec:TQOPauli} and \ref{subsec:TQOword}, but they are a natural generalization of the TQO conditions in BHM: 
\\
\begin{itemize}[labelsep=0em, leftmargin=5.5em, labelwidth=5em, itemindent=0em, align=parleft]
    \item [\textbf{TQO-I}: ] The ground subspace is the codespace of a quantum code with code distance $d\geq c\log(N)$ for constant $c>0$. 

    \item [\textbf{TQO-II}: ] Local ground subspaces are consistent with the global one. 
\end{itemize}
The first statement is a statement of local indistinguishability, and BHM discuss via the example of an ``unstable toric code" why TQO-I is necessary but not sufficient~\cite{Bravyi_2010}.  

Informally, we prove that for such Hamiltonians, all \textit{graph-local perturbations perturb graph-locally} \textit{i.e.} the ground states of the perturbed Hamiltonian and of the unperturbed one are related by quasi-local unitary transformations. In particular, for weak and extensive perturbations comprised of sums of bounded norm interactions with an appropriate local strength smaller than a finite threshold, we show that (i) the spectral gap remains finite in the perturbed model, and (ii) that the ground state splitting is exponentially small in the code distance, scaling as $\delta \leq C N \exp(- c' d)$ for finite constants $C,c' > 0$. This splitting decays to zero in the thermodynamic limit, so the ground state degeneracy and gap remain robust.

When applied to local LDPC codes on Euclidean lattices, our results significantly strength\-en parametric bounds on the ground state splitting relative to those obtained by BHM. As an example, when applied to the $2D$ toric code on an $L\times L$ lattice with $d\propto L$, our results predict the optimal scaling $\delta \leq C L^2\exp{-c' L}$, while BHM obtain a weaker bound $\delta = \textrm{poly}(L)\exp(-c_2 L^{3/8})$ with $c', c_2>0$. We also require a weaker distance bound $d \geq c\log(N)$ relative to BHM who require $d \geq cL$. 

With minimal additional effort, we extend our results to: 
\begin{enumerate}
    \item[(i)] Topologically trivial and local gapped paramagnetic stabilizer models with a unique ground state, on any geometry.
    Indeed, these models belong to the class of models we specify in Sec.~\ref{sec:summary}, with the specification that $K = 0$, and their stability is established by noting that nowhere in our proof do we invoke conditions on $K$.
    
    \item [(ii)] Classical LDPC codes with symmetry-breaking order, under symmetry-respecting perturbations.
    Examples range from perturbed Ising models on any graph to perturbed ``good" classical LDPC codes on expander graphs. In general, a classical code that encodes $K$ bits has a symmetry group $\mathbb{Z}_2^K$ (see \cite{rakovszky2023gauge} for a discussion of this point) which should be respected by the perturbations. 
    We detail this discussion in Sec.~\ref{sec:classical}.

    % \item [(iii)] Perturbations which are non-Hermitian, while still being sufficiently local and weak, relevant to Glauber dynamics and its quantum generalizations. 
    
\end{enumerate}

The key technical tool enabling our result is an iterative procedure that performs successive rotations to iteratively eliminate non-frustration-free terms in the Hamiltonian (as it is inspired by the iterative procedure in the proof of the KAM theorem ~\cite{kolmogorov1954, moser1962, arnold1963}, we call this procedure ``KAM" for Kolmogorov, Arnold, Moser). We generalize the approach in De Roeck and Sch{\"u}tz~\cite{deroeck2017exponentially}, which considered the stability of non-interacting paramagnetic Hamiltonians comprised of single-site stabilizers, to topologically ordered quantum stabilizer codes. Notably, this approach provides much stronger control over the locality of rotations by tracking the growth of the qubit-support of rotated operators at each stage and  the number of stabilizers excited by rotated operators at each stage (in contrast to only tracking the spatial extent of the operator). We note that there are several alternative (though related) schemes for proving the stability of the gap for many-body systems \cite{yarotsky2006ground,albanese1989spectrum,frohlich2020lie}. 

Our results rigorously establish two cornerstones of zero-temperature many-body physics -- the stability of gapped phases and the idea of quasi-adiabatic continuation -- remain well-defined for large classes of models beyond the setting of geometrically local models in finite-dimensional Euclidean spaces. This opens the door to a plethora of mostly unexplored questions in quantum many-body physics, such as the classification of non-Euclidean and $k$-local (but geometrically non-local) phases, and a systematic study of their robust physical properties. This is particularly topical in light of stunning experimental developments giving us access to tunable interactions and geometries~\cite{Kollar2019,Periwal2021,Bluvstein2022,LukinLDPC}, opening up a wide new frontier for many-body physics.   
% (the latter of which we make some preliminary comments on in this paper)
.

\section{Models and Summary of Results}
\label{sec:summary}

Before presenting a formal statement of our results in Section~\ref{sec:theorems}, we discuss the types of Hamiltonians and perturbations considered in this work. In Section~\ref{sec:stabilizer_Hamiltonian}, we review a special class of Hamiltonians, stabilizer Hamiltonians, with frustration-free commuting Pauli terms. In Section~\ref{subsec:TQOPauli}, we discuss natural conditions on these Hamiltonians, including the LPDC condition giving a generalized notion of locality. Our main theorems assume these conditions. In Section~\ref{subsec:Zloc_condition}, we define an operator norm that penalizes nonlocality, and we discuss restrictions on the class of perturbations we consider in terms of this norm. 

Section~\ref{sec:theorems} gives the formal statement of our results, including our results on energy gaps and splittings (Theorem~\ref{thm: main}), relationships between perturbed and unperturbed ground states (Theorem~\ref{thm: locality}), and analyticity properties of the ground states (Theorem~\ref{thm: analyticity}).

\subsection{Stabilizer codes, checks, and graphs}\label{sec:stabilizer_Hamiltonian}

We first recall a few standard notions for stabilizer codes.

% \noindent\textbf{Hilbert space, Pauli group, stabilizer group}.
We are given a set $\Lambda$ of $|\Lambda|\equiv N$ sites, and to each site $x \in \Lambda$ we associate a qubit.
The Hilbert space of all qubits in $\Lambda$ is $\caH_\Lambda \coloneqq \bigotimes_{x\in \Lambda}\caH_x \cong (\mathbb{C}^2)^{\otimes N}$.
We denote by $\caP_\Lambda$ the \textit{Pauli group} supported on $\Lambda$.
Elements of $\caP_\Lambda$ are called \textit{Pauli strings}\footnote{The use of the word ``string" here is historical and does not imply a one dimensional geometry.}, which are operators on $\caH_\Lambda$.
A Pauli string $p \in \caP_\Lambda$ takes the form
\begin{equation}\begin{split}
    p = \bigotimes_{x \in \Lambda} p_x, \quad \text{ where } p_x \in \left\{
    \mathbb{1} = \begin{pmatrix}1 & 0 \\ 0 & 1\end{pmatrix},
    \mathsf{X} = \begin{pmatrix}0 & 1 \\ 1 & 0\end{pmatrix}, 
    \mathsf{Y} = \begin{pmatrix}0 & -i \\ i & 0\end{pmatrix},
    \mathsf{Z} = \begin{pmatrix}1 & 0 \\ 0 & -1\end{pmatrix} \right\}.
\end{split}\end{equation}
The \textit{qubit support of a Pauli string} $p$ is defined as 
\begin{equation}\begin{split}
    \supp(p) = \{x : p_x \neq \mathbb{1}\}.
\end{split}\end{equation}
A stabilizer code on $\Lambda$ is uniquely defined by an abelian subgroup $\caG \subseteq \caP_\Lambda$.
The \textit{code space} $P$ is the subspace of $\mathcal{H}_\Lambda$ defined as follows
\begin{equation}\begin{split}
    P \coloneqq \left\{\ket{\psi} \in \caH_\Lambda : \caG \ket{\psi} = \ket{\psi} \right\}.
\end{split}\end{equation}
We have that $K = \log_2 \dim P = N - \log_2 |\caG|$.
The number $K$ is often referred to as \textit{the number of logical qubits}.
The \textit{code distance} of $\caG$ is defined as 
\begin{equation}\begin{split}
    d = \min \{|{\rm supp}(p) | :  P p P \not \propto P \}.
\end{split}\end{equation}
We are interested in subsets of $\caG$ that generate $\caG$: $\langle \{ C_\alpha \}_{\alpha \in \caE} \rangle = \caG$ where $C_\alpha \in \{C_\alpha\}_{\alpha \in \caE} \subseteq \caG$, and we use the notation $\caE$ for an indexset of Pauli operators in a generating subset. 

We focus on \textit{stabilizer Hamiltonians} of the following commuting-projector form
\begin{equation}\label{eq:hamdef}
    H_0 = \sum_{\alpha \in \caE} E_\alpha, \quad \text{ where } E_\alpha = (1-C_\alpha)/2 
\end{equation}
Here, each $C_\alpha \in \{C_\alpha\}_{\alpha \in \caE} \subseteq \caG$ is called a \textit{check}, and we henceforth define $N_c \coloneqq |\caE|$. Given $\caG$, the choice of $\caE$ is equivalent to the choice of Hamiltonian $H_0$. We note in passing that different choices of $\caE$ for a given $\caG$ can have different stability properties, and the dependence on $\caE$ is reflected in the TQO-II condition given in Section~\ref{subsec:TQOPauli}.

In other words, the checks generate the stabilizer group.
Therefore,
\begin{equation}\begin{split}
    \prod_{\alpha \in \caE} G_\alpha = P, \quad \text{ where } G_\alpha = (1+C_\alpha)/2.
\end{split}\end{equation}
Here we use $P$ to refer to both the codespace and the projector onto this space. 
The ground state sector of $H_0$ corresponds to all checks being satisfied, which is identical to $P$.

We define the \textit{check support of a qubit $x$} as follows
\begin{equation}
\label{eq:check_supp_of_qubit}
    \supp_c(x) \coloneqq \{ \alpha: x \in \supp(C_\alpha) \}.
\end{equation}
We similarly define \textit{check support of a Pauli string $p$} as
\begin{equation}
\label{eq:check_supp_of_pauli}
    \supp_c(p) \coloneqq \bigcup_{x \in \supp(p)} \supp_c(x).
\end{equation}

% by the following condition: check $C_\alpha \in \supp_c(x)$ if and only if $x \in \supp(C_\alpha)$.

The checks define a graph structure on $\Lambda$.
We introduce an edge $(x,y)$ for $x,y\in \Lambda$ if $x \in \supp(C_\alpha)$ and $y\in \supp(C_\alpha)$ for some check $C_\alpha$.
The distance between any two qubits is the usual graph distance.
We similarly define a graph structure on $\caE$ with the checks as the vertices, and we introduce an edge $(C_\alpha, C_{\alpha'})$ if $x \in \supp(C_\alpha)$ and $x \in \supp(C_{\alpha'})$ for some qubit $x$.

Let $S \subseteq \caE$.
We define the qubit support of $S$ to be
\begin{equation}
\label{eq:def_qubit_support_of_S}
    \supp(S) \coloneqq \bigcup_{\alpha \in S} \supp(C_\alpha).
\end{equation}
We define the $\caG_S \subseteq \caG$ as the subgroup generated by checks in $S$, 
\begin{equation}
\label{eq:def_G_S}
    \caG_S \coloneqq \langle \{C_\alpha\}_{\alpha \in S} \rangle.
\end{equation}
Meanwhile, let $\caG(S) \subseteq \caG$ be the following subgroup,
\begin{equation}
    \caG(S) \coloneqq \{p \in \caG : \supp(p) \subseteq \supp(S) \}.
\end{equation}
% the subgroup of $\caG$ of Pauli strings $p$ where $\supp(p) \subseteq \bigcup_{\alpha \in S} \supp(C_\alpha)$.
% $\bigcup_{x \in \supp(p)} \supp_c(x) \subset S$.
We have hence $\caG_S \subseteq \caG(S)$ but the inclusion is in general strict.
For example, consider an annulus-shaped region of qubits $A$ in the 2D toric code, and take $S$ to be the set of all plaquette checks that overlap with $A$.
The group $\caG(S)$ contains a loop operator which is the product of all plaquettes enclosed by $A$, but this operator is not within $\caG_S$. 

In the remainder of this paper, we always use Greek letters ($\alpha,  \beta, \ldots$) to refer to checks, and Latin letters ($x, y, \ldots$) to refer to qubits.

% It follows that, in both graphs, the volume of ball of radius $r$ is upper bounded by $(w_c w_q)^r$.
% We define $\kappa \equiv \log(w_c w_q)$.

\subsection{Conditions on stable unperturbed stabilizer Hamiltonians}
\label{subsec:TQOPauli}

Our main result concerns the stability of  Hamiltonians $H_0$ associated with stabilizer codes, according to Eq.~\eqref{eq:hamdef}.   We are  interested in families of codes with growing $N$.
When discussing statistical mechanics in Euclidean spaces, the limit $N \to \infty$ is called the \textit{thermodynamic limit}, and this limit is essential for defining a sharp notion of phases and phase transitions. For codes on general graphs, it may not be possible to smoothly change $N$ (in units of 1) while maintaining local properties such as graph degree. Nevertheless, to associate codes to many-body phases, we have in mind families of models in which $N$ can be made arbitrarily large, even if the increase in $N$ does not happen in regular steps of 1.
% This is true, for instance, in well studied Tanner codes defined on families of expander graphs~\cite{sipserspielman}. 

% Thus,  we will require models which permit such a  
% and is a convenient one when associating the code Hamiltonians to many-body phases.

With these definitions, we are now ready to state the conditions on unperturbed stabilizer Hamiltonians $H_0$ which obey the stability theorems in Sec.~\ref{sec:theorems}.

\begin{itemize}[labelsep=0em, leftmargin=5.5em, labelwidth=5em, itemindent=0em, align=parleft, itemsep = 2em]
\item[\textbf{LDPC}]
We say the code family is LDPC if the following conditions hold on the qubit-weight $w_q$ and the check-weight $w_c$: 
\begin{equation}\begin{split}
    w_q \equiv&\, \max_{x \in \Lambda} |\supp_c(x)| = \Theta(1), \\
    w_c \equiv&\, \max_{\alpha \in \caE} |\supp(C_\alpha)| = \Theta(1).
\end{split}\end{equation}
Throughout, we will use $\Theta(1)$ to mean asymptotically upper and lower bounded by constants as $N \to \infty$. 
It follows from the definition of the graphs $\Lambda$ and $\caE$ that their graph degrees are both upper bounded by $w_c w_q$. 
The \textbf{LDPC} condition guarantees that the energy density of $H_0$ for all states in $\caH_\Lambda$ is upper bounded by a constant. If we additionally impose that each qubit is part of at least one check (so that no qubits are ``idle''), then the LDPC condition guarantees that $N_c \propto N$ corresponding to an extensive bandwidth $\propto N$ for the spectrum of $H_0$, analogous to conventional thermodynamics.

\item[\textbf{Growth of balls}]
Let $B_r(x)$ be the set of all sites within a distance $r$ of $x$ in the graph $\Lambda$.
Similarly, let $B_r(\alpha)$ be all checks within a distance $r$ of $\alpha$ in the graph $\caE$.
We consider Hamiltonians for which there exists $\kappa = \Theta(1)$ such that 
\begin{equation} \label{eq:def_kappa}
\begin{split}
    \forall x \in \Lambda, &\ | B_r({x})| \leq e^{\kappa r}; \\
    \forall \alpha \in \caE, &\ |B_r(\alpha)| \leq e^{\kappa r}.
\end{split}
\end{equation}
We note that \textbf{Growth of balls} follows from \textbf{LDPC} with any $\kappa \geq \log(w_c w_q+1)$. 

Note that on $D$-dimensional local Euclidean graphs, $B_r \propto r^D$; so our condition allows for much stronger growth of volumes, thereby encompassing a large family of graphs including the expanders. 

\item[\textbf{TQO-I}]
For every set of qubits $A \subseteq \Lambda$ with $|A| < d$ and for every Pauli string $p$ such that $\supp(p) \subseteq A$, we have
\begin{equation}\begin{split}
\forall g \in \caG: [p,g]=0 \qquad \Rightarrow \qquad p \in \caG.
\end{split}\end{equation}
In other words, there are no logicals supported on $A$.
For stabilizer codes, \textbf{TQO-I} follows directly from the definition of the code distance $d$. 
% The theorems in the next section are meaningful i.e. they give a stable gap and ground state degeneracy if $d\geq c \log(N)$ for $c>0$. 

% For every set $S \subseteq \caE$ with $|S| < d/w_c$ and for every Pauli string $p$ such that $\supp(p) \subseteq \bigcup_{\alpha \in S} \supp(C_\alpha)$, we have
% \begin{equation}\begin{split}
% \forall g \in \caG: [p,g]=0 \qquad \Rightarrow \qquad p \in \caG.
% \end{split}\end{equation}
% In other words, there are no logicals supported on $S$.
% For stabilizer codes, \textbf{TQO-I} follows from the definition of the code distance $d$.

\item[{\textbf{TQO-II}}] 
For every connected set $S$ in the check graph $\caE$, if $|S| < \widetilde{d} = \Theta(d)$, then we have 
\begin{equation}\label{eq:TQO2}\begin{split}
\caG(S) \subseteq \caG_{B_{\ell|S|}(S)}
\end{split}\end{equation}
where $\ell = \Theta(1)$, and
$B_{\ell|S|}(S)=\{\alpha : \mathrm{dist}(\alpha,S)\leq \ell |S|\}$. 
% The theorems in the next section are meaningful i.e. they give a stable gap and ground state degeneracy if $\widetilde{d}\geq c \log(N)$ for $c>0$. 
% (little bit more difficult to see through it, but it is analogous to the condition in BHM)
\end{itemize}
We define $\dstar := \min(\frac{d}{w_c}, \widetilde{d}) = \Theta(d)$. The theorems in Section~\ref{sec:theorems} are meaningful i.e. they give a stable gap and ground state degeneracy if $\dstar \geq c \log(N)$ for $c>0$. 

For local stabilizer codes on Euclidean geometries, our \textbf{TQO-II} condition is closely related
% \sout{equivalent} 
to the TQO-II condition of BHM in their Lemma 2.1 \cite{Bravyi_2010}. Therefore, we expect our conditions to hold for a wide class of codes in Euclidean spaces. In App.~\ref{app:TQO} we argue that this condition indeed holds for various important examples of non-Euclidean qLDPC codes, in particular to known constructions of good quantum LDPC codes and to hypergraph products of good classical LDPC codes. We further note that the condition \textbf{TQO-II} here is a very mild one; we are able to use such a mild condition because of how we track our rotated perturbations throughout the proof. We note that the hyperbolic toric code satisfies both \textbf{TQO-I} and \textbf{TQO-II} with $\ell = \Theta(1)$ and $\dstar = \Theta(\log(n))$, and so we are able to prove stability of the hyperbolic toric code.

In Figure~\ref{fig: graph_props}, we use the familiar example of the vanilla toric code to illustrate some of our notation, particularly the important quantities $\wc$ and $\wq$ that encode the LDPC nature of the Hamiltonian.
\begin{figure}[h]
\centering
\includegraphics[width=.9\linewidth]{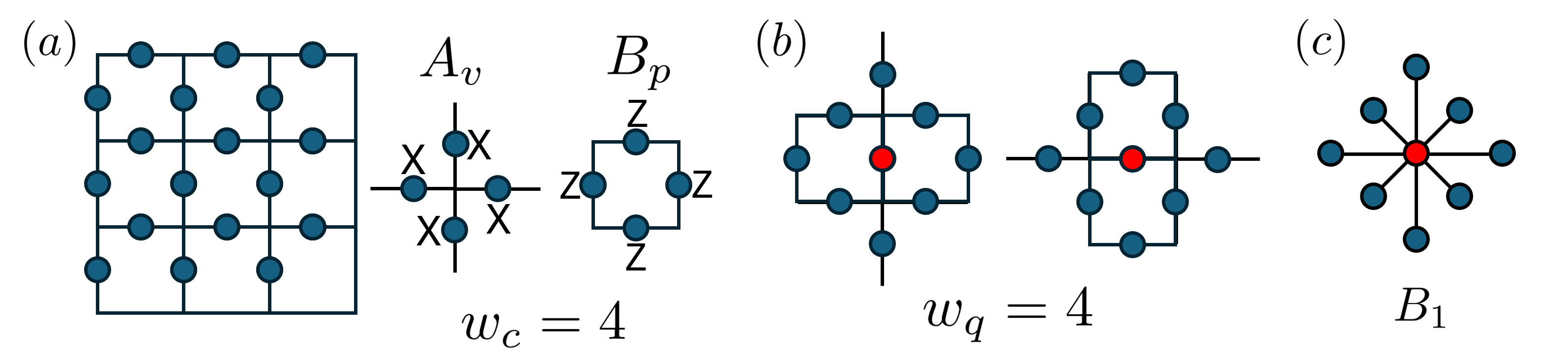}
 \caption{Examples of a few geometric quantities illustrated via the vanilla toric code. (a) The toric code's stabilizer group $\mathcal{G}$ is generated by the set $\{A_v\}_v \cup \{B_p\}_p$ of $\mathsf{X}$-stars and $\mathsf{Z}$-plaquettes. In this figure, we consider $\caE$ (see equation~\ref{eq:hamdef} and surrounding discussion) such that $\{C_\alpha\}_{\alpha \in \caE} = \{A_v\}_v \cup \{B_p\}_p$. This is the usual set of checks of the toric code. For this $\caE$, the maximum qubit support of any check is $\qc=4$, since each check contains exactly four qubits. (b) With this choice of $\caE$, the maximum check support of the qubits is $\wq=4$, since a given qubit (red) touches exactly four stabilizers (two stars and two plaquettes). (c) The checks define a graph structure on the qubits, where two qubits are connected by an edge if they are both part of the same check. Pictured is a ball of radius $1$ surrounding the red qubit (at site, say, $x$), which we denote $B_1(x)$.}
\label{fig: graph_props}
\end{figure}

\subsection{Conditions on allowed perturbations}
\label{subsec:Zloc_condition}

We will consider the stability of Hamiltonians $H_0$ obeying the conditions in Section~\ref{subsec:TQOPauli} to extensive perturbations $Z^{(0)}$ whose energy bandwidth $\propto N$. In the simplest case, $Z^{(0)}$ is a sum of $\Theta(N)$ strictly graph-local operators whose support is bounded within balls of finite size on the graph. However, we allow for more general quasi-local perturbations subject to the norm below; informally, perturbations that are local enough and small enough will suffice.

To state the class of perturbations we can handle, we will use a norm that penalizes nonlocal operators. Given the (unique) decomposition of an operator $O = \sum_{p} c_p p$ into the Pauli basis, we define an \textbf{intensive Pauli norm} $||| \cdot |||_{\mu}$:
\begin{equation}\label{eq:paulinorm}
    |||O|||_{\mu} = \sup_{x} \sum_{p: x\in \supp(p)} |c_p| e^{\mu |\mqs(p)|} 
\end{equation}
where $\mqs(p)$ is a connected set of qubits that covers the support of $p$, $\supp(p) \subset \mqs(p)$, with $|\mqs(p)| =\min_{\mathcal{A}: \, \supp(p) \subset \mathcal{A},\,\mathcal{A} \text{ connected}} |\mathcal{A}|$.

The supremum over $x$, and more particularly the restriction that $\supp(p)$ must contain $x$, means that this norm can and will be \textit{intensive} for the \textit{extensive} perturbations we consider\footnote{As an example, if $O$ is a field $h \sum_{x} \mathsf{X}_x$, then $|||O|||_{\mu} = h e^{\mu}$. If $O$ consists of a single Pauli string with support on all sites, such as $O = \prod_{x} \mathsf{X}_x$, then $|||O|||_{\mu} = e^{\mu N}$. Having a small $|||O|||_{\mu}$ requires that $O$ cannot contain too much weight on Pauli strings with large support, and requiring small $|||O|||_{\mu}$ for larger values of $\mu$ is more stringent.}.

To have control of the perturbation $Z^{(0)}$ for all system sizes $N$, we require that 
\begin{equation}\label{eq:mu0_norm_Z0}
    |||Z^{(0)}|||_{\widetilde{\mu}_{0}} = \Theta(1)
\end{equation}
for some $\widetilde{\mu}_0 > \mu_{\#}$, where 
\begin{equation}
    \mu_{\#} := w_q \mu_* + (w_q \log(2) + \kappa + \log(2 w_c))
\end{equation}
and, letting $c$ be the maximal number such that $d_\ast \geq c \log(N)$ for all system sizes $N$,
\begin{equation}\label{eq:def_mu_star}
    \mu_{*} := \max(\kappa \ell + \log(4)(\wc+1), c^{-1}).
\end{equation}
Recall that $\kappa$ is the geometric constant appearing in \eqref{eq:def_kappa}, and $\ell$ is the constant appearing in \textbf{TQO-II} and $\dstar$ appears underneath \textbf{TQO-II}, c.f. Sec.~\ref{subsec:TQOPauli}.
For all $\widetilde{\mu}_0 > \mu_{\#}$, there is a threshold strength $\widetilde{\epsilon} = \Theta(1)$ that is a function of $\widetilde{\mu}_0, \mu_{\#}$ and the geometrical constants $\wq, \wc, \kappa$, but not the system size $N$, c.f. Section~\ref{sec:running_couplings}. We further require 
\begin{equation}\label{eq:mu0_norm_Z0_eps0}
    \widetilde{\epsilon}_0 := |||Z^{(0)}|||_{\widetilde{\mu}_{0}} \leq \widetilde{\epsilon}.
\end{equation}

Roughly speaking, $\widetilde{\mu}_0$ controls the locality of $Z^{(0)}$, i.e. how fast the perturbation terms decay with increasing operator weight, while $\widetilde{\epsilon}_0$ controls the overall local strength of the perturbation. We require $\widetilde{\mu}_0 > \mu_{\#}$ to ensure that the perturbation terms are decaying sufficiently rapidly. If the perturbation $Z^{(0)}$ satisfies $\eqref{eq:mu0_norm_Z0}$, then multiplying $Z^{(0)}$ by a sufficiently small, system-size independent constant will make it satisfy \eqref{eq:mu0_norm_Z0_eps0}.

Many perturbations satisfy $\eqref{eq:mu0_norm_Z0}$. 
Extensive and strictly graph-local $Z^{(0)}$ (i.e. a sum of terms supported within finite-sized balls and with uniformly bounded operator norms) satisfy \eqref{eq:mu0_norm_Z0}. More generally, quasilocal perturbations with sufficiently short ``tails" will satisfy \eqref{eq:mu0_norm_Z0}. Note that the definition of locality relevant here is locality on the graphs $\Lambda$ and  $\caE$, defined directly in terms of the unperturbed Hamiltonian (which defines a notion of distance between qubits). Thus, we do not allow for perturbations by some other arbitrary LDPC code for example, but only those whose structure is compatible with that of $H_0$. One context where this arises naturally is if $H_0$ is defined in terms of some fixed underlying graph (such as many of the known constructions of good qLDPC codes, such as quantum Tanner codes and their relatives~\cite{leverrier2022quantum,panteleev2022asymptotically,dinur2022good}) and the perturbations are required to also be local on the same graph.

%\YL{A comment on $\mu_\ast$: we need $\mu$ to be sufficiently large in Lemma~\ref{lem: bound on ghosts}, Lemma~\ref{lem: explicit expression a}, Proposition~\ref{lem: locality of a}. We also need $\mu \geq \kappa\ell + \log(2)$ from Section~\ref{sec:proof_rel_bounded}. For these purposes, it is sufficient to take $\mu^\ast = \kappa \ell + \log(64) \geq \log(w_c w_q+1)\cdot \ell + \log(64)$.}

% \wdr{I think we should put somewhere, best here, the following remark:}

\subsection{Statement of main theorems}\label{sec:theorems}

In what follows, we will use the generic notation $c,C$ for constants that can depend on the $\Theta(1)$-model parameters, i.e.\ on $w_q,w_c,\kappa, \ell$. We use $c$ whenever we want to stress that $c>0$ and $C$ whenever we want to stress that $C < \infty$. 
The precise value of $c,C$ can change from line to line. Whenever we allow constants to depend on some additional parameters, we indicate this in the notation.   We are now ready to state our main results
% \wdr{end of remark}

\begin{theorem}\label{thm: main}
Suppose $H_0$ obeys the conditions of Section \ref{subsec:TQOPauli} and $Z^{(0)}$ obeys the conditions of Section \ref{subsec:Zloc_condition}.  Constants $\widetilde{\epsilon}_0, \widetilde{\epsilon}, \widetilde{\mu}_0, \mu_{\#}$ are defined in Section \ref{subsec:Zloc_condition} and obey $\widetilde{\mu}_0 > \mu_{\#}$ and $\widetilde{\epsilon}_0 \leq \widetilde{\epsilon}$. $N_c$ is defined in Section \ref{sec:stabilizer_Hamiltonian} and $\dstar$ is defined in Section \ref{subsec:TQOPauli}. 

Then, there are constants $b$,  $\mu_{\infty}>\mu_*$,  and $\epsilon_0 \leq \widetilde{\epsilon}_0$ such that the spectrum of the perturbed Hamiltonian
\begin{equation}
H=H_0+Z^{(0)} -b
\end{equation}
is contained in a union of intervals $\bigcup_{k\geq 0} I_k$, where $k \in \{ 0, 1, 2, \cdots \}$ runs over the spectrum of $H_0$ and 
\begin{equation}
    I_k =[k-(kC' \epsilon_0 + \delta), k+(kC' \epsilon_0 + \delta)]
\end{equation}
with $\delta = C{N_c} \epsilon_0 e^{-\mu_\infty \dstar}$.
\end{theorem}
Each interval is centered on $k$ and has width $2(kC' \epsilon_0 + \delta)$. The nearly-degenerate ground states correspond to the interval $I_{k=0}$ with splitting at most $2 \delta$, which exponentially decays in $\dstar$. 
Thus, for small $\epsilon_0$, the spectral gap above $I_0$ remains $\Theta(1)$, and the splitting $\delta$ of the eigenvalues in $I_0$ rapidly decays to zero as $C{N_c} \epsilon_0 e^{-\mu_\infty \dstar}$.

Let us denote the spectral projector corresponding to $I_0$ by $\widetilde P$ and the spectral projector corresponding to the ground state sector of $H_0$ by $P$.

\begin{theorem}\label{thm: locality}
Under the same assumptions as in Theorem \ref{thm: main} above, we have
\begin{equation}
\widetilde P = U P U^{-1}
\end{equation}
for a unitary $U$ that is locality-preserving and locally close to identity, in the following sense: For any operator $O$ supported in a connected set $Y\in \Lambda$, we can write
\begin{equation}\label{eq: definition locality observables under u}
U^{-1}OU= O+\sum_{r=1}^{\mathrm{diam}(\Lambda)} O_r +O_{\mathrm{bg}}
\end{equation}
where $\mathrm{diam}(\Lambda)$ is the diameter of the graph $\Lambda$, and
\begin{enumerate}
    \item  $O_r$ is supported in $\{x: \mathrm{dist}(Y,x) \leq r\}$.
    \item  $O_r$ satisfies 
    \begin{equation} 
    ||O_r|| \leq \epsilon_0 C(\mu_0,Y) e^{-cr}||O||
    \end{equation}
    for some $c>0$ and constant $C(\mu_0,Y)$ that can depend on $\mu_0$ and $Y$.
    \item  $O_{\mathrm{bg}}$ (with $\mathrm{bg}$ standing for `background') is bounded as 
    $|| O_{\mathrm{bg}} || \leq     C{N_c} \epsilon_0 e^{-\mu_\infty \dstar} ||O||$.
\end{enumerate}
The same statement holds true for $UOU^{-1}$. 
\end{theorem}
 
If the number ${N_c} \epsilon_0 e^{-\mu_\infty \dstar}$ is upper-bounded by $e^{-c\mathrm{diam(\Lambda)}}$ ,
for some $c>0$, then the background term $O_{\mathrm{bg}}$ can be omitted in the statement of the theorem, at the cost of increasing the constant $C(\mu_0,Y)$. Indeed, in that case,   $||O_{\mathrm{bg}}||$ is small enough to be included in $O_r$ for $r=\mathrm{diam(\Lambda)}$. 
This is the case for all examples that we are aware of.

Finally, we state an analyticity property of the spectral projector $\widetilde{P}$.  
\begin{theorem}\label{thm: analyticity}
Suppose $H_0$ and $Z^{(0)}$ satisfy the conditions of Theorem \ref{thm: main}.
If the perturbation $Z^{(0)}$ is analytic in some parameter $\gamma$ on a complex domain $D \subset \mathbb{C}$, such that \eqref{eq:mu0_norm_Z0_eps0} is satisfied uniformly in $\gamma\in D$, then
the ground state space projector $\widetilde P$ can be continued analytically to a function of $\gamma \in D$. It is of the form $\widetilde{P}=UPU^{-1}$ with the invertible operator $U$ such that \eqref{eq: definition locality observables under u} holds and the items 1,2,3 of Theorem \ref{thm: locality} still hold (upon redefining the constants $C(\mu_0,Y)$ and $C$), uniformly in $\gamma\in D$.
\end{theorem}
This theorem is a simple consequence of the fact that our results continue to hold when the perturbation is no longer Hermitian. A corresponding statement about the spectrum of the non-Hermitian perturbed operator $H$ can be made as well (see Section \ref{sec:effects_of_error_on_spectrum} for details).

\section{General Setup}
\label{sec:setup}

% \vedika{WDR: please read this paragraph. }
In Section~\ref{sec:KAM_scheme}, we will describe the iterative KAM procedure for iteratively performing unitary transformations to eliminate non-frustration-free terms.
A key book-keeping tool of this procedure, which allows us to provide much stronger control on the locality of rotated operators, is that the procedure keeps track of the number of stabilizers excited by operators at each stage of the rotation, even if the stabilizer is de-excited at a later stage.
Our procedure is an adaptation of Ref. \cite{deroeck2017exponentially}, which considers the stability of paramagnetic Hamiltonians comprised of sums of single-site stabilizers, of the form $H_0 =\sum_x \mathsf{X}_x$. In the paramagnetic case, a simple basis change on every site allows all operators to be represented by their ``raising" and ``lowering'' action on the stabilizers, which helps track the number of excitations.
When the stabilizers are multi-site, as they will be for general quantum codes, we need more formalism to represent operators in a way that allows us to track their action on the stabilizers, which is described in this section.

\subsection{Operators and ``Words"}
\label{sec:words}
% We are given a graph $\Lambda$ and a set of checks or checks $\mathcal E$. The ground state sector corresponds to all checks being satisfied. 
% In the case of the Ising model, this set is simply the set of edges of $\Lambda$.
% We abbreviate the stabilizers as
% $$
% C_e
% $$

% \vedika{decide notation for stabilier generators, stabilizers, P, G, E etc}

In Sec.~\ref{sec:summary}, 
we defined projectors on whether a stabilizer is satisfied or not, namely
\begin{equation}\begin{split}
 G_\alpha = (1+C_\alpha)/2, \qquad   E_\alpha = (1-C_\alpha )/2,\qquad G_\alpha +E_\alpha =1
\end{split}
\label{eq:proj}
\end{equation}
The notation $G/E$ refers to ground state/excited. 
Recall that our unperturbed Hamiltonian is 
\begin{equation}\begin{split}
H_0=\sum_{\alpha \in \mathcal E} E_\alpha.
\end{split}\end{equation}

We now set up a formalism to deal with the operators that will appear in the perturbation and those generated in the iterative KAM procedure.
We will consider quadruples (``words'') $\mathbf S= (S_+,S_-,S_e,S_g)$ of disjoint sets of checks, and we will write 
\begin{equation}\begin{split}
S \coloneqq
S_+ \cup S_- \cup S_e \cup S_g.
\label{eq:Ssup}
\end{split}\end{equation}
Here $\{e, g, +, -\}$ should be read as ``\{excited, ground, raising, lowering\}''. 
We use words to keep track of the action of each operator on the checks, and group operators into different classes depending on their action on the checks.  

We define a class of operators $\mathcal X_\bfs$. 
An operator $X \in \caX_\bfs$ iff there exists an operator $\widetilde X$ such that
\begin{equation}\begin{split}
\label{def:X_S}
    X= \left(\prod_{\alpha \in S} \mathrm{left}_\alpha(\bfs) \right) \widetilde X  \left(\prod_{\alpha \in S} \mathrm{right}_\alpha(\bfs) \right)
\end{split}\end{equation}
where \begin{equation}\begin{split}
     \mathrm{left}_\alpha(\bfs) &= \begin{cases}
         G_\alpha &   \text{if} \,  \alpha \in S_{-} \cup S_g \\
         E_\alpha &   \text{if} \,  \alpha \in S_{+} \cup S_e
     \end{cases}    
    \end{split}\end{equation}
    \begin{equation}\begin{split} \mathrm{right}_\alpha(\bfs) &= \begin{cases}
         G_\alpha &   \text{if} \,  \alpha \in S_{+} \cup S_g \\
         E_\alpha &   \text{if} \,  \alpha \in S_{-} \cup S_e
     \end{cases}    
\end{split}\end{equation}
and we require that
% \begin{equation}\begin{split}
%     \bigcup_{x \in \supp(\widetilde{X})} \supp_c(x) \subseteq S.
% \end{split}\end{equation}
\begin{equation}
\label{eq:supp_of_x_tilde}
    % \supp(\widetilde{X}) \subseteq \supp(S).
    \bigcup_{x \in \supp(\widetilde{X})} \supp_c(x) \subseteq S
    \quad \Leftrightarrow \quad
    \Lambda \setminus \supp(\widetilde{X}) \supseteq \supp(\caE \setminus S).
\end{equation}
We note that $\widetilde{X}$ can be expanded as a sum of Pauli operators, and by $\supp(\widetilde{X})$ we mean the union of the qubit support of all non-zero Pauli operators in the expansion of $\widetilde{X}$;  $\supp(S)$ is defined in \eqref{eq:def_qubit_support_of_S}.
% , and $\supp_c(x)$ is the \textit{check support} of qubit $x$, as defined in \eqref{eq:check_supp_of_qubit}.
Here, we require that all checks that overlap with $\widetilde{X}$ are also in $S$.

In what follows, we will use $X_\bfs$ to denote an element of $\caX_\bfs$.

The above definition also shows how to treat perturbations built out sums of local terms. Each local perturbation term should be thought of as $\widetilde X$. One decomposes them in operators corresponding to classes $\caX_{\bfs}$. 
We note that the decomposition of an given operator is not unique.
For strictly local perturbations (i.e. those living in finite-sized balls), it is always possible to expand them into \textit{connected} words i.e. words where $S$ is connected in the graph $\caE$, such that the expansions have a finite norm (to be defined below in \eqref{eq:def_mu_norm}).

Let us provide an example. In a 1D Ising model, the stabilizers are $C_{x+\frac{1}{2}}=\mathsf{Z}_x\mathsf{Z}_{x+1}$ where $\mathsf{Z}_x$ is the Pauli $Z$ operator on site $x$. To consider a local perturbation such as $\mathsf{X}_x$, we decompose it into different words which describe the raising/lowering of $\mathsf{X}_x$ on its adjacent stabilizers, which in turn is conditioned on the state of the adjacent stabilizers prior to acting with $\mathsf{X}_x$. So, we write, 
\begin{align}
\mathsf{X}_x &= (G_{x+\frac{1}{2}} + E_{x+\frac{1}{2}}) (G_{x-\frac{1}{2}} + E_{x-\frac{1}{2}}) \mathsf{X}_x  (G_{x-\frac{1}{2}} + E_{x-\frac{1}{2}}) (G_{x+\frac{1}{2}} + E_{x+\frac{1}{2}})\nonumber\\
&= E_{x+\frac{1}{2}}E_{x-\frac{1}{2}}\mathsf{X}_x G_{x-\frac{1}{2}}G_{x+\frac{1}{2}} + E_{x+\frac{1}{2}}G_{x-\frac{1}{2}}\mathsf{X}_x E_{x-\frac{1}{2}}G_{x+\frac{1}{2}}\nonumber\\ 
&+ G_{x+\frac{1}{2}}E_{x-\frac{1}{2}}\mathsf{X}_x G_{x-\frac{1}{2}}E_{x+\frac{1}{2}} + G_{x+\frac{1}{2}}G_{x-\frac{1}{2}}\mathsf{X}_x E_{x-\frac{1}{2}}E_{x+\frac{1}{2}}
\label{eq:decomposeX}
\end{align}
where we used \eqref{eq:proj}, the fact that $\mathsf{X}_x$ flips adjacent stabilizers from ground to excited and vice versa, and the fact that $G_e E_e=0$.  The four terms in the second line of \eqref{eq:decomposeX} belong, in order, to the classes
\begin{equation} \begin{split}
    E_{x+\frac{1}{2}}E_{x-\frac{1}{2}}\mathsf{X}_x G_{x-\frac{1}{2}}G_{x+\frac{1}{2}} &\in \mathcal{X}_{\mathbf{S} = \{ S_+ = \{C_{x+\frac{1}{2}}, C_{x-\frac{1}{2}}\}, \emptyset, \emptyset, \emptyset\}} \\
    E_{x+\frac{1}{2}}G_{x-\frac{1}{2}}\mathsf{X}_x E_{x-\frac{1}{2}}G_{x+\frac{1}{2}} &\in \mathcal{X}_{\mathbf{S} = \{ S_+ = \{C_{x+\frac{1}{2}}\}, S_-=\{C_{x-\frac{1}{2}}\}, \emptyset, \emptyset, \}}\\
    G_{x+\frac{1}{2}}E_{x-\frac{1}{2}}\mathsf{X}_x G_{x-\frac{1}{2}}E_{x+\frac{1}{2}} &\in \mathcal{X}_{\mathbf{S} = \{ S_+ = \{C_{x-\frac{1}{2}}\}, S_-=\{C_{x+\frac{1}{2}}\}, \emptyset, \emptyset, \}}\\
    G_{x+\frac{1}{2}}G_{x-\frac{1}{2}}\mathsf{X}_x E_{x-\frac{1}{2}}E_{x+\frac{1}{2}} &\in \mathcal{X}_{\mathbf{S} = \{ \emptyset, S_- = \{C_{x+\frac{1}{2}}, C_{x-\frac{1}{2}}\}, \emptyset, \emptyset\}}
\end{split} \end{equation}
where $\emptyset$ is the empty set. We will always decompose operators into classes defined by words in this way.

We now characterize the multiplication of operators  $X_{\bfs'} \in \caX_{\bfs'}$ and $X_{\bfs} \in \caX_\bfs$.
The multiplication is defined as the multiplication of operators.
By \eqref{def:X_S}, the multiplication might result in zero, due to incompatible projectors.
Whenever the product is nonzero, we have that 
\begin{equation}\begin{split}
X_{\bfs'} X_{\bfs}  \in  \caX_{\bfs' \bfs},
\end{split}\end{equation}
where the multiplication (or concatenation) of words $\bfs' \bfs$ is defined as follows.
As a set, we have $S'S = S' \cup S$.
For each $\alpha \in S'S$, we assign it to exactly one of $(S'S)_{g,e,+,-}$.
The rule of assignment is summarized in the following ``multiplication table'' (Table~\ref{tab:word_multiplication_table}).
An entry $0$ signifies when $X_{\bfs'} X_{\bfs} = 0$, in which case we say $\bfs' \bfs$ is undefined.
% Denoting  $\overline{S} \coloneqq \caE \setminus S$ and $\overline{S'} \coloneqq \caE \setminus S'$ as the complements of the sets $S, S'$, we have
\begin{table}[h]
\small
\centering
\renewcommand{\arraystretch}{1.5}
\begin{tabular}{ |p{1cm}||p{1.5cm}|p{1.5cm}|p{1.5cm}|p{1.5cm}|p{1.5cm}|}
 \hline
 $\alpha$ in & $\caE \setminus S$ & $S_g$ & $S_e$ & $S_+$ & $S_-$\\
 \hline \hline
 $\caE \setminus S'$ & $\caE \setminus (S'S)$ & $(S'S)_g$ & $(S'S)_e$ & $(S'S)_+$ & $(S'S)_-$\\
 \hline
 $S'_g$& $(S'S)_g$ & $(S'S)_g$ & 0 & 0 & $(S'S)_-$\\
 \hline
 $S'_e$& $(S'S)_e$ & 0 & $(S'S)_e$ & $(S'S)_+$ & 0 \\
 \hline
 $S'_+$& $(S'S)_+$ & $(S'S)_+$ & 0 & 0 & $(S'S)_e$\\
 \hline
 $S'_-$& $(S'S)_-$ & 0 &  $(S'S)_-$ & $(S'S)_g$ & 0\\
 \hline
\end{tabular}
\renewcommand{\arraystretch}{1}
\caption{\small Multiplication table of words.}
\label{tab:word_multiplication_table}
\end{table}

Finally, we note
\begin{enumerate}
\item  In general it holds that 
\begin{equation}\begin{split}
[X_{\bfs},X_{\bfs'}] =0 \qquad \text{whenever $S \cap S'=\emptyset$}.
\end{split}\end{equation}
\item Given a quadruple $\bfs_i$, we write $(S_{i})_{+,-,g,e}$ for its components. If the subscript is omitted, i.e. $S_i$, then we mean the union of components as in \eqref{eq:Ssup}. 
\item If two words $\bfs,\bfs'$ have $S\cap S'\neq \emptyset$, and $\bfs''=\bfs'\bfs$ is defined, then 
\begin{equation} \label{eq: enhanced subadditivity}
|S''| +|S \cap S'| \leq  |S|+|S'|
\end{equation}    
This relation will be used throughout. 

\end{enumerate} 

\subsection{Operator-collections and their norms} \label{sec: collections and norms}

Since we will consider extensive operators like Hamiltonians, it is convenient to define norms that are insensitive to the total volume (sometimes also called intensive or local norms) and that encode spatial decay properties of the local Hamiltonian terms. We use several different norms in this paper, including (i) the \textbf{intensive Pauli norm} $||| \cdot |||_{\tilde \mu}$ defined in \eqref{eq:paulinorm}; (ii) the \textbf{intensive word norm} $|| \cdot ||_{\mu}$ defined below in \eqref{eq:def_mu_norm}; and (iii) the usual non-intensive \textbf{operator norm} $|| \cdot ||$, defined to be the largest singular value of the operator.

We will consider \textit{operator-collections} $O=(O_{\bfs})_{\bfs}$, indexed by words $\bfs$ and such that $O_\bfs \in \caX_\bfs$. We define the \textbf{intensive word norm} $|| \cdot ||_{\mu}$ of an operator-collection: 
\begin{equation} \label{eq:def_mu_norm}
||O||_{\mu}=\sup_{\alpha} \sum_{\mathbf{S}: \alpha \in S}  ||O_{\mathbf S}|| \cdot e^{\mu |S|}, 
\end{equation}
for some parameter $\mu \geq 0$. Note that this is a different norm (distinguished by $||_\mu$ rather than $|||_\mu$) from the intensive Pauli norm Eq.~\eqref{eq:paulinorm}.

To such an operator-collection, we associate an \textit{extensive operator}
\begin{equation} \label{eq: collection to operator}
    O=\sum_\bfs O_\bfs,
\end{equation}
which is denoted by the same symbol.
Note that 
\begin{equation} \label{eq: collection to operator bound}
||O|| \leq  N_c \cdot ||O||_{0}
\end{equation}
where $N_c$ is the number of terms in the Hamiltonian (the number of checks); $N_c \coloneqq |\caE|$.
% Because of the geometry we have in mind, we need $\mu>\kappa$ so that $||O||_{\mu}<\infty$ implies that
% $|| O|| $ would be bounded by 

It is important to realize that an operator-collection determines the extensive operator via the above sum \eqref{eq: collection to operator}, but the converse does not hold. 
Therefore, whenever we add extensive operators, or take their commutator, the associated operator-collection should be explicitly defined, as we do now.
Let us have a pair of operator-collections 
$(O_\bfs)_{\bfs}, (O'_\bfs)_{\bfs}$,
\begin{enumerate}
    \item We define the interaction corresponding to the sum $O+O'$ by point-wise addition, i.e.
    \begin{equation}\begin{split}
   (O+O')_\bfs=  (O_\bfs+O'_\bfs).
    \end{split}\end{equation}
    \item We define the interaction corresponding to the commutator as 
\begin{equation} \label{def: commutator} \begin{split}
   ([O,O'])_\bfs= \sum_{\substack{\bfs_1,\bfs_2 \\ \bfs=\bfs_1\bfs_2, S_1 \cap S_2\neq \emptyset}} (O_{\bfs_1} O'_{\bfs_2} -O'_{\bfs_1} O_{\bfs_2}) 
    \end{split}\end{equation}
    where we define the summand to be zero whenever $\bfs_1\bfs_2$ is not well-defined.
\end{enumerate}

\subsection{Initial choice of operator-collections}\label{sec:initial_choice_operator_collection}

We write the operator $H_0 = \sum_{\alpha} E_{\alpha}$ as an operator-collection through the choice
\begin{equation}
    (H_0)_{\mathbf{S}} = \begin{cases} E_\alpha & (S_+ = \emptyset,S_-= \emptyset,S_e=\{\alpha\},S_g= \emptyset) \\
0 & \text{otherwise}
\end{cases}
\end{equation}

For the perturbation $Z^{(0)}$, we will first write it as a sum of Paulis, $Z^{(0)} = \sum_{p} c_p p$. We break each given $p$ into an operator-collection $p = \sum_\mathbf{S} p_{\mathbf{S}}$. To construct this decomposition for a given $p$, we first find its qubit support $\supp(p)$. We then find a minimal connected set $\mqs(p)$ of qubits such that $\supp(p) \subset \mqs(p)$. The $\mathbf{S}$ that are included in the sum have $S = \ext(p) := \cup_{x \in \mqs(p)} \supp_c(x)$. Each $p_\mathbf{S}$ is constructed by directly sandwiching $p$ by the $\mathrm{left}(\mathbf{S})$ and $\mathrm{right}(\mathbf{S})$ operators in Section~\ref{sec:words}. Finally, we choose $(Z^{(0)})_{\mathbf{S}} = \sum_{p: \ext(p) = S} c_p p_{\mathbf{S}}$.

This construction ensures that at the initial step, all operator-collections are defined so that if $(H_0)_{\bfs}$ and/or $(Z^{(0)})_{\bfs}$ are nonzero, then $S$ is a connected set relative to the graph of stabilizers.

\subsection{Relationships between the Pauli norm and the word norm}
\label{subsec:norm_rel}
We have introduced two norms; $|||O|||_{\mu}$, which is unique, and $||O||_{\mu}$, which depends on the choice of operator-collection used to describe $O$. While most of our work is in terms of $||O||_{\mu}$, we sometimes need to change between $|||O|||_{\mu}$ and $||O||_{\mu}$ at the cost of a small increase in $\mu$.

\begin{proposition}\label{prop:tripleboundsdouble}
Let $O$ be an operator-collection constructed analogously to that of $Z^{(0)}$. Then
\begin{equation}
||O||_{\mu} \leq |||O|||_{w_q\mu + (w_q\log(2) + \kappa + \log(2 w_c))}
\end{equation}
\end{proposition}

\begin{proposition}\label{prop:doubleboundstriple}
Let $O$ be an operator whose choice of operator-collection has all $S$ connected. Then
\begin{equation}
|||O|||_{\mu} \leq ||O||_{w_c\mu + \log(4 w_q)}
\end{equation}
\end{proposition}
We prove these propositions in Appendix~\ref{sec:proof of norm relationship}.

\subsection{TQO conditions restated in terms of words}
\label{subsec:TQOword}
% \vedika{should we just call these TQO-I and TQO-II without the primes, since we don't refer to primes later? Should be restated in terms of d*}

We may restate the TQO conditions from Sec.~\ref{subsec:TQOPauli} in terms of words.
These are equivalent to \textbf{TQO-I} and \textbf{TQO-II}.
\begin{itemize}[labelsep=0em, leftmargin=5.5em, labelwidth=5em, itemindent=0em, align=parleft, itemsep = 2em]
\item[{\textbf{TQO-I}}] 
Let the word $\mathbf{S}$ satisfy $|S| < d/w_c$. If $O_\mathbf{S} \in \mathcal{X}_\mathbf{S}$ satisfies $[O_\mathbf{S},g]=0$ for all $g \in \caG$, then $O_\mathbf{S} \in \mathcal{A}_\caG$, where $\mathcal{A}_\caG$ is the algebra generated by $\caG$ with complex coefficients.
In other words, there are no logicals supported on $S$.

\item[{\textbf{TQO-II}}] 
Let the word $\mathbf{S}$ be such that $S$ is connected in the check graph and $|S| < \widetilde{d}$. 
If $O_\mathbf{S}\in \mathcal{A}_\caG$, then $O_\mathbf{S}\in \mathcal{A}_{\caG_{B_\ell(\mathbf{S})}}$, where we recall
\begin{equation}\begin{split}
    % B_l(S)=\{x: \mathrm{dist}(x,S)\leq l|S|\} \\
    B_{\ell|S|}(S)=\{\alpha : \mathrm{dist}(\alpha,S)\leq \ell |S|\}.
\end{split}\end{equation}

% As we will see below, all words that we will need to consider are connected, by construction.

\end{itemize}
Equivalently to Section~\ref{subsec:TQOPauli}, we define $\dstar \coloneqq \min(d/w_c, \widetilde{d})$. As noted there, the theorems in Section~\ref{sec:theorems} are meaningful i.e. they give a stable gap and ground state degeneracy if $\dstar \geq c \log(N)$ for $c>0$. 

% \newpage

\section{Scheme of the proof}
\label{sec:KAM_scheme}

\subsection{Outline of the proof}

In this section we describe the iterative KAM procedure.
We first give an informal outline of the proof strategy.

As we detail in Sec.~\ref{sec: splitting of ham}, at each stage of the procedure, we split the  Hamiltonian into four parts: $H_0$, frustration-free terms ($D^{(n)}$ and $M^{(n)}$) that share the ground state space of $H_0$, non-frustration-free terms ($V^{(n)}$), and errors ($E^{(n)}$).
The error term $E^{(n)}$ only contains words of weight $|S| \geq \dstar$, whereas the others only contain words of weight $|S| < \dstar$.
The distinction between $D^{(n)}$ and $M^{(n)}$ is technical, and we do not discuss this explicitly here.

With this splitting, we formally construct a unitary rotation (in particular its generator $A^{(n)}$, see Sec.~\ref{sec:construction_of_A}) to eliminate the non-frustration-free term $V^{(n)}$ to its leading order, in a sense made precise in Sec.~\ref{sec:KAM_step}.
Formally, this rotates the Hamiltonian into a new basis, in which the non-frustration-free term $V^{(n+1)}$ is expected to be subleading compared to $V^{(n)}$.

Assuming conditions stated in Sec.~\ref{sec:theorems}, we can obtain control over the locality of $A^{(n)}$ for small $V^{(n)}$ (Proposition~\ref{lem: locality of a}).
With this, we further control the norm of the rotated operators, where Proposition~\ref{lem: commutator} is used extensively.
The locality of operators is characterized by the word norm $||\ldots||_\mu$, and the iteration of the word norms are captured by ``flow equations'', see Sec.~\ref{sec:running_couplings}.
Iterating the flow equations, we can show that (1) $D^{(n)} + M^{(n)}$ retain a finite word norm; (2) $V^{(n)}$ is weakened double-exponentially with $n$; and (3) $E^{(n)}$ remains small (Sec.~\ref{sec:bound_on_error}).

This procedure is carried out iteratively until a scale $n_*$,  where $V^{(n_*)}$ is subleading to $E^{(n_*)}$, and further rotations will no longer be necessary.
We use the bounds on each of the terms to obtain control over the specturm of the perturbed Hamiltonian, as detailed in Sec.~\ref{sec: relative boundedness} and Sec.~\ref{sec:effects_of_error_on_spectrum}, thereby proving Theorem~\ref{thm: main}.
Finally, in Sec.~\ref{sec:pert_spec_proj}, we relate the ground space projectors of the perturbed and unperturbed Hamiltonians with a locality-preserving unitary rotation, thereby proving Theorem~\ref{thm: locality}.

% \noindent\YL{[BEGIN OF A SLIGHTLY DIFFERENT PHRASING OF 4.1-4.4]}

\subsection{Splitting of Hamiltonians}\label{sec: splitting of ham}
The initial Hamiltonian is written as
\begin{equation}
H=H^{(0)}=H_0+Z^{(0)}.
\end{equation}
At each stage, we will have the Hamiltonian written (after rotation) as
\begin{equation}\begin{split} \label{eq: rotated ham}
 H^{(n+1)} &  = \exp{-i \mathrm{ad}_{A^{(n)}}} ( H^{(n)} ) \\
 % &= H_0+Z^{(n+1)}+ \exp{i \mathrm{ad}_{A^{(n)}}} (E^{(n)}) \\
 &= H_0+Z^{(n+1)} + E^{(n+1)} \\
&= H_0+ D^{(n+1)} + V^{(n+1)} + M^{(n+1)}  + E^{(n+1)}
\end{split}\end{equation}
% \YL{In this new definition, we just have
% \begin{equation}
% Z^{(n+1)} = D^{(n+1)} + V^{(n+1)} + M^{(n+1)}.
% \end{equation}
% We adopt this new definition of $Z^{(n+1)}$ as the previous equations involving $Z^{(n)}$ have some typos.
% }
where $\mathrm{ad}_A(B)=[A,B]$ and 
% \begin{equation}\begin{split}
%  H^{(n)} &= H_0+(\Delta H)^{(n)}+E^{(n-1)} \\
% &= H_0+ D^{(n)} + V^{(n)} + M^{(n)}  + E^{(n)}
% \end{split}\end{equation}
where $Z^{(n)}$ will be given as an operator-collection, i.e.\  
\begin{equation}
Z^{(n)}=\sum_{\bfs} Z^{(n)}_{\bfs}
\end{equation}
and similarly for $D^{(n)},V^{(n)},M^{(n)}$, with moreover a splitting $V^{(n)}= (V^{(n)})^++(V^{(n)})^- $.  
The latter operator-collections  have particular conditions on the words $\bfs$ that they contain.  We state these conditions here, dropping the superscript $n$
\begin{enumerate}
\item Every one of $V^+,V^-,M,D$ contains only words $\bfs$ such that $|S|<\dstar$. 
    \item $V^+$ contains only words $\bfs$ such that $S_e=S_-=\emptyset$ but $S_+ \neq \emptyset$. 
        \item $V^-$ contains only words $\bfs$ such that $S_e=S_+=\emptyset$ but $S_- \neq \emptyset$. 
             \item $M$ contains only words $\bfs$ such that $S_e=S_+=S_-=\emptyset$ (but hence $S_g \neq \emptyset$). 
             \item $D$ contains only words $\bfs$ for which both $S_e\cup S_- \neq \emptyset$ and $S_e\cup S_+\neq \emptyset$.  
\end{enumerate}
A bit of thought shows that the every word $\bfs$ with $S<d$ is therefore assigned to one (and only one) of $V^+,V^-,M,D$. 
This is how we will define in practice the operator-collections  $V^+,V^-,M,D$ starting from an operator-collection $Z$. 

\begin{figure}[h]
\centering
\includegraphics[width=.9\linewidth]{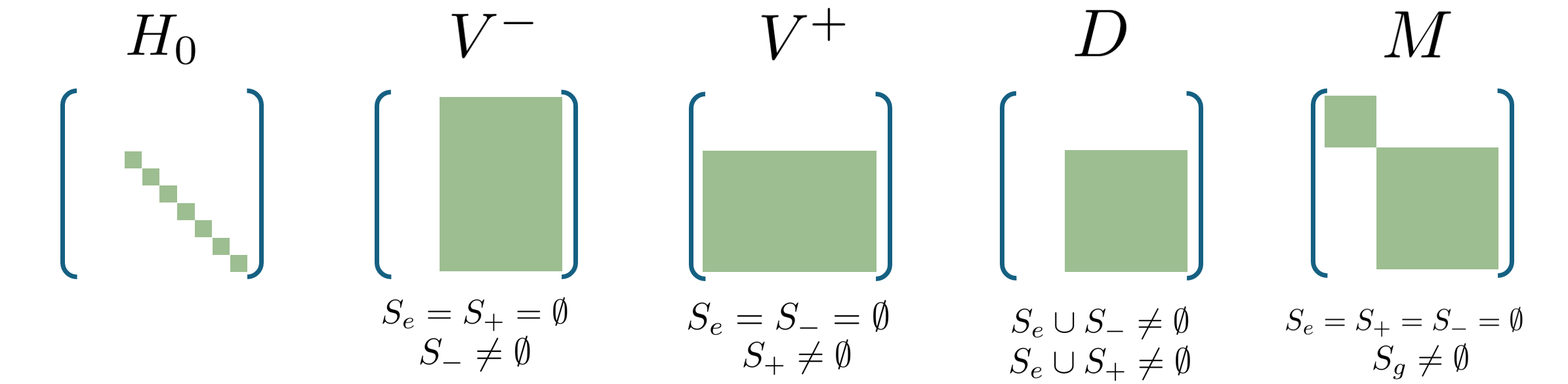}
 \caption{Depiction of the breakdown of $H^{(n)}$ into $H_0,D^{(n)},(V^{(n)})^+,(V^{(n)})^-,M^{(n)}$ in the eigenbasis of $H_0$, with the upper left block corresponding to the ground states of $H_0$. The white regions of the matrices correspond to entries that are necessarily zero by the definitions of these operator collections. When the context is clear, we sometimes omit the superscript $(n)$ for brevity. The aim of the KAM procedure is to rotate away $V^+, V^-$ into $D$ and $M$ (Section~\ref{sec:KAM_step}). Though $D$ and $M$ are both block diagonal, they contain different words, and we handle $M$ with special care when constructing the generator of rotations (Section~\ref{sec:construction_of_A}). We maintain control over the locality of the rotated operators (Sections~\ref{sec: locality of a} and~\ref{sec: controlling_rotation}) in the sense of the intensive word norm; this control over locality has several benefits. By TQO-I, only sufficiently nonlocal operators can be logical operators and hence cause splitting between the ground states. We throw nonlocal operators into $E$ (not pictured, can have any entry nonzero), and we bound $E$'s operator norm through our control over the locality, allowing us to upper bound the splitting. As an aside, TQO-I ensures that $M$ is a linear combination of stabilizers (Section~\ref{sec:proof_rel_bounded}) and is then strictly diagonal in the eigenbasis of $H_0$. Control over locality also plays an important role in arguing (Section~\ref{sec:proof_rel_bounded}) that the lower right blocks of $D$ and $M$ corresponding to the excited states of $H_0$ cannot affect the low-energy spectrum of $H_0 + D + M$ too much; in particular, in the Hamiltonian $H_0+D+M$, this block cannot shift the energies of the excited states of $H_0$ below those of the ground states of $H_0$.}
\label{fig: decomposition}
\end{figure}

% where 
% \begin{enumerate}
%     \item $H_0$ is the unperturbed Hamiltonian.
%     \item $D$ is the frustration-free correction, such that the unperturbed ground-state sector is annihilated by all terms in $D$.
%     \item $V$ contains the frustrating terms. 
%     \item $M$ contains terms that are purely $gggg$.
%     \item $S$ are spectator terms, or just error terms.
% \end{enumerate}

\subsection{A KAM step}\label{sec:KAM_step}
Let us put an index $(n)$ (scale) on all of these terms except $H_0$, and then the generator of rotations $A^{(n)}$ is defined such that it satisfies
\begin{equation}
 \overline P (i[H_0+D^{(n)}+M^{(n)},A^{(n)}]+V^{(n)}) P =  0 \label{eq: basic a property one}
\end{equation}
where $P$ is the ground state projector (c.f. Section~\ref{sec:stabilizer_Hamiltonian}) and $\overline P = (1-P)$. 

\begin{equation}
 P (i[H_0+D^{(n)}+M^{(n)},A^{(n)}]+V^{(n)}) \overline P =0    \label{eq: basic a property two} 
\end{equation}
This leads to definitions of all terms, as follows.
Define $\widetilde{Z}^{(0)} = Z^{(0)}$, and
\begin{equation}\begin{split}
\label{eq:tildeZn+1_from_Zn}
    \widetilde{Z}^{(n+1)} \equiv Z^{(n)}+ \sum_{k=1}^{\infty}    \frac{1}{k!} 
(-i\mathrm{ad}_{A^{(n)}   })^k   (H_0+Z^{(n)} ).
\end{split}\end{equation}
This gives a unique definition of the operator $\widetilde Z^{(n+1)}$. By the rules stated in subsection \ref{sec: collections and norms}, the expression \eqref{eq:tildeZn+1_from_Zn} also defines an operator-collection $\widetilde Z^{(n+1)}$. 
Next, we construct $Z^{(n+1)}$ so as to satisfy \eqref{eq: rotated ham}, namely 
\begin{equation}\label{eq: delta h}
  Z^{(n+1)}_\bfs = \begin{cases}
      \widetilde{Z}^{(n+1)}_\bfs & \bfs:|S| < \dstar \\
  0 & \text{otherwise}
  \end{cases}
\end{equation}

The operator-collections  $D^{(n+1)},V^{(n+1)},M^{(n+1)}$ are now defined as indicated in subsection \eqref{sec: splitting of ham}, as restrictions of $Z^{(n+1)}$ to the appropriate words $\bfs$.
The crucial thing to note here is that, by virtue of equations (\ref{eq: basic a property one},\ref{eq: basic a property two}), the operator-collection $V^{(n+1)}$ receives only contributions from (in other words: ``is the restriction of")
\begin{equation} \label{eq: origin of v}
\sum_{k=1}^{\infty}    \frac{1}{k!} 
(- i\mathrm{ad}_{A^{(n)}   })^k   (V^{(n)} ) +   \sum_{k=2}^{\infty}    \frac{1}{k!} 
(-i\mathrm{ad}_{A^{(n)}   })^k   (H_0+D^{(n)}+M^{(n)} ).   
\end{equation}
The error operator is then defined as 
\begin{align}\label{eq: definition error}
E^{(n+1)} &= \exp{-i \mathrm{ad}_{A^{(n)}}}E^{(n)}+ \sum_{\bfs: |S|\geq \dstar}(\widetilde{Z}^{(n+1)})_{\bfs} \\
&= \exp{-i \mathrm{ad}_{A^{(n)}}}E^{(n)}+ \sum_{\bfs: |S|\geq \dstar}(\widetilde{Z}^{(n+1)}-Z^{(n)})_{\bfs}   \label{eq: definition error}
\end{align}
In Section~\ref{sec:bound_on_error}, we bound the operator norm of $E^{(n)}$ using that it contains high weight operators that will be highly suppressed by virtue of the finite intensive word norms.

% \begin{equation}\begin{split}
%  H^{(n+1)} =&H_0+D^{(n)}+M^{(n)} +  \sum_{k=2}^{\infty}\frac{1}{k!} \mathrm{ad}_{A^{(n)}}(H_0+D^{(n)}+M^{(n)})   \exp{i \mathrm{ad}_{A^{(n)}}} E^{(n)}   
% \end{split}\end{equation}

\subsection{Bounds on running coupling constants}\label{sec:running_couplings}
We start from the perturbation strength
\begin{equation}
\epsilon_0= ||Z^{(0)}||_{\mu_0}, 
\end{equation}
implying also
\begin{equation}
||V^{(0)}||_{\mu_0}   \leq \epsilon_0, \qquad ||D^{(0)}||_{\mu_0}  \leq \epsilon_0, \qquad  ||M^{(0)}||_{\mu_0}  \leq \epsilon_0,
\end{equation}
and we define the running coupling constants
\begin{equation}
\eta_n \coloneqq ||  D^{(n)}||_{\mu_n} + ||  M^{(n)}||_{\mu_n}, \qquad 
\epsilon_n \coloneqq || V^{(n)}||_{\mu_n},\qquad n\geq 1
\end{equation}
for a decreasing sequence $(\mu_n)_{n=0,1,\ldots}$ that will be fixed below.
% In fact, since $V^{(n)}$ is the genuine perturbation, rather than $Z^{(n)}$, we could improve our results by 
Let us now derive some flow equations, or rather: ``flow inequalities", for these running coupling constants.
For these to be valid, we need that $\mu_n$ is larger than the critical value $\mu_\ast$ (see \eqref{eq:def_mu_star}), as this is a prerequisite for Proposition \ref{lem: locality of a}, which will be used below.  
We assume therefore that $\mu_n >\mu_*$ for any $n$, and we will fix $\mu_n$ more precisely below.
 % \wdr{I guess that at the latest here, it would be good to define $\mu_*$.}
 % \\

~

\noindent\emph{Derivation of flow equations}. By Proposition \ref{lem: locality of a}, we have
\begin{equation}\label{eq: bound on a norm}
    || A^{(n)}||_{\mu_n} \leq \frac{2\epsilon_n}{1-e\eta_n},
\end{equation}
provided that $e\eta_n<1$. 
To get a bound on $\epsilon_{n+1}$, we start from 
\eqref{eq: origin of v} and we use equation \eqref{eq: bound on a norm} and Proposition \ref{lem: commutator}. This yields
\begin{equation}
\epsilon_{n+1} \leq \sum_{k=1}^{\infty}  \frac{(k+1)!}{k!} \frac{(4e\epsilon_n (1-e\eta_n)^{-1})^k}{|\mu_n-\mu_{n+1}|^{k+1}} \epsilon_n + 
\sum_{k=2}^{\infty}  \frac{(k+1)!}{k!} \frac{(4e\epsilon_n (1-e\eta_n)^{-1})^k}{|\mu_n-\mu_{n+1}|^{k+1}} (e^{\mu_0}+ \eta_n), 
\end{equation}
where we have also used the equality $
   ||H_0||_{\mu_0} = e^{\mu_0} $.
To bound $\eta_n$, we directly use
\eqref{eq:tildeZn+1_from_Zn} to get, again using Proposition \ref{lem: locality of a} and \ref{lem: commutator}, 
\begin{equation} \label{eq: inductive bound etan}
\eta_{n+1} \leq \eta_n+ \epsilon_n +  \sum_{k=1}^{\infty}  \frac{(k+1)!}{k!} \frac{(4e\epsilon_n (1-e\eta_n)^{-1})^k}{|\mu_n-\mu_{n+1}|^{k+1}} (e^{\mu_0}+ \eta_n+\epsilon_n).
\end{equation}
% \YL{In the previous version, the $\epsilon_n$ term was missing. Unimportant though.}
We will now fix the decay rates $\mu_n$. As already stated, we need $\mu_n>\mu_*$ and so a convenient choice is 
\begin{equation}
\mu_n-\mu_{n+1} =
\frac{3}{\pi^2}
(\mu_0-\mu_*)\frac{1}{(n+1)^2}
\end{equation}
which ensures that 
\begin{equation}
\mu_n \geq \mu_{\infty}:=\mu_*+ (\mu_0-\mu_*)/2,
\end{equation}
for all $n$. 

Our most important result is then Proposition \ref{prop: running couplings} below. It is derived from the above estimates by inspection. 
\begin{proposition}\label{prop: running couplings}
Let $\mu_0>\mu_*$ and let $\mu_n$ be fixed as explained above.  
There is a $\epsilon(\mu_0,\mu_*) >0$, such that, if $\epsilon_0\leq \epsilon(\mu_0,\mu_*)$, then 
\begin{equation}
\eta_{n}\leq 2\epsilon_0,\qquad   \epsilon_{n} \leq  \epsilon^{(3/2)^{n}}_0.
\end{equation}
\end{proposition}
 We connect this now to the assumptions made in our results:
\begin{proposition}
     The conditions of Proposition~\ref{prop: running couplings} follow from the conditions of Theorem~\ref{thm: main}. 
\end{proposition}
\begin{proof}
For ease, define $f(z) = \frac{1}{\wq}(z - (w_q\log(2) + \kappa + \log(2 w_c)))$. Note that $\widetilde{\mu}_0 > \mu_{\#}  \leftrightarrow f(\widetilde{\mu}_0) > \mu_{*}$, and $|||Z^{(0)}|||_{\widetilde{\mu}_0} < \widetilde{\epsilon}(\widetilde{\mu_0}) \rightarrow ||Z^{(0)}||_{f(\widetilde{\mu}_0)} < \widetilde{\epsilon}(\widetilde{\mu}_0)$ by Proposition~\ref{prop:tripleboundsdouble}. In Theorem~\ref{thm: main}, we choose $\widetilde{\epsilon}(\widetilde{\mu}_0, \mu_{\#}) := \epsilon(f(\widetilde{\mu}_0), f(\mu_{\#})) = \epsilon(f(\widetilde{\mu}_0), \mu_*)$. Then the conditions of Theorem~\ref{thm: main} being satisfied imply that the conditions of Proposition~\ref{prop: running couplings} are satisfied with the choice $\mu_0 = f(\widetilde{\mu}_0)$ and $\epsilon_0= ||Z^{(0)}||_{\mu_0}$.
\end{proof}

%In particular, the condition on $|||Z|||_{\widetilde{\mu}_0}$ in theorem 1 will translate onto a condition satisfying $\widetilde{\mu_0}$ Proposition~\ref{prop:tripleboundsdouble}, we translate the condition $\mu_0 > \mu_*$ given in Proposition~\ref{prop: running couplings} to the condition $\widetilde{\mu}_0 > \mu_{\#}$ given in Theorem~\ref{thm: main}. In particular, given a $\widetilde{\mu}_0 > \mu_{\#}$ for which $|||Z^{(0)}|||_{\widetilde{\mu}_0} \leq \epsilon$, then the choice of $\mu_0 = $ $||Z^{(0)}||_{}$.

% \YL{The following uses Proposition 2 from Sec.6.}
% \wdr{\begin{equation}\begin{split}
%   \epsilon_{n+1} =  || \exp{i \mathrm{ad}_{A^{(n)}}} V^{(n+1)}||_{\mu_{n+1}}  
%    &\leq   \sum_{k=1}^{\infty}    \frac{1}{k!} (\frac{k}{k+1})  \frac{ (2e)^k(k+1)!}{(\mu_n-\mu_{n+1})^{k+1}}  ||A^{(n)}||^k_{\mu_n}   ||V^{(n)}||^k_{\mu_n} 
%   \\
%    &\leq   \sum_{k=1}^{\infty}  \frac{k(2e)^k}{(\mu_n-\mu_{n+1})^{k+1}}  \epsilon^{k+1}_n  
% \end{split}\end{equation}
% So if $ C\frac{\epsilon_n}{(\mu_n-\mu_{n+1})} <1/2  $, then 
% $$
% \epsilon_{n+1} \leq C \frac{\epsilon^2_n}{(\mu_n-\mu_{n+1})^2}
% $$}
% \wdr{The lines above are still not corrected, as I first want to fix the constants correctly and see how precisely it needs to be done. }

\subsection{Bounds on error terms $E^{(n)}$}\label{sec:bound_on_error}

Since the error term is the only term that can contain logical operators, it is ultimately responsible for the splitting between ground states of the rotated Hamiltonian. We can strongly bound the error term $E$, starting from \eqref{eq: definition error}:
\begin{align}
|| E^{(n+1)} || \leq&\,  || E^{(n)} ||  +  \sum_{\bfs: |S|\geq \dstar} || 
 (\widetilde{Z}^{(n+1)}-Z^{(n)})_{\bfs} || \\
\leq&\,  || E^{(n)} || + N_c e^{-\mu_{n+1} \dstar} || \widetilde{Z}^{(n+1)}-Z^{(n)}||_{\mu_{n+1}}
\end{align}
Our bound for $|| \widetilde{Z}^{(n+1)}-Z^{(n)}||_{\mu_{n+1}}$ is the same as the right hand side of \eqref{eq: inductive bound etan} but without the term $(\eta_n+\epsilon_n)$ (because now we are subtracting $Z^{(n)}$).  By using Proposition \ref{prop: running couplings}, we can bound the resulting series as   
\begin{equation}
    || \widetilde{Z}^{(n+1)}-Z^{(n)}||_{\mu_{n+1}} \leq C  n^4 \epsilon_{n} 
\end{equation}
Therefore,
\begin{equation}
|| E^{(n+1)} || \leq  || E^{(n)} ||  + C {N_c} e^{-\mu_{n}\dstar} n^4 \epsilon_{n}.
\end{equation}
Since $E^{(0)}=0$, and using Proposition~\ref{prop: running couplings} ($\epsilon_{n} \leq  \epsilon^{(3/2)^{n}}_0$), we get 
 the uniform in $n$-bound
\begin{equation}
|| E^{(n)} || \leq  C {N_c} e^{-\mu_{\infty}\dstar} \epsilon_0
\end{equation}
% \wdr{last bound is precise; penultimate not quite, though it does not matter. It is a bit painful. The fixing of constants in previous subsection should have been done so that this is correct, but pffff..}
In principle, we could run the renormalization scheme an infinite number of times, to obtain the Hamiltonian
\begin{equation}
H^{(\infty)}=H_0+D^{(\infty)}+M^{(\infty)}+E^{(\infty)},
\end{equation}
i.e.\ where the $V$-term has been completely eliminated.
However, there is no need to do this, since after $n_*$ KAM steps, with $n_*$ the smallest natural number such that
\begin{equation}
\epsilon_{n_*} \leq  C' \epsilon_0 e^{-\mu_{\infty}\dstar},
\end{equation}
we already obtain the operator $V^{(n_*)}$, such that 
$|| V^{(n_*)}|| \leq   || E^{(n_*)}||   \leq  C {N_c}\epsilon_0 e^{-\mu_{\infty}\dstar}$,
which leads to the same result. 

% \YL{Is $n_*$ dependent on $d$?}  \wdr{Yes, that's the idea}

% \noindent\YL{[--------------------------------------------------- END ---------------------------------------------------]}

\subsection{Relative boundedness}\label{sec: relative boundedness}

Up to now, we have passed from $H^{(0)}$ to $H^{(n_*)}$. By construction,  the operators
\begin{equation}
H_0+ D^{(n)}
\end{equation}
satisfy 
\begin{equation}
(H_0+ D^{(n)})P=P(H_0+ D^{(n)})=0. 
\end{equation}
Once one adds the term $M^{(n)}$, this is no longer true, as the energy of the ground state sector shifts.  Either way, one needs a nontrival argument to show that the spectral gap of $
H_0+ D^{(n)}+ M^{(n)}
$ remains open.  For any operator $O$, we define the ground space expectation 
\begin{equation}
  \langle O\rangle= \frac{\mathrm{Tr} P OP }{\mathrm{Tr} P}.
\end{equation}
Then 
\begin{equation}
K^{(n)}= H_0+ D^{(n)}+ M^{(n)}-\langle  M^{(n)} \rangle 
\end{equation}
also satisfies
\begin{equation} \label{eq: zero eigenvalue of kn}
K^{(n)} P= PK^{(n)}= 0.
\end{equation}
This is because every term $M_{\bfs}$ satisfies
\begin{equation}
PM_{\bfs}=M_{\bfs}P=\langle M_{\bfs} \rangle P
\end{equation}
as a direct consequence of  \textbf{TQO-I}.
% ....
% ....
% \wdr{is by now propably stated in a few places...}
Furthermore, it is manifest (c.f.\eqref{eq: collection to operator bound}) that 
\begin{equation}
|\langle  M^{(n)} \rangle |\leq  ||M^{(n)}  ||_0  {N_c} \leq C \epsilon_0 {N_c},
\end{equation}
uniformly in $n$, so the shift in energy density is small.  
To prove that the gap remains open, we establish
\begin{proposition}\label{prop: relative boundedness}
The operator $K^{(n)} - H_0 $ is relatively bounded with respect to $H_0$, with relative bound $C_0\epsilon_0$. More precisely, it holds that
\begin{equation}
|| (K^{(n)} -H_0)\psi || \leq    C_0\epsilon_0 ||H_0\psi||.
\end{equation}
for any $\psi \in \mathcal{H}_\Lambda$ and a 
constant $C_0$ uniform in $n$.
% \wdr{Do we have a notation for Hilbert space $\mathcal{H}$?}
\end{proposition}
This proposition is proven in Section~\ref{sec:proof_rel_bounded}.
As is well-known, relative boundedness leads to some stability of the spectrum. 
In particular, one deduces that the resolvent $(z-K^{(n)})^{-1}$ enjoys the bound
\begin{equation}\label{eq: first bound resolvent}
|| (z-K^{(n)})^{-1}|| \leq \frac{1}{r_z}  \frac{1}{1-C_0\epsilon (1+(|z|/r_z))  },\qquad r_z=\mathrm{dist}(z,\mathrm{spec}(H_0))
\end{equation}

Given that $$\mathrm{spec}(H_0)\subset \{0,1,2,\ldots\},$$  we deduce that, provided $C_0\epsilon_0 <1 $ is small enough,  the spectrum of $K^{(n)}$ is contained in discs centered on $m=0,1,2,\ldots$ with radius $Cm\epsilon_0$.
In particular, if $K^{(n)}$ is Hermitian, then
\begin{equation}
\mathrm{spec}(K^{(n)}) \subset \bigcup_{m=0,1,2,\ldots}   [(1-C\epsilon_0)m,(1+C\epsilon_0)m]
\end{equation}
% This follows for example from Lemma 3.2 in \cite{nachtergaele2022quasi}.
Note that we already explicitly checked that $0$ is an eigenvalue in \eqref{eq: zero eigenvalue of kn}, and the above inclusion says moreoever that $0$ is an isolated part of the spectrum.

\subsection{Effect of $E^{(n)}$ on spectrum}\label{sec:effects_of_error_on_spectrum}
Finally, we discuss the spectral problem for the operator $H^{(n_*)}$ that we cast as
\begin{equation}
H^{(n_*)}=  K^{(n_*)} +   (V^{(n_*)}+E^{(n_*)})   +  \langle M^{(n_*)}\rangle
\end{equation}
The last term is a multiple of identity, and the operator $(V^{(n_*)}+E^{(n_*)})$ is bounded as
\begin{equation}\label{eq: ridiculous perturbation}
||V^{(n_*)}+E^{(n_*)} || \leq C {N_c} \epsilon_0 e^{-\mu_\infty \dstar} =: \delta,
\end{equation}
by the results of subsection \ref{sec:bound_on_error}.
Hence, the resolvent $(z-H^{(n_*)}-\langle M^{(n_*)}\rangle )^{-1}$ exists for $z$ such that
$$
\delta < || (z-K^{(n)})^{-1}||^{-1}
$$
By inspection of  \eqref{eq: first bound resolvent}, we deduce that the spectrum of $H^{(n_*)}-\langle M^{(n_*)}\rangle $ is contained in discs centered on $m=0,1,2,\ldots$ with radius $Cm\epsilon_0+C\delta$. 
If all operators are Hermitian, then these results can be slightly strenghtened (since in that case $|| (z-K^{(n)})^{-1}|| = \frac{1}{\mathrm{dist}(z, \mathrm{spec}(K^{(n)} )}$) and we obtain that
\begin{equation}
\mathrm{spec}(H^{(n_*)}-\langle M^{(n_*)}\rangle ) \subset \bigcup_{m=0,1,2,\ldots}   [(1-C_0\epsilon_0)m -\delta,(1+C_0\epsilon_0)m+\delta].
\end{equation}
and since $\langle M^{(n_*)}\rangle$ is a number, this implies Theorem \ref{thm: main}.

% Given the information on resolvent of $K^{(n_*)} $ given in \eqref{eq: perturbed resolvent}, and the strong norm bound above,  
% we deduce that $z$ satisfying 
% \begin{equation}
% 4{N_c} \epsilon_0 e^{-\mu_\infty \dstar} \leq \frac{1}{r(z)}
%    \frac{1}{1- C_0\epsilon_0 (1+ \frac{|z|}{r(z)})   }  
% \end{equation}
% is outside the spectrum of $
% H^{(n_*)}$. 
% In particular, 
% we deduce that the spectral patch emanating from the eigenvalue at $0$ lies   in a ball of radius $C {N_c} \epsilon_0 e^{-\mu_\infty \dstar}$. The spectral patch emanating from the eigenvalues at the natural $m$ lie in a ball of radius $C\epsilon_0 m$, for any $m< \frac{1}{C'\epsilon_0}$.

Moreover, it follows that, for $m$ such that $(1+C_0\epsilon_0)m+\delta <  (1-C_0\epsilon_0)(m+1)-\delta $, the spectral patch contained in $[(1-C_0\epsilon_0)m -\delta,(1+C_0\epsilon_0)m+\delta]$ has the same cardinality as the unperturbed eigenvalue $m$ of $H_0$.

\subsection{The perturbed spectral projector}\label{sec:pert_spec_proj}

We recall that $H_0$ and $K^{(n_*)}$ have an isolated eigenvalue at zero, with corresponding spectral projector $P$.   For the operator $H^{(n_*)}$, the eigenvalue and spectral projector are perturbed, and we call the latter $P'$. By standard, rigorous perturbation theory, see e.g.\ \cite{kato2013perturbation}, we can construct an operator $A^{(\infty)}$ with 
\begin{equation} \label{eq: bound on a infinity}
||A^{(\infty)} || \leq C {N_c} \epsilon_0 e^{-\mu_\infty \dstar}
\end{equation}
such that
\begin{equation}
   P'= \exp{-i\mathrm{ad}_{A^{(\infty)}}}(P) 
\end{equation}
Ultimately, we are interested in the spectral projector $\widetilde P$ of the original Hamiltonian $H=H_0+ Z^{(0)}$.  To relate it to the spectral projector $P'$ of $H^{(n_*)}$, we simply need to undo the successive rotations, i.e.\
\begin{equation}\begin{split} \label{eq: final spectral projector}
    \widetilde P  &= 
    \exp{-i\mathrm{ad}_{A^{(0)}}}\ldots \exp{-i\mathrm{ad}_{A^{(n_*)}}} (P') \\
    &=     \exp{-i\mathrm{ad}_{A^{(0)}}}\ldots \exp{-i\mathrm{ad}_{A^{(n_*)}}}  \exp{-i\mathrm{ad}_{A^{(\infty)}}}(P) =: U P U^{-1}
\end{split}\end{equation} 
The operator $U$ in the last line can simply be defined as the product $e^{-i A^{(0)}}\ldots e^{-i A^{(n_*)}} e^{-i A^{(\infty)}}$. It is unitary if the perturbation is Hermitian, since then $A^{(n)}$ and $A^{(\infty)}$ can be chosen Hermitian. 
We now come to the proof of Theorem \ref{thm: locality} and Theorem \ref{thm: analyticity}, i.e.\ we investigate the locality properties and analyticity of the maps $O \mapsto U O U^{-1}$ and its inverse $O \mapsto  U^{-1} OU $. 
For the sake of explicitness, let us focus on the latter, as the former is treated analogously. It is important that we allow the perturbation to be non-Hermitian, as we also did in the previous subsections.
We use the shorthand
\begin{equation}
\caU^{(j)}=  \exp{i\mathrm{ad}_{A^{(j)}}},\qquad  \caU^{(\infty)}=  \exp{i\mathrm{ad}_{A^{(\infty)}}}
\end{equation}
so that the object of interest reads
\begin{equation}
  U^{-1} OU=  \caU^{(\infty)} \caU^{(n_*)}\ldots \caU^{(0)} (O)
\end{equation}
We write this as
\begin{align} \label{eq: splitting transfos}
U^{-1} OU  &= O  +
 ( O' -O  )+ (\caU^{(\infty)} (O') -O' )
\end{align}
where $O'=\caU^{(n_*)}\ldots \caU^{(0)} (O) $.   The three terms on the right-hand side of \eqref{eq: splitting transfos} correspond to the three terms on the right-hand side of \eqref{eq: definition locality observables under u}.
The third term, which is hence identified as the background term $O_{\mathrm{bg}}$, is bounded in norm by  $C||A^{(\infty)} ||  ||O'|| \leq C {N_c} \epsilon_0 e^{-\mu_\infty \dstar} ||O'||$ by using the smallness of $||A^{(\infty)} || $, i.e.\ \eqref{eq: bound on a infinity}. This will yield the claimed bound once we establish that $||O'||\leq C||O||$, which will follow from the treatment below.  Note that, 
if the perturbation is Hermitian, we have $||O'||=||O||$ by unitarity.

To estimate $|| O' -O  ||$, we will interpret $O',O$ as operator-collections and we first
establish locality-preservation for $\caU^{(n)}\ldots \caU^{(0)}$ when acting on operator-collections:
\begin{lemma}
For any operator-collection $B$ and any $n$,
\begin{equation}
    || \caU^{(n)}\ldots \caU^{(0)}(B)-B ||_{\mu_{n+1}} \leq C \epsilon_0 ||B||_{\mu_0} 
\end{equation}
\end{lemma}
\begin{proof}
We write
\begin{equation}
\caU^{(n)}\ldots\caU^{(0)}-1 = \sum_{j=0}^n     
( \caU^{(n)} \ldots \caU^{(j+1)} )(\caU^{(j)}-1),
\end{equation}
and we apply this map to an operator-collection $B$. 
We use \eqref{eq: effect of exp first} to bound $(\caU^{(j)}-1)(B)$ in the norm $||\cdot||_{\mu_{j+1}}$, and then successively \eqref{eq: effect of exp zeroth} to bound  $\caU^{(i)}(B')$ in the norm $||\cdot||_{\mu_{i+1}}$ with $i>j$ for the resulting operator-collection $B'$. In applying these bounds, we exploit the estimate
 $|| A^{(n)}||_{\mu_n} \leq C\epsilon_n$, see \eqref{eq: bound on a norm}. The lemma then follows by the fast decay of $\epsilon_n$.
\end{proof}

\noindent Finally, we  pass from the action of $\caU^{(0)}\ldots \caU^{(n)}$ on operator-collections , to its action on local operators $O$ supported in a given set $Y$. In particular, we want to argue that the second term on the right hand side of  \eqref{eq: splitting transfos} is exponentially quasi-local (which, as argues, also settles the third term).

%To that end, we encode $O$ as an operator-collection $O$, where all words have support in a set $\widetilde Y$.  \YL{We may choose $\widetilde{Y} = \{\alpha : \supp(C_\alpha) \cap Y \neq \emptyset\}$, so that $|\widetilde{Y}| \leq w_q |Y|$.}
% \wdr{translate between stabilizers and words}.  
%Then, for any $\mu_0$, we have $||O||_{\mu_0}\leq C^{ \mu_0 |\widetilde Y|}$ for some geometrical constant $C$.  By the above lemma, we know that
%\begin{equation}
%||\caU^{(n_*)}\ldots\caU^{(0)}(O)- O ||_{\mu_n} \leq  C \epsilon_0 C^{ \mu_0 |\widetilde Y|}
%\end{equation}
%Moreover, by the explicit definition of commutators of operator-collections in \eqref{def: commutator}, we know that the support of every word $\bfs$ that contributes to $\caU^{(n_*)}\ldots\caU^{(0)}(O)- O$, intersects $\widetilde Y$. This yields that   $\caU^{(n_*)}\ldots\caU^{(0)}(O)- O$ is exponentially localized around $\widetilde Y$ in the sense of words $\bfs$. 
%Upon translating this to a decay on the qubit graph, we conclude the proof of Theorem \ref{thm: locality}.

To that end, we encode $O$ as an operator-collection $O$ according to the procedure for $Z^{(0)}$ in Section~\ref{sec:initial_choice_operator_collection}. We write $O = \sum_{p} c_p p$ and note that $|c_p|<||O||$. We break each $p$ into an operator-collection $p = \sum_{\mathbf{S}} p_{\mathbf{S}}$. To construct this decomposition for a given $p$, we first find its qubit support $\supp(p)$ and a minimal connected set $\mqs(p)$ of qubits such that $\supp(p) \subset \mqs(p)$. We will choose $\mqs(p)$ entirely within $Y$, which we can do by the connectedness of $Y$; $|\mqs(p)| \leq |Y|$. The $\bfs$ that are included in the sum have $S = \ext(p)$ and accordingly $|S| \leq w_q |\mqs(p)| \leq w_q |Y|$. For a given $p$, there are at most  $2^{\ext(p)}$ $\mathbf{S}$ such that $p_{\mathbf{S}}$ is nonvanishing, and all of them satisfy $||p_{\mathbf{S}}|| \leq 1$. Together, these ensure that $||O||_{\mu_0} \leq ||O|| 2^{w_q |Y|} e^{\mu_0 w_q |Y|} = ||O|| e^{(w_q \mu_0 + \log(2))|Y|}$. By the above lemma, we know that
\begin{equation}
||\caU^{(n_*)}\ldots\caU^{(0)}(O)- O ||_{\mu_n} \leq  C \epsilon_0 ||O|| e^{(w_q \mu_0 + \log(2))|Y|}
\end{equation}
Moreover, by the explicit definition of commutators of operator-collections in \eqref{def: commutator}, we know that the support of every word $\bfs$ that contributes to $\caU^{(n_*)}\ldots\caU^{(0)}(O)- O$, intersects $\widetilde Y$. This yields that   $\caU^{(n_*)}\ldots\caU^{(0)}(O)- O$ is exponentially localized around $\widetilde Y$ in the sense of words $\bfs$. 
Upon translating this to a decay on the qubit graph, we conclude the proof of Theorem \ref{thm: locality} and Theorem \ref{thm: analyticity}.

\section{Construction of the generator of rotations $A$} \label{sec:construction_of_A}
We need to find
$A={A^+}+A_-$ such that it satisfies the equations
\begin{equation}\label{eq: basic equation a}
  \overline P (i[H_0+D+M,{A^+}] +{V^+} )P =0 
\end{equation}
\begin{equation}\label{eq: basic equation a other way}
  P (i[H_0+D+M,A_-] +V_- ) \overline P =0 
\end{equation}
where $P$ projects on the GS sector and $\overline P=1-P$. This decomposition of $A$ corresponds to choosing ${A^+}$ so that the off-block-diagonal part of ${V^+}$ is rotated away at lowest order, and choosing $A_-$ so that the off-block-diagonal part of $V_-$ is rotated away at lowest order. Together, these ensure that the off-block-diagonal part of $V$ is rotated away at lowest order. Furthermore, the decomposition is also useful for constraining the kinds of words $\mathbf{S}$ that appear in solutions to these equations.
% In case $V$ is Hermitian, we can find $A_-$ by setting $A^-=(A^+)^\dagger$.
% In the general case, we let ${A^+}={A^+}(W)$ be our constructed solution, for some $V=+=W$, and then we set
% $$
% A_-= ({A^+}((V_-)^{\dagger}))^{\dagger}.
% $$
% Either way, it is sufficient to find
%  $A^+=A^{+}(V^{+})$. 

In the following, we will only explicitly construct ${A^+}$ as the case of $A_-$ is completely analogous. Indeed, if $D, M$ are self-adjoint, then $A_-=({A^+})^\dagger$. In the general case, we set $A_-=(A'_+)^\dagger$ where $A'_+$ is the solution to 
\begin{equation}\label{eq: basic equation a non hermitian}
  \overline P (i[H_0+D^\dagger+M^\dagger,A'_+] +{V^+} )P =0 
\end{equation}

Note that there are many possible solutions to equations~\ref{eq: basic equation a} and~\ref{eq: basic equation a other way}. Our goal is to find a solution that has good locality properties, so that the KAM steps will not change the locality of the rotated Hamiltonian too much. In the following subsections, we will construct a formal solution $A^+$ satisfying equation~\ref{eq: basic equation a}, and then we check its convergence and its locality properties in Section \ref{sec: locality of a}.

\subsection{Ghost terms}

Recall that $M$ is an operator-collection that consists purely of ghosts satisfying the TQO condition in Sec.~\ref{subsec:TQOPauli}. 
Recall also from Sec.~\ref{sec: splitting of ham} that $M$ only contains words $\bfs$ with $|S| < \dstar$.
\begin{lemma}\label{lem: bound on ghosts}
Let $\bfs'$ be a word with only $+$ and $g$. Then, for any $X_{\bfs'} \in \caX_{\bfs'} $, there is a complex number $q({\bfs',\bfs}) $ such that
        \begin{equation}  \label{eq: bound ghost comm}
 [M_\bfs,X_{\bfs'} ]  P =  q({\bfs',\bfs}) X_{\bfs'}  P.
 % , \qquad  |q({\bfs',\bfs})| \leq 2 ||M_{\bfs}||.
\end{equation}
Moreover, provided that $\mu$ satisfies $4^{w_c + 1} e^{\kappa \cdot \ell-\mu} < 1$, we have 
    \begin{equation} \label{eq: bound total ghost term}
  \frac{1}{|S'_+|} \sum_{\bfs} |q({\bfs',\bfs}) | \leq ||M||_{\mu}.
\end{equation}  
\end{lemma}
\begin{proof}
The operator $M_{\bfs}$ can be decomposed in Pauli strings $p$ such that $\supp_c(p) \subseteq S$, c.f. \eqref{eq:check_supp_of_pauli}, \eqref{eq:supp_of_x_tilde}.
% \vedika{$\supp(p)$ is a qubit set. $\supp_c$ is defined in 2.7 if that's what you want}
Therefore,
\begin{equation}\begin{split}
M_{\bfs}=\sum_{p: \supp_c(p) \subseteq S} c_p p.
\end{split}\end{equation}
Note that $||M_{\bfs}|| = \text{largest singular value of } M_{\bfs} \geq \sqrt{\text{average of the squared singular values}} = \sqrt{\sum_p |c_p|^2} \geq |c_p|, \, \forall p$. Therefore, each of the coefficients $c_p$ is bounded as 
\begin{equation}\begin{split}
|c_p| \leq  || M_{\bfs}||.
\end{split}\end{equation}
By construction, the operator $M_{\bfs}$ commutes with all checks $\{C_\alpha\}_{\alpha \in \caE}$. 
Using that each Pauli string either commutes or anticommutes with $C_\alpha$, we derive that any $p$ for which $c_p\neq 0$, also commutes with all stabilizers, and therefore, because of \textbf{TQO-I}, such $p$ belongs to $\caG$. 
Next, we use \textbf{TQO-II} to recast every such $p$ as a product $\widetilde p=\prod_{\alpha \in \hat{S}(p)} C_\alpha$, where $\hat{S}(p) \subseteq B_{\ell|S|}(S)$.
For each such product $\widetilde p$, we have
\begin{equation}\begin{split}
[ p,X_{\bfs'}]P=[\widetilde p,X_{\bfs'}]P=  \eta
 (1-({-1})^{|\hat{S}(p) \cap S'_+|} )X_{\bfs'} P 
\end{split}\end{equation}
where $\eta \in {\pm 1}$.   Therefore, \eqref{eq: bound ghost comm} holds with 
\begin{equation}\begin{split}
|q({\bfs',\bfs}) | \leq \sum_{p: \supp_c(p) \subseteq S}  2|c(p)| \leq    2 \times 4^{w_c|S|} ||M_\bfs||
\end{split}\end{equation}
where we use that there are at most $w_c|S|$ qubits touching the stabilizers in $S$ (and hence at most $4^{w_c|S|}$ Pauli strings $p$ in the Pauli decomposition of $M_\bfs$). We also note that $q(\cdot,\cdot)$ vanishes whenever $\hat{S}(p) \cap S'_+=\emptyset$.
Therefore, we obtain
\eqref{eq: bound total ghost term} by 
\begin{equation}\begin{split}
   \sum_{\bfs} |q({\bfs',\bfs})| & \leq  \sum_{\bfs: B_{\ell|S|}(S) \cap S'_+ \neq \emptyset}    2 \times 4^{w_c|S|} ||M_\bfs|| \\
   & \leq \sum_{\alpha \in S'_+} \sum_{r = 1}^\infty \sum_{\beta: \mathrm{dist}(\alpha,\beta) \leq \ell \cdot r} \sum_{\bfs: |S| = r, \beta \in S} 2 \times 4^{w_c r} ||M_\bfs|| \\
   % & \leq \sum_{r=1}^{\infty} \sum_{x: \mathrm{dist}(x,S'_+)  \leq r+r_0 } \sum_{\bfs: S \ni x, |S|=r}    2 \times 16^{|S|} ||M_\bfs|| \\
      & \leq 2 |S'_+|  ||M||_\mu  \sum_{r=1}^{\infty} e^{\kappa \cdot \ell \cdot r}  (4^{w_c}e^{-\mu})^{r}
      \\& \leq |S'_+| 2 \frac{4^{\wc} e^{-(\mu-\kappa l)} }{1-4^{\wc} e^{-(\mu-\kappa l)}} ||M||_\mu
      \\& \leq |S'_+| \,||M||_\mu
\end{split}\end{equation}
which ends the proof.
\end{proof}
We note that it is here that we use that $H_0$ is a stabilizer model.

\subsection{A local solution of equation \eqref{eq: basic equation a} }
We consider a sequence 
\begin{equation}\begin{split}
\bfs_{i},\qquad i=0,\ldots, k
\end{split}\end{equation}
and, whenever it is well-defined, 
\begin{equation}\begin{split}
\bfs'_{i}= \bfs_{i}\bfs'_{i-1}  ,\qquad i=1,\ldots, k,\qquad   \bfs'_{0}=\bfs_{0}
\end{split}\end{equation}
We say that $\bfs_i$ is a ghost whenever $(S_i)_{\neg g} =\emptyset$, where $(S_{i})_{\neg g}=S_i\setminus (S_{i})_{g}$. 
\begin{definition}\label{def: admissible}
We say a sequence $\bfs_{i=0,\ldots,k}$ is admissible iff.\
\begin{enumerate}
\item $\bfs'_i$, for $i=0,\ldots,k$, are well-defined.
\item $\bfs_i$ is not a ghost.
\item   $(S_i)_e \cup (S_i)_- \neq \emptyset$.  (i.e.\ $\bfs_i$ gives zero when right multiplied by $P$)
\item $(S'_i)_e=(S'_i)_-=\emptyset$. (i.e. $\bfs'_i$ does not give zero when right multiplied by $P$)
\end{enumerate}    
\end{definition}

We also need the following definition:
\begin{equation}\label{eq: def delta} \begin{split}
\Delta(\bfs')= \sum_{\text{ghost} \,  \bfs: S \cap S' \neq \emptyset}   \frac{1}{|S_+'|} q({\bfs',\bfs}).
\end{split}\end{equation}
with $q(\cdot,\cdot)$ as defined previously, starting from $M$. 
We can then state the following lemma, which is for the moment formal as it depends on absolute convergence of the series in $k$ in a norm $||\cdot ||_{\mu}$. Such convergence will be established in Section \ref{sec: locality of a}.
\begin{lemma}\label{lem: explicit expression a}
If $\mu > \mu_*$ and $||M||_{\mu} \leq  1$, then the series 
\begin{equation}\begin{split}
A^+ &=   \sum_{k=0}^\infty (-1)^k\sum_{\substack{\bfs_0,\ldots,\bfs_k \\  \text{admissible}}} (1+\Delta(\bfs_{k}'))^{-1} \frac{1}{|(S'_k)_+|}  D_{\bfs_k} (1+\Delta(\bfs_{k-1}'))^{-1}
\ldots \\
  &\qquad \ldots \frac{1}{|(S'_2)_+|} D_{\bfs_2}  (1+\Delta(\bfs_1'))^{-1} \frac{1}{|(S_1')_+|} D_{\bfs_1} (1+\Delta(\bfs_0'))^{-1} \frac{1}{|(S_0')_+|} V^+_{\bfs_0}  \,
\end{split}\end{equation}
formally satisfies the equation \eqref{eq: basic equation a}.
\end{lemma}
% \vedika{We don't think this expression agrees with the alternate representation in (16) later. We think this is the right expression:
% \begin{align*}
% A^+ &=   \sum_{k=0}^\infty (-1)^k\sum_{\substack{\bfs_0,\ldots,\bfs_k \\  \text{admissible}}} (1+\Delta((S_{k}')_+))^{-1} \frac{1}{|(S'_k)_+|}  D_{\bfs_k} (1+\Delta((S_{k-1}')_+))^{-1}
% \ldots  \\
%   &\qquad \ldots \frac{1}{|(S'_2)_+|} D_{\bfs_2}  (1+\Delta((S_1')_+))^{-1} \frac{1}{|(S_1')_+|} D_{\bfs_1} (1+\Delta((S_0')_+))^{-1} \frac{1}{|(S_0')_+|} V^+_{\bfs_0}  \,
% \end{align*}
% }

\subsection{Another representation for $A^+$}

To check Lemma \ref{lem: explicit expression a}, it is convenient to re-expand the summands in Lemma \ref{lem: explicit expression a}, and for that we need another definition:
\begin{definition}\label{def: admissible ghost}
We say a sequence $\bfs_0,\ldots,\bfs_k$is ghost-admissible whenever the following holds. 
\begin{enumerate}
    \item $\bfs_0$ is not a ghost.
\item Let $\widetilde{\bfs}_{i=0,\ldots,l}$ be the restricted sequence obtained by dropping all ghosts from $\bfs_{0,\ldots,k}$ and relabeling so that the index set is contiguous. Then   $\widetilde{\bfs}_{i=0,\ldots,l}$ is admissible.  
\end{enumerate}
To a ghost-admissible sequence $\bfs_0,\ldots,\bfs_k$, we associate the sequence $\bfs'_i$ as follows
\begin{equation}
    \bfs'_i=\begin{cases}
        \bfs_0  &  i=0 \\
        \bfs_i\bfs'_{i-1}  & i>0 \quad \text{and} \quad \bfs_i \, \text{is not a ghost} \\
        \bfs'_{i-1}  & i>0 \quad \text{and} \quad \bfs_i \, \text{is a ghost} \\
    \end{cases}
\end{equation}
\end{definition}
We can now give the alternative form of the series for $A^+$, namely
\begin{equation}\label{eq: alternative}
\begin{split}
    A^+ &=    \sum_{k=0}^\infty A^+_k
    \\ A^+_k &= \sum_{\substack{\bfs_0,\ldots,\bfs_k \\  \text{ghost-admissible}}}     (A^+_k)_{\bfs'_k}
\end{split}
\end{equation}
with 
\begin{equation}\begin{split}
 (A^+_k)_{\bfs'_k} &=  
\frac{i}{|(S'_k)_+|}  \widetilde D_{\bfs_k} 
 \ldots \frac{1}{|(S'_2)_+|}   
  \widetilde D_{\bfs_2} \frac{1}{|(S_1')_+|}  \widetilde 
 D_{\bfs_1}  \frac{1}{|(S_0')_+|} V^+_{\bfs_0}  \,   \label{eq: a rearrange}
\end{split}\end{equation}
% \vedika{RHS multiplied by $i$ }. 
where 
\begin{equation}\begin{split}
    \widetilde D_{\bfs_i}   
=
\begin{cases}
     -D_{\bfs_i}  & \bfs_i \,  \text{is not a ghost }  \\
      -q({\bfs'_{i-1}},\bfs_i) & \bfs_i \,  \text{is a ghost}
\end{cases}.
\end{split}\end{equation}
By the conditions on $\mu$, we derive from Lemma \ref{lem: bound on ghosts} that $\Delta(\bfs')<1$ for any $\bfs'$. 
We can now resum the sum of ghost terms in \eqref{eq: alternative} between any pair of consecutive non-ghost terms, as $\sum_{j=0}^\infty (-\Delta(\bfs'))^j=(1+\Delta(\bfs'))^{-1}$.  This shows that the above series is the same as the one for $A^+$ given in Lemma \ref{lem: explicit expression a}, provided that both converge absolutely, c.f. remark before Lemma \ref{lem: explicit expression a}.

It remains to check that \eqref{eq: a rearrange} satisfies the equation \eqref{eq: basic equation a}.   This happens order by order, i.e.\ we claim
\begin{equation} \label{eq: higher k}
    \overline P ([H_0,A^+_{k}] + [D,A^+_{k-1}] + [M,A^+_{k-1}] ) P=0, \qquad k>0
\end{equation}
and 
\begin{equation} \label{eq: zero k}
\overline P (i[H_0,A^+_{0}] +V^+)  P=0,
\end{equation}
% I don't comment on the latter equation, which is easy to verify.  
The latter equation is directly verified from \eqref{eq: a rearrange} and we focus on the former. 
From here on, the index $k$ is traded for $m$.
A first observation (from inspecting \eqref{eq: alternative} and using the definition of $\Delta$) is that
\begin{equation}
(A^+_k )_{\bfs_k'}=  \frac{1}{|(S'_{k})_+|}\left(\sum_{\bfs'_{k-1}, \bfs_k:   {\bfs'_k}=\bfs_k{\bfs'_{k-1}}} \quad
  D_{\bfs_k} (A^+_{k-1})_{\bfs_{k-1}'} \right)  -  
  (A^+_{k-1})_{\bfs_{k}'}  \Delta(\bfs'_{k})
\end{equation}
where the sequence is admissible, i.e. $\bfs_k,\bfs'_{k-1},\bfs_k'$ satisfy the conditions of Definition~\ref{def: admissible}.

To check that equation \eqref{eq: higher k} is satisfied, we observe
\begin{enumerate}
\item 
\begin{equation}\begin{split}
\overline P [M, (A^+_{k-1})_{\bfs_{k-1}'} ] P 
    & =   
 \overline P (A^+_{k-1})_{\bfs_{k-1}'}  P \times |(S'_{k-1})_+|\Delta(\bfs'_{k-1})\\
   & =   
 \overline P (A^+_{k-1})_{\bfs_{k}'}  P \times |(S'_{k})_+|\Delta(\bfs'_{k})
\end{split}\end{equation}
where, in the last equation, we used the convention $\bfs'_k=\bfs'_{k-1}$ whenever $\bfs_k$ is a ghost.
\item
\begin{equation}\begin{split}
\overline P [D, (A^+_{k-1})_{\bfs_{k-1}'} ] P & =\sum_{\bfs_k} 
 \overline P D_{\bfs_k} (A^+_{k-1})_{\bfs_{k-1}'}  P  
\end{split}\end{equation}
because all terms in $D$ give zero when right multiplied with $P$.  The important thing here is to realize that on the right hand side, one can restrict to $\bfs_k,\bfs'_{k-1}$ such that the sequence is admissible.

\item  
\begin{equation}\begin{split}
    \overline P ([H_0,(A^+_{k})_{\bfs'_k}]) P  &=  
|(S'_k)_+|\,  \overline P  (A^+_{k})_{\bfs'_k} P 
\end{split}\end{equation}
\end{enumerate}
With these observations, we finally confirm the validity of equation \eqref{eq: higher k}.

\section{Locality of A} \label{sec: locality of a}

We will establish locality of the operator collection $A$, which was constructed in Section~\ref{sec:construction_of_A}. As explained before, it suffices to consider ${A^+}$. 
We will prove 
\begin{proposition}\label{lem: locality of a}
Provided that $\mu$ satisfies $\mu>\mu_*>1$, and 
$e||D||_\mu -||M||_\mu <1$, we have
\begin{equation}\begin{split}
||A^+||_\mu \leq   ||V^+||_\mu \frac{1-||M||_\mu}{1- e||D||_\mu -||M||_\mu}.
\end{split}\end{equation}
\end{proposition}

\subsection{Preliminary considerations}
From Lemma \ref{lem: bound on ghosts} and \ref{lem: explicit expression a}, we get 
\begin{equation}\begin{split} \label{eq: general expression bound a}
 ||A^{+} ||_\mu &\leq 
\sum_{k=0}^\infty (1-||M||_\mu)^{-k} \bigg[ \\&
\sup_\alpha \sum_{\substack{\bfs_0,\ldots,\bfs_k \\  \text{admissible} \\   \alpha \in S'_{k} }}
\frac{1}{|(S'_k)_+|} ||D_{\bfs_k}||
\ldots 
   \frac{1}{|(S'_2)_+|} ||D_{\bfs_2}|| \frac{1}{|(S_1')_+|}  ||D_{\bfs_1}||  \frac{1}{|(S_0')_+|} ||V^+_{\bfs_0}|| \,
   e^{\mu |S'_k|}\bigg]
\end{split}\end{equation}
We now decompose the size $|S'_k|$ by using the enhanced subadditivity 
\eqref{eq: enhanced subadditivity}, which yields 
$$
|S'_k| \leq \sum_{i=0}^k |S_i| -  \sum_{i=1}^k |S_i \cap S'_{i-1}|
$$
Using this, we can bound the summand in \eqref{eq: general expression bound a} by
\begin{equation}\begin{split}
 z(\bfs_0,\ldots,\bfs_k) \coloneqq \frac{   e^{-\mu |S_k\cap S'_{k-1}|}  }{|(S'_k)_+|} d_{\bfs_k} 
\ldots 
 \frac{ e^{-\mu |S_1\cap S'_0|} }{|(S_1')_+|}  d_{\bfs_1}  \frac{1 }{|(S_0')_+|} v_{\bfs_0} 
 % =: z(\bfs_0,\ldots,\bfs_k)
\end{split}\end{equation}
where we tried to reduce clutter by abbreviating 
\begin{equation}\begin{split}
d_{\bfs_i}= ||D_{\bfs_i}|| e^{\mu |S_i|}, \qquad v_{\bfs_0}= ||V^+_{\bfs_0}|| e^{\mu |S_0|}.
\end{split}\end{equation}
So we have to bound 
\begin{equation}\begin{split}
z(k):=\sup_\alpha
\sum_{\substack{\bfs_0,\ldots,\bfs_k \\  \text{admissible} \\   \alpha \in S'_{k} }}
z(\bfs_0,\ldots,\bfs_k) 
\end{split}\end{equation}
We resolve the condition  $x\in S'_k$ by requiring that $x\in S_j$ and summing over $j$. The sum with the condition $x\in S_j$ is denoted by
\begin{equation}\begin{split}
z(k) \leq \sum_{j=0}^k \sup_\alpha 
\sum_{\substack{\bfs_0,\ldots,\bfs_k \\  \text{admissible} \\   \alpha \in S_{j} }}
z(\bfs_0,\ldots,\bfs_k) = : \sum_{j=0}^k z_{k,j}.
\end{split}\end{equation}

Here is what we will actually prove
\begin{lemma}\label{lem: combi} For any $0<j\leq k$,
\begin{equation}\begin{split}
z_{k,j}\leq z(j-1) ||D||^{k-j+1}_{\mu}.
\end{split}\end{equation}
For $j=0$, we have 
\begin{equation}\begin{split}
z_{k,0}\leq  ||D||^{k}_{\mu} ||V^+||_{\mu}.
\end{split}\end{equation}
\end{lemma}
Lemma \ref{lem: combi} is proven below. For now, we use it to give the 
\begin{proof}[Proof of Proposition \ref{lem: locality of a}]
Using $z(k)\leq \sum_{j=0}^k z_{k,j}$ and Lemma \ref{lem: combi}, we can inductively bound $z(k)$ as 
\begin{equation}\begin{split}
\frac{z(k)}{||D||^{k}_{\mu} ||V^+||_\mu } \leq  \sum_{p=1}^k  \frac{k^p}{p!} \leq e^k.
\end{split}\end{equation}
Then, provided that $e||D||_\mu +|| M ||_\mu <1
$, we estimate 
\begin{equation}\begin{split}
|| A^+||_\mu \leq \sum_{k=0}^\infty e^k  (\frac{||D||_\mu}{1-||M||_\mu})^{k} ||V^+||_\mu \leq  ||V^+||_\mu \frac{1-||M||_\mu}{1- e||D||_\mu -||M||_\mu}.
\end{split}\end{equation}
\end{proof}.

\subsection{Proof of Lemma \ref{lem: combi}}
\subsubsection{Some preparatory work}
We now state a few remarks that will be used.  They are specific to the sets $\bfs_0,\ldots, \bfs_k$ appearing in our construction. In particular, we assume that the sequence is admissible. 
We start with two anchoring equations:
\begin{equation}\label{eq: forward constraint}
S_{i}\cap (S'_{i-1})_+  \neq \emptyset
\end{equation}
and 
\begin{equation}\label{eq: backwards constraint}
   (S_{i})_{\neg g}\cap S'_{i-1}  \neq \emptyset 
\end{equation}
Both follow in a rather direct way from the defintion of an admissible sequence. The first equation will be used to anchor $S_i$ in $(S'_{i-1})_+ $ and it will be used often.
The second equation is used to anchor $S'_{i-1}$ in $(S_{i})_{\neg g}$ and it will only be used at the end of subsection \ref{sec: argument for j larger}.

\begin{lemma}\label{lem: set transfer one}
If $\mu>1$, then
\begin{equation}\begin{split}
\frac{  e^{-\mu |S_k\cap S'_{k-1}|}  }{|(S'_k)_+|}  \leq \frac{1}{|(S_k)_{\neg g}|}
\end{split}\end{equation}
\end{lemma}
\begin{proof}
We note that
\begin{equation}\begin{split}
(S_k)_{\neg g}\subseteq (S'_k)_{+}  \cup (S_k\cap S'_{k-1})
\end{split}\end{equation}
Indeed,  $S_k \setminus S'_{k-1} $ is a subset of $S'_k$ and hence it consists of  $+$ or $g$. 
Therefore, for $\mu > 1$ we have
\begin{equation}\begin{split}
\frac{  e^{-\mu |S_k\cap S'_{k-1}|}  }{|(S'_k)_+|}     |(S_k)_{\neg g}| \leq 
\frac{|(S'_k)_{+}| +|S_k\cap S'_{k-1}|}{|(S'_k)_{+}|}   e^{-\mu |S_k\cap S'_{k-1}|}  \leq 1
\end{split}\end{equation}
and hence the claim follows.
\end{proof}
\begin{lemma}\label{lem: set transfer two}
If $\mu>1$, then
\begin{equation} \label{eq: shift of splus}
    \frac{e^{-\mu |S'_{i-1} \cap S_i| }}{|(S'_{i})_+|} \leq   \frac{1}{|(S'_{i-1})_+|} 
\end{equation}
\end{lemma}
\begin{proof}
From $\bfs'_{i}=\bfs_{i}\bfs'_{i-1}$, we get
\begin{equation}\begin{split}
(S'_{i-1})_+ \subseteq  (S'_{i})_+ \cup (S'_{i-1} \cap S_i).
\end{split}\end{equation}
This directly implies the statement of the lemma.
\end{proof}

\subsubsection{The argument for $j=0$}\label{sec: argument for j is zero}
We recall that we are estimating
\begin{equation}\begin{split}
 z_{k,0}= \sup_\alpha 
\sum_{\substack{\bfs_0,\ldots,\bfs_k \\  \text{admissible} \\   \alpha \in S_{0} }}
z(\bfs_0,\ldots,\bfs_k).
\end{split}\end{equation}
This case is special in that we can afford to bound 
\begin{equation}\label{eq: crude bound}
z(\bfs_0,\ldots,\bfs_k) \leq \widetilde z(\bfs_0,\ldots,\bfs_k)  :=    d_{\bfs_k}   \frac{ 1 }{|(S_{k-1}')_+|}   d_{\bfs_{k-1}}
\ldots 
 \frac{ 1 }{|(S_1')_+|}  d_{\bfs_1}  \frac{1 }{|(S_0')_+|} v_{\bfs_0}.
\end{equation}
That is, we dropped the exponential factors in $z(\ldots)$ and we also dropped the left-most denominater $\frac{ 1 }{|(S_{k}')_+|}$.    
To perform the sum over admissible sequences $\bfs_0,\ldots, \bfs_k$, we note that 
\begin{equation}\begin{split}
S_{i+1} \cap (S'_{i})_+ \neq \emptyset.
\end{split}\end{equation}
We use this to ``anchor'' the sets $S_{i+1}$ in $(S'_{i})_+$. Concretely, starting with $i=k-1$, we rewrite
\begin{equation}\begin{split}
\sum_{\bfs_k:  S_{k} \cap (S'_{k-1})_+ \neq \emptyset }
d_{\bfs_k}  \leq \sum_{\beta \in (S'_{k-1})_+ }  
\sum_{\bfs_k:  S_{k} \ni \beta}         d_{\bfs_k}  \leq 
\sum_{\beta \in (S'_{k-1})_+ } || D||_{\mu}  \leq 
|(S'_{k-1})_+| || D||_{\mu}.
\end{split}\end{equation}
Hence we have obtained 
\begin{equation}\begin{split}
\sum_{\substack{\bfs_0,\ldots,\bfs_k \\  \text{admissible} \\   \alpha \in S_{0} }}
\widetilde z(\bfs_0,\ldots,\bfs_k) \leq    || D||_{\mu}  
\sum_{\substack{\bfs_0,\ldots,\bfs_{k-1}\\  \text{admissible} \\   \alpha \in S_{0} }}
\widetilde z(\bfs_0,\ldots,\bfs_{k-1}).
\end{split}\end{equation}
We can iterate this argument, and this yields the case $j=0$ in Lemma \ref{lem: combi}. 

\subsubsection{The argument for $j>0$} \label{sec: argument for j larger}
We start from 
\begin{equation}\begin{split} \label{eq: general case k anchoring for d}
z_{k,j} =                             \sup_\alpha  \sum_{\bfs_j: S_j \ni \alpha} d_{\bfs_j} 
\sum_{\substack{\bfs_0,\ldots,\bfs_{j-1} }}   \frac{  e^{-\mu |S_j\cap S'_{j-1}|}  }{|(S'_j)_+|}    z(\bfs_0,\ldots,\bfs_{j-1})   
\sum_{\substack{\bfs_{j+1},\ldots,\bfs_{k}}} 
\prod_{t=j+1}^k\frac{  e^{-\mu |S_t\cap S'_{t-1}|}  }{|(S'_t)_+|} d_{\bfs_t}  
\end{split}\end{equation}
where of course the sums are still constrained by the requirement that $\bfs_0,\ldots,\bfs_k$ is admissible. 
We now use Lemma \ref{lem: set transfer one} with $i=j$ and we use Lemma \ref{lem: set transfer two} with $i=j+1,\ldots,k$. We obtain then 
\begin{equation}\begin{split} \label{eq: general case next}
z_{k,j} \leq
\sup_\alpha  \sum_{\bfs_j: S_j \ni \alpha}  d_{\bfs_j} 
\sum_{\substack{\bfs_0,\ldots,\bfs_{j-1}  
  }}   \frac{ 1  }{|(S_j)_{\neg g}|}    z(\bfs_0,\ldots,\bfs_{j-1})   
\sum_{\substack{\bfs_{j+1},\ldots,\bfs_{k} 
}} 
\prod_{t=j+1}^k\frac{  1  }{|(S'_{t-1})_+|} d_{\bfs_t}    
% \label{eq: last}
\end{split}\end{equation}
% \begin{equation}\begin{split} \label{eq: general case next}
% z_{k,j} &\leq                              \sup_x  \sum^x_{\bfs_j}  d_{\bfs_j} 
% \sum_{\substack{\bfs_0,\ldots,\bfs_{j-1} \\ (S_{j})_{\neg g}\cap S'_{j-1})  \neq \emptyset}}   \frac{ 1  }{|(S_j)_{\neg g}|}    z(\bfs_0,\ldots,\bfs_{j-1})   
% \sum_{\substack{\bfs_{j+1},\ldots,\bfs_{k} \\ S_{j+1} \cap (S'_{j})_+ \neq \emptyset}} 
% \prod_{t=j+1}^k\frac{  1  }{|(S'_{t-1})_+|} d_{\bfs_t}   \\
%       &\leq    \sup_x  \sum^x_{\bfs_j}  d_{\bfs_j} 
% \sum_{\substack{\bfs_0,\ldots,\bfs_{j-1} \\ (S_{j})_{\neg g}\cap S'_{j-1})  \neq \emptyset}}   \frac{ 1  }{|(S_j)_{\neg g}|}    z(\bfs_0,\ldots,\bfs_{j-1})   
% \sum_{\substack{\bfs_{j+1},\ldots,\bfs_{k} \\ S_{j+1} \cap (S'_{j})_+ \neq \emptyset}} 
% \prod_{t=j+1}^k\frac{  d_{\bfs_t}   }{|(S'_{t-1})_+|}   \\
%  &\leq    \sup_x  \sum^x_{\bfs_j}  \frac{ d_{\bfs_j}   }{|(S_j)_{\neg g}|}  \sum_{y \in (S_j)_{\neg g}}
% \sum^y_{\bfs_0,\ldots,\bfs_{j-1} }     z(\bfs_0,\ldots,\bfs_{j-1})   \sum_{y'\in (S_j')_+}
% \sum^{y'}_{\substack{\bfs_{j+1},\ldots,\bfs_{k} }} 
% \prod_{t=j+1}^k\frac{  d_{\bfs_t}   }{|(S'_{t-1})_+|}  \\
%  &\leq    \sup_x  \sum^x_{\bfs_j}  \frac{ d_{\bfs_j}   }{|(S_j)_{\neg g}|}  \sum_{y \in (S_j)_{\neg g}}
% z(j-1) \sum_{y'\in (S_j')_+}
% \sum^{y'}_{\substack{\bfs_{j+1},\ldots,\bfs_{k} }} 
% \prod_{t=j+1}^k\frac{  d_{\bfs_t}   }{|(S'_{t-1})_+|}  
% \end{split}\end{equation}
We first deal with the last sum, i.e.
\begin{equation}\begin{split}
\label{eq: last}
\sum_{\substack{\bfs_{j+1},\ldots,\bfs_{k} }} 
\prod_{t=j+1}^k\frac{  1  }{|(S'_{t-1})_+|} d_{\bfs_t}  
\end{split}\end{equation}
Here we use the strategy that was already explained and used in subsection \ref{sec: argument for j is zero}: We sum first over $\bfs_k$ which is anchored in $(S'_{k-1})_{+}$. The factor $\frac{1}{|(S'_{k-1})_{+}|}$ controls the choice of anchoring points, and we get $||D||_{\mu}$ times the rest. We iterate this and we obtain in this way
\begin{equation}
\sum_{\substack{\bfs_{j+1},\ldots,\bfs_{k} }} 
\prod_{t=j+1}^k\frac{  1  }{|(S'_{t-1})_+|} d_{\bfs_t}    \leq   ||D||_{\mu}^{k-j}.
\end{equation}
Plugging this into \eqref{eq: last}, we get
\begin{equation}\begin{split} \label{eq: general case next next}
z_{k,j} &\leq   ||D||_{\mu}^{k-j}   \sup_\alpha  \sum_{\bfs_j: \alpha \in S_j}  d_{\bfs_j} 
\sum_{\substack{\bfs_0,\ldots,\bfs_{j-1} }}   \frac{ 1  }{|(S_j)_{\neg g}|}    z(\bfs_0,\ldots,\bfs_{j-1})      \\
&\leq    ||D||_{\mu}^{k-j}  \sup_\alpha  \sum_{\bfs_j: \alpha \in \bfs_j}    d_{\bfs_j}  \sum_{\beta \in (S_j)_{\neg g}}   \frac{  1 }{|(S_j)_{\neg g}|} 
z(j-1)  \\
&\leq    ||D||_{\mu}^{k-j}   z(j-1)    \sup_\alpha  \sum_{\bfs_j: \alpha \in \bfs_j}    d_{\bfs_j}  \\
& \leq    ||D||_{\mu}^{k-j+1}   z(j-1). 
\end{split}\end{equation}
To obtain the second inequality, we used the anchoring equation \eqref{eq: backwards constraint}, i.e. $(S_{j})_{\neg g}\cap S'_{j-1}  \neq \emptyset$.
This concludes the proof.

\section{Controlling the transformation generated by $A$}\label{sec: controlling_rotation}

We estimate the operator-collections  defined by repeated applications of commutators.

\begin{proposition} \label{lem: commutator}
Let $B^{(0)},B^{(1)}, \ldots B^{(k)}$ be operator-collections and let $\mu'<\mu$, 
% and let 
% $$
% \frac{||A ||_{\mu}}{|\mu-\mu'|^k}   \leq c,
% $$
then
\begin{equation}\begin{split}
|| \mathrm{ad}_{B^{(k)}} \ldots \mathrm{ad}_{B^{(1)}}(B^{(0)})||_{\mu'}  \leq  \frac{(k+1)! (2e)^k}{|\mu-\mu'|^{k+1}}   \prod_{t=0}^k || B^{(t)}||_{\mu}.
\end{split}\end{equation}
\end{proposition} 
This result will in particular be useful to control exponentials: indeed, if 
\begin{equation}\begin{split}
\frac{(2e) ||B^{(1)}||_  \mu}{|\mu-\mu'|}\leq 1/2,
\end{split}\end{equation}
then it follows that
\begin{equation}\label{eq: effect of exp first}
||\exp{i\mathrm{ad}_{B^{(1)}}} (B^{(0)})-B^{(0)}||_{\mu'} \leq \frac{2 ||B^{(0)}||_{\mu}  }{|\mu-\mu'|} \sum_{k=1}^{\infty} (\frac{(2e) ||B^{(1)}||_{\mu}}{|\mu-\mu'|}   )^k \leq \frac{8e}{|\mu-\mu'|^2} ||B^{(0)}||_{\mu}||B^{(1)}||_{\mu}   
\end{equation}
and similarly,
\begin{equation}\label{eq: effect of exp zeroth}
||\exp{i\mathrm{ad}_{B^{(1)}}} (B^{(0)})||_{\mu'} \leq  ||B^{0}||_\mu (1+\frac{8e ||B^{(1)}||_{\mu} }{|\mu-\mu'|^2}  ) 
\end{equation}

\subsection{Proof of Proposition \ref{lem: commutator}}
 We split 
\begin{equation}\begin{split}
 \mathrm{ad}_{B^{(k)}} \ldots \mathrm{ad}_{B^{(1)}}(B^{(0)})=\sum_{\bfs_0,\ldots,\bfs_k} [B^{(k)}_{\bfs_k}, \ldots   \ldots[B^{(1)}_{\bfs_1},B^{(0)}_{\bfs_0}] \ldots]
\end{split}\end{equation}
If the condition
\begin{equation}\label{eq: connected supports}
S_{i+1} \cap (S_0\cup S_1\cup \ldots S_{i}) \neq \emptyset,\qquad \forall i=0,\ldots,k-1
\end{equation}
% \YL{Here the combinatorics is similar to the previous section, but slightly simpler, in that we only need connectedness, but not admissible.}
is not satisfied, then the corresponding term in the sum vanishes. 
Therefore, we assume that this condition holds and we expand all commutators $[O_1,O_2]$ into the two terms $+O_1O_2$ and $-O_2O_1$. 
This gives hence $2^k$ terms. For the sake of concreteness, we focus on the expression where we always choose the term with $'+'$. The other terms are estimated analogously.  Hence, we define again, for $i>0$
\begin{equation}
\bfs'_{i} = \mathbf{S}_i \bfs'_{i-1}, \qquad    \bfs'_{0}=  \bfs_{0} 
\end{equation}
and we investigate 
\begin{equation}\begin{split}
 ||\sum_{\bfs_0,\ldots,\bfs_k}  B^{(k)}_{\bfs_k}\ldots B^{(1)}_{\bfs_1}B^{(0)}_{\bfs_0}||_{\mu'} 
   & \leq  \sup_\alpha \sum_{\substack{\bfs_0,\ldots,\bfs_k \\ \alpha \in S'_{k} }}  e^{\mu' |S'_k|}
     ||B^{(k)}_{\bfs_k}||   \ldots ||B^{(1)}_{\bfs_1}|| ||B^{(0)}_{\bfs_0}||  \\
      & \leq   \frac{(k+1)!}{|\mu-\mu'|^{k+1}}   \sup_\alpha \sum_{\substack{\bfs_0,\ldots,\bfs_k \\ \alpha \in S'_{k} }}   z(\bfs_0,\ldots,\bfs_k) \label{eq: last sighting facto}
\end{split}\end{equation}
where we used $e^{-X} X^{k+1} \leq (k+1)!$ for $X>0$, and we 
abbreviated 
\begin{equation}
z(\bfs_0,\ldots,\bfs_k)= \frac{1}{|S_k'|^{k+1}}   b^{(k)}_{\bfs_k}\ldots b^{(0)}_{\bfs_0},\qquad  b^{i}_{\bfs_i}= ||B^{(i)}_{\bfs_i}|| e^{\mu |S_i|}  
\end{equation}
We define
% \YL{Consider different letter than $z$, to contrast $z(k)$ in Sec. 5}
\begin{equation}\begin{split}
  z(k) =&   \sup_\alpha \sum_{\substack{\bfs_0,\ldots,\bfs_k \\ \alpha \in S'_{k} }} z(\bfs_0,\ldots,\bfs_k) 
 \end{split}\end{equation} 
 and 
\begin{equation}\begin{split}
 z(k)  \leq \sum_{j=1}^{k} z_{k,j} \equiv \sum_{j=1}^{k}  \sup_\alpha \sum_{\substack{\bfs_0,\ldots,\bfs_k \\ \alpha \in S_{j} }} z(\bfs_0,\ldots,\bfs_k) 
\end{split}\end{equation}
We will show
\begin{lemma} \label{lem: combi rotation}
For $j=0,\ldots,k$,
\begin{equation}
z_{k,j} \leq   z(j-1)   \prod_{t=j}^k || B^{(t)}||_{\mu}
\end{equation}
where $z(-1)$ is defined to be $1$. 
\end{lemma}
% \YL{In the above Lemma, $z_{j-1}$ should be $z(j-1)$.}
This will yield the proof of Proposition \ref{lem: commutator}. Indeed by an inductive argument, very analogous to the reasoning following Lemma \ref{lem: combi}, we get
\begin{equation}
z(k) \leq e^k  \prod_{t=0}^k || B^{(t)}||_{\mu}.
\end{equation}
Then, using \eqref{eq: last sighting facto} and recalling that there are $2^k$ terms, we get the Proposition.  It remains to give the proof of the Lemma.

\subsection{Proof of Lemma \ref{lem: combi rotation}}
\subsubsection{The case $j=0$}
We have to bound 
\begin{equation}\begin{split}
&\sum_{\substack{\bfs_0,\ldots,\bfs_k \\ \alpha \in S_{0} }} z(\bfs_0,\ldots,\bfs_k) \\
\leq & 
\sum_{{\bfs_0: S_{0} \ni \alpha }}  
\frac{1}{|S_k'| }  b^{(0)}_{\bfs_{1}}  
\sum_{{\bfs_1: S_{1} \cap \in S_{0}' }}  
\frac{1}{|S_k'| }  b^{(1)}_{\bfs_{1}}  \ldots 
\sum_{{\bfs_k: S_k \cap \in S_{k-1}' }}  
\frac{1}{|S_k'| }  b^{(k)}_{\bfs_k} 
\end{split}\end{equation}
To perform the rightmost sum, we dominate 
$\frac{1}{|S_k'| }  \leq \frac{1}{|S_{k-1}'| } $
and we bound the sum by $|| B^{(t)}||_{\mu}$.
% $||B_{k}||_\mu$. 
% \YL{should be $|| B^{(t)}||_{\mu}$. the reasoning here is similar to the unlabeled equation below (33)}
We continue this process interatively, obtaining $\prod_{t=0}^k || B^{(t)}||_{\mu}$.
% \YL{should be $\prod_{t=0}^k || B^{(t)}||_{\mu}$}
The leftmost factor $\frac{1}{|S_k'| } $ was not exploited and we simply bounded it by $1$.
\subsubsection{The case $j>0$}
% \wdr{still to correct errors... }
    \begin{equation}\begin{split} \label{eq: general case k anchoring for b}
z_{k,j}  \leq                           \sup_\alpha  \sum_{\bfs_j: S_j \ni \alpha} b^{(j)}_{\bfs_j} 
\sum_{\substack{\bfs_0,\ldots,\bfs_{j-1} }}   \frac{ 1 }{|S'_j|}    z(\bfs_0,\ldots,\bfs_{j-1})   
\sum_{\substack{\bfs_{j+1},\ldots,\bfs_{k}}} 
\prod_{t=j+1}^k\frac{ 1 }{|S'_t|} b^{(t)}_{\bfs_t}  
\end{split}\end{equation}
where we bounded $\frac{1}{|S_k'|} \leq  \frac{1}{|S_j'|}  $ and  $\frac{1}{|S_k'|} \leq  \frac{1}{|S_t'|}  $.
We now use the same reasoning as for case $j=k$ to bound the last sum;
we obtain
\begin{equation}\begin{split} 
% \label{eq: general case k anchoring}
z_{k,j}  & \leq                           \sup_\alpha  \sum_{\bfs_j: S_j \ni \alpha} b^{(j)}_{\bfs_j} \sum_{\beta \in S_j}
\sum_{\substack{\bfs_0,\ldots,\bfs_{j-1} \\  \beta \in S_{j-1}' }}   \frac{ 1 }{|S'_j|}    z(\bfs_0,\ldots,\bfs_{j-1})   
 \prod_{t=j+1}^k || B^{(t)}||_{\mu} \\
  & \leq                           \sup_\alpha  \sum_{\bfs_j: S_j \ni \alpha} b^{(j)}_{\bfs_j} 
z(j-1)  
 \prod_{t=j+1}^k || B^{(t)}||_{\mu} \\
 & \leq                         || B^{(j)}||_\mu
z(j-1)  
 \prod_{t=j+1}^k || B^{(t)}||_{\mu}
 \leq
 z(j-1)  
 \prod_{t=j}^k || B^{(t)}||_{\mu}.
\end{split}\end{equation}
% \YL{in the above equation, first line, third summation sign, $S'_j$ should be $S'_{j-1}$.
% The overall reasoning here is similar to (36-38), particularly using that $\mathbf{S}_j \cap S'_{j-1} \neq \emptyset$.
% }

This concludes our proof of Lemma~\ref{lem: combi rotation}, which we used in our proof of Proposition~\ref{lem: commutator}. As noted above, Proposition~\ref{lem: commutator} is useful for controlling the locality of an operator after a KAM step and for controlling the difference between an operator and its rotation under a KAM step. This control will allow us to ensure that our operator collections still have good intensive word norms $||\cdot ||_{\mu_n}$ for the appropiate $\mu_n$ after the KAM steps. It is also important in allowing us to show that $\epsilon_n \coloneqq || V^{(n)}||_{\mu_n}$ is rapidly decreasing in $n$.

\section{Proof of relative boundedness (Proposition~\ref{prop: relative boundedness})} \label{sec:proof_rel_bounded}

For any operators $A$ and $B$ acting on the same Hilbert space $\mathcal{H}$, we use the notation $B \rbleq A$, read ``$B$ is relatively bounded by $A$," to mean that 
\begin{equation}
    \forall |\psi\rangle \in \mathcal{H}, \,\,\, ||  B \ket{\psi}|| \leq ||A \ket{\psi}|| 
\end{equation}

For $c_1,c_2 \geq 0$, if $B_1 \rbleq c_1 A$ and $B_2 \rbleq c_2 A$, then $B_1 + B_2 \rbleq (c_1+c_2)A$ by triangle inequality. Additionally, suppose that $\{A_i\}$ are a set of Hermitian positive semi-definite operators that mutually commute. Then for sets of numbers $\{c_i\}$ and $\{c'_i\}$ with $0 \leq c_i \leq c'_i$, $\sum_i c_i A_i \rbleq \sum_i c'_i A_i$. We will use these facts freely in the following.

It is useful to define some more notation. Let $X$ be a subset of checks, and define the projector $$P(X) = \prod_{\alpha \in X} E_{\alpha} \prod_{\beta \notin X} G_\beta.$$ Summing over all subsets of checks gives the identity, $\sum_X P(X) = 1$. The idea of $P(X)$ is that it projects onto excited stabilizer generators at the check locations contained in $X$.

We will first relatively bound $D^{(n)}$ following a method inspired by section 3.2 of \cite{bravyi2010topological}.

\begin{equation}
\begin{split}
\langle \psi| (D^{(n)})^{\dagger} D^{(n)} |\psi \rangle 
&= \sum_{\mathbf{S}} \sum_{\mathbf{S'}} \langle \psi| (D^{(n)}_{\mathbf{S}})^{\dagger} D^{(n)}_{\mathbf{S'}} |\psi \rangle
\\&= \sum_{X,Y,Z} \sum_{\mathbf{S}} \sum_{\mathbf{S'}} \langle \psi|P(X) (D^{(n)}_{\mathbf{S}})^{\dagger}P(Y)  D^{(n)}_{\mathbf{S'}} P(Z)|\psi \rangle
\\&\leq \sum_{X,Y,Z}  \sum_{\mathbf{S}} \sum_{\mathbf{S'}} ||P(Y) D^{(n)}_{\mathbf{S}} P(X)||\, ||P(Y)  D^{(n)}_{\mathbf{S'}} P(Z)||\, ||P(X)|\psi\rangle||\, ||P(Z)|\psi\rangle||
\\&\leq \sum_{X,Y,Z}  \sum_{\mathbf{S}} \sum_{\mathbf{S'}} ||P(Y) D^{(n)}_{\mathbf{S}} P(X)||\, ||P(Y)  D^{(n)}_{\mathbf{S'}} P(Z)||\, \frac{\langle \psi|P(X)|\psi\rangle+ \langle \psi|P(Z)|\psi\rangle}{2}
\\&= \sum_{X,Y,Z}  \sum_{\mathbf{S}} \sum_{\mathbf{S'}} ||P(Y) D^{(n)}_{\mathbf{S}} P(X)||\, ||P(Y)  D^{(n)}_{\mathbf{S'}} P(Z)||\, \langle \psi|P(X)|\psi\rangle
\end{split}
\end{equation}

For many choices of $X,Y,Z, \mathbf{S}, \mathbf{S'}$, the summand vanishes. For a given choice of $X, \mathbf{S}, \mathbf{S'}$, there is \textit{at most one} choice for $Y$ and $Z$ such that both $||P(Y) D^{(n)}_{\mathbf{S}} P(X)||$ and $||P(Y)  D^{(n)}_{\mathbf{S'}} P(Z)||$ are nonvanishing. For $Y$, the choice is $Y(X,\mathbf{S}) := S_+ \cup (X/S_-)$. For $Z$, the choice is $Z(X,\mathbf{S},\mathbf{S}') := S'_- \cup  (Y(X,\mathbf{S})/(S_+))$. 

Given the structure of $D^{(n)}_{\mathbf{S}}$ with $(S_- \cup S_e)$ and $(S_+ \cup S_e)$ both non-empty, it is also necessary that $S \cap X \neq \emptyset$ and $S' \cap Y \neq \emptyset$ in order that $D^{(n)}_{\mathbf{S}} P(X) \neq 0$ and $P(Y) D^{(n)}_{\mathbf{S'}} \neq 0$ respectively. Note that this precludes $X = \emptyset$.

Putting all these restrictions together,
\begin{equation}
\begin{split}
&\langle \psi| (D^{(n)})^{\dagger} D^{(n)} |\psi \rangle 
\leq \sum_{X \neq \emptyset} \langle \psi|P(X)|\psi\rangle \sum_{\mathbf{S}} ||P(Y(X,\mathbf{S})) D^{(n)}_{\mathbf{S}} P(X)|| \sum_{\mathbf{S'}} ||P(Y(X,\mathbf{S}))  D^{(n)}_{\mathbf{S'}} P(Z(X,\mathbf{S}, \mathbf{S'}))||
\\&= \sum_{X \neq \emptyset} \langle \psi|P(X)|\psi\rangle \sum_{\mathbf{S}: S \cap X \neq \emptyset} ||P(Y(X,\mathbf{S})) D^{(n)}_{\mathbf{S}} P(X)|| \sum_{\mathbf{S'}: S' \cap Y(X,\mathbf{S}) \neq \emptyset} ||P(Y(X,\mathbf{S}))  D^{(n)}_{\mathbf{S'}} P(Z(X,\mathbf{S}, \mathbf{S'}))||
\\&\leq \sum_{X \neq \emptyset} \langle \psi|P(X)|\psi\rangle  \sum_{\mathbf{S}: S \cap X \neq \emptyset} ||D^{(n)}_{\mathbf{S}}|| \sum_{\mathbf{S'}: S' \cap Y(X,\mathbf{S}) \neq \emptyset} ||D^{(n)}_{\mathbf{S'}}||
\\&\leq \sum_{X \neq \emptyset} \langle \psi|P(X)|\psi\rangle  \sum_{\mathbf{S}: S \cap X \neq \emptyset} ||D^{(n)}_{\mathbf{S}}||\, |Y(X,\mathbf{S})| \, ||D^{(n)}||_0
\end{split}
\end{equation}
Note that $|Y(X,\mathbf{S})| \leq |X|+|S| < |X| e^{|S|}$, with the last inequality following by $X \neq \emptyset$. Then
\begin{equation}
\begin{split}
\langle \psi| (D^{(n)})^{\dagger} D^{(n)} |\psi \rangle 
&\leq ||D^{(n)}||_0 \sum_{X} |X| \langle \psi|P(X)|\psi\rangle  \sum_{\mathbf{S}: S \cap X \neq \emptyset} ||D^{(n)}_{\mathbf{S}}|| e^{|S|}\, 
\\&\leq ||D^{(n)}||_0 ||D^{(n)}||_1 \sum_{X} |X|^2 \langle \psi|P(X)|\psi\rangle
\\&\leq ||D^{(n)}||_0 ||D^{(n)}||_1 \langle \psi | (H_0)^2 |\psi\rangle
\end{split}
\end{equation}
Note that $\mu_n>\mu_*>\log(4) > 1$, so $||D^{(n)}||_0 \leq ||D^{(n)}||_1 \leq ||D^{(n)}||_{\mu_n} \leq \eta_n \leq 2 \epsilon_0$, implying 
\begin{equation}
    D^{(n)} \rbleq 2 \epsilon_0 H_0
\end{equation}

Using \textbf{TQO-I} and \textbf{TQO-II}, $M^{(n)}$ can also be relatively bounded. We will introduce the notation
$$G(\ell, S) = \prod_{\alpha \in B_{\ell|S|}(S)} G_{\alpha},$$ which will reduce clutter in the following. Recall the notation $B_{\ell|S|}(S)=\{\alpha: \mathrm{dist}(\alpha,S)\leq \ell |S|\}.$ 

By construction, $M^{(n)}_\mathbf{S}$ is only nonzero for $|S|<d_{\star}$. That is, \textbf{TQO-I} applies to $M^{(n)}_\mathbf{S}$, as $M^{(n)}_\mathbf{S}$ commutes with all the stabilizers and is not too large, and so $M^{(n)}_\mathbf{S}$ is itself a linear combination of terms that are products of stabilizer generators. Then, by \textbf{TQO-II}, the locations of these stabilizer generators can be taken to be within $B_{\ell|S|}(S)$ ($\ell$ is $\Theta(1)$).
In particular, 
\begin{equation}
    M^{(n)}_\mathbf{S}G(\ell, S) = G(\ell, S) M^{(n)}_\mathbf{S} =  \langle M_{\mathbf{S}}^{(n)} \rangle G(\ell, S)
\end{equation}
This follows from the fact that $M^{(n)}_\mathbf{S}$ is a linear combination of terms in $\caG_{B_{\ell|S|}(S)}$, and each of those terms reduces to a constant when acting on $G(\ell, S)$.

We will define $\overline{M^{(n)}_\mathbf{S}} = M^{(n)}_\mathbf{S} - \langle M_{\mathbf{S}}^{(n)} \rangle$.
If we put any ground state $\ket{g}$ on the right hand side of $M^{(n)}_\mathbf{S} G(\ell, S) = \langle M_{\mathbf{S}}^{(n)} \rangle G(\ell, S)$, we get $M^{(n)}_\mathbf{S}\ket{g} = \langle M_{\mathbf{S}}^{(n)} \rangle \ket{g}$, so $||M^{(n)}_\mathbf{S}|| \geq |\langle M_{\mathbf{S}}^{(n)} \rangle|$. Accordingly, $\Big|\Big|\overline{M^{(n)}_\mathbf{S}}\Big|\Big| \leq ||M^{(n)}_\mathbf{S}|| + |\langle M_{\mathbf{S}}^{(n)} \rangle| \leq 2 ||M^{(n)}_\mathbf{S}||$. Notice that $\overline{M^{(n)}_\mathbf{S}} \prod_{\alpha \in B_{\ell|S|}(S)} G_{\alpha} = \prod_{\alpha \in B_{\ell|S|}(S)} G_{\alpha} \overline{M^{(n)}_\mathbf{S}} = 0$, and $\overline{M^{(n)}_\mathbf{S}}$ commutes with every $E_\alpha$ and $G_\alpha$.

Then
\begin{equation}
\begin{split}
\langle \psi | (\overline{M^{(n)}})^\dagger \overline{M^{(n)}} | \psi \rangle &= \sum_{\mathbf{S}} \sum_{\mathbf{S}'} \langle \psi | (\overline{M^{(n)}_\mathbf{S}})^\dagger \overline{M^{(n)}_{\mathbf{S}'}} | \psi \rangle
\\&= \sum_{\mathbf{S}} \sum_{\mathbf{S}'} \langle \psi | (1-G(\ell, S')) (1-G(\ell, S))(\overline{M^{(n)}_\mathbf{S}})^\dagger \overline{M^{(n)}_{\mathbf{S}'}} (1-G(\ell, S)) (1-G(\ell, S'))| \psi \rangle
\\&\leq \sum_{\mathbf{S}} \sum_{\mathbf{S}'} \Big|\Big|\overline{M^{(n)}_\mathbf{S}}\Big|\Big|\, \Big|\Big|\overline{M^{(n)}_{\mathbf{S}'}}\Big|\Big|\,||(1-G(\ell, S)) (1-G(\ell, S'))|\psi\rangle||^2
\\&= \sum_{\mathbf{S}} \sum_{\mathbf{S}'} \Big|\Big|\overline{M^{(n)}_\mathbf{S}}\Big|\Big|\, \Big|\Big|\overline{M^{(n)}_{\mathbf{S}'}}\Big|\Big|\,\langle \psi|(1-G(\ell, S)) (1-G(\ell, S'))|\psi\rangle
\\&= \langle \psi| \left(\sum_{\mathbf{S}} \Big|\Big|\overline{M^{(n)}_\mathbf{S}}\Big|\Big| (1-G(\ell, S))\right)^2 |\psi \rangle
\end{split}
\end{equation}
which means $\overline{M^{(n)}} \rbleq \sum_{\mathbf{S}} \Big|\Big|\overline{M^{(n)}_\mathbf{S}}\Big|\Big| (1-G(\ell, S)) \rbleq \sum_{\mathbf{S}} \Big|\Big|\overline{M^{(n)}_\mathbf{S}}\Big|\Big|\sum_{\alpha \in B_{\ell|S|}(S)} E_\alpha$. We can massage this further, freely using that the $E_{\alpha}$ are positive semi-definite, Hermitian, and commuting.
\begin{equation}
\begin{split}
    \overline{M^{(n)}} &\rbleq \sum_\mathbf{S} \Big|\Big|\overline{M^{(n)}_\mathbf{S}}\Big|\Big| \sum_{\alpha \in B_{\ell|S|}(S)} E_{\alpha}
    \\&= \sum_\alpha E_\alpha  \sum_{\mathbf{S}:\alpha \in B_{\ell|S|}(S)} \Big|\Big|\overline{M^{(n)}_\mathbf{S}}\Big|\Big|
    \\&= \sum_\alpha E_\alpha  \sum_{k=1}^\infty \sum_{\mathbf{S}:\alpha \in B_{\ell|S|}(S), |S| =k} \Big|\Big|\overline{M^{(n)}_\mathbf{S}}\Big|\Big|
    \\&\rbleq \sum_\alpha E_\alpha  \sum_{k=1}^\infty \sum_{\mathbf{S}:\alpha \in B_{\ell|S|}(S), |S| =k} (\sum_{\beta: \text{dist}(\alpha,\beta) \leq \ell k} I_{\beta \in S}) \Big|\Big|\overline{M^{(n)}_\mathbf{S}}\Big|\Big|
\end{split}
\end{equation}
In the last line, we have used that in order for $B_{\ell|S|}(S)$ to contain $\alpha$, by definition there must be a $\beta \in S$ that is within a distance $\ell |S|$ from $\alpha$. Rearranging, 
\begin{equation}
\begin{split}
    \sum_\mathbf{S} \overline{M^{(n)}_\mathbf{S}} &\rbleq \sum_\alpha E_\alpha  \sum_{k=1}^\infty \sum_{\beta: \text{dist}(\alpha,\beta) \leq \ell k} \sum_{\mathbf{S}:\alpha \in B_{\ell|S|}(S), |S| =k, \beta \in S} \Big|\Big|\overline{M^{(n)}_\mathbf{S}}\Big|\Big|
    \\&\rbleq \sum_\alpha E_\alpha  \sum_{k=1}^\infty \sum_{\beta: \text{dist}(\alpha,\beta) \leq \ell k} \sum_{\mathbf{S}:\beta \in S, |S|=k}  \Big|\Big|\overline{M^{(n)}_\mathbf{S}}\Big|\Big|
    \\&\rbleq \sum_\alpha E_\alpha  \sum_{k=1}^\infty e^{-\mu_n k} \sum_{\beta: \text{dist}(\alpha,\beta) \leq \ell k} \sum_{\mathbf{S}:\beta \in S, |S|=k}  e^{\mu_n S} \Big|\Big|\overline{M^{(n)}_\mathbf{S}}\Big|\Big|
    \\&\rbleq \sum_\alpha E_\alpha  \sum_{k=1}^\infty e^{-\mu_n k} \sum_{\beta: \text{dist}(\alpha,\beta) \leq \ell k} \Big|\Big|\overline{M^{(n)}_\mathbf{S}}\Big|\Big|_{\mu_n}
    \\&\rbleq \Big|\Big|\overline{M^{(n)}_\mathbf{S}}\Big|\Big|_{\mu_n} \sum_\alpha E_\alpha  \sum_{k=1}^\infty e^{-\mu_n k} |B_{\ell k}(\alpha)|
    \\&\rbleq \frac{1}{1-e^{-\mu_n + \ell \kappa}} \Big|\Big|\overline{M^{(n)}_\mathbf{S}}\Big|\Big|_{\mu_n} \sum_\alpha E_\alpha
    \\&\rbleq \frac{4}{1-e^{-\mu_n + \ell \kappa}} \epsilon_0 H_0
    \\&\rbleq 8 \epsilon_0 H_0
\end{split}
\end{equation}
We use $\mu_n \geq \ell \kappa + \log(2)$ to perform the sum on $k$ and to bound the constant out front, and we use $\Big|\Big|\overline{M^{(n)}_\mathbf{S}}\Big|\Big|_{\mu_n} \leq 2 ||M^{(n)}_\mathbf{S}||_{\mu_n} \leq 2 \eta_n \leq 4 \epsilon_0$ from Proposition~\ref{prop: running couplings}.

Putting together the bounds on $\overline{M^{(n)}}$ and $D^{(n)}$, we have $(K^{(n)}-H_0) \rbleq 10 \epsilon_0 H_0$.

\section{Extensions to classical LDPC code Hamiltonians}\label{sec:classical}

Here we explain how our results can be extended to classical code Hamiltonians without topological order, but instead exhibit \textit{symmetry-breaking} orders.
These models will not satisfy our \textbf{TQO} conditions as stated in Sec.~\ref{subsec:TQOPauli}.
Therefore, we will restate sufficient conditions that allow the application of the same proof strategy, leading to stability results similar to our main Theorems in Sec.~\ref{sec:theorems}.

We define the Pauli $Z$ group $\caP_\Lambda^Z \subseteq \caP_\Lambda$ as follows
\begin{equation}
    \caP_\Lambda^Z \coloneqq \{p \in \caP_\Lambda : p = \bigotimes_x p_x \text{ with } p_x \in \left\{
    \mathbb{1}, \mathsf{Z} \right\} \}.
\end{equation}
Similarly, we define the Pauli $X$ group
\begin{equation}
    \caP_\Lambda^X \coloneqq \{p \in \caP_\Lambda : p = \bigotimes_x p_x \text{ with } p_x \in \left\{
    \mathbb{1}, \mathsf{X} \right\} \}.
\end{equation}
A classical code Hamiltonian has the following form, 
\begin{equation}
    H_0 = \sum_{\alpha \in \caE} E_\alpha, \quad \text{ where } E_\alpha = (1-C_\alpha)/2,  \quad C_\alpha \in P_\Lambda^Z.
\end{equation}
The checks generate a stabilizer group, $\caG = \langle \{ C_\alpha \}_{\alpha \in \caE} \rangle$, thereby defining a stabilizer code.
Compared to general stabilizer Hamiltonians in Sec.~\ref{sec:stabilizer_Hamiltonian}, the only additional requirement is that each check is a $Z$ Pauli string.
Therefore, all results in Sec.~\ref{sec:stabilizer_Hamiltonian} apply here, in particular that $K = N - \log_2 |\caG|$, and the ground state degeneracy of $H_0$ is $2^K$.

The \textit{symmetry group} of $H_0$ is defined as follows,
\begin{equation}
    \mathbf{X} \coloneqq \mathcal{Z}(\caG) \cap \caP_\Lambda^X.
\end{equation}
Here, $\mathcal{Z}(\caG)$ denotes the \textit{centralizer} of $\caG$ within $\caP_\Lambda$.
Elements of $\mathbf{X}$ are conventionally known as the Pauli $X$ logical operators of the code $\caG$.
We have $|\mathbf{X}| = 2^K$.

We define the \textit{symmetric code distance} as follows.
\begin{equation}
    d_{\rm sym} = \min\{ | \supp(p) |: p \in \mathbf{X}, p \neq \mathbb{1} \}.
\end{equation}
This definition coincides with the classical code distance when $\caG$ is treated as a classical binary code.

We say an operator is $\mathbf{X}$-symmetric if it commutes with all elements in $\mathbf{X}$.
By definition, $H_0$ is $\mathbf{X}$-symmetric.

With these, we state conditions on $H_0$ in parallel to Sec.~\ref{subsec:TQOPauli}.
We require \textbf{LDPC} and \textbf{Growth of balls}, which are identical with those in Sec.~\ref{subsec:TQOPauli}.
We additionally require the following:

\begin{itemize}[labelsep=0em, leftmargin=5.5em, labelwidth=5em, itemindent=0em, align=parleft, itemsep = 2em]
\item[\textbf{TQO-Ic}]
For every set of qubits $A \subseteq \Lambda$ with $|A| < d_{\rm sym}$ and for every $\mathbf{X}$-symmetric $p \in \caP_\Lambda$ such that $\supp(p) \subseteq A$, we have
\begin{equation}\begin{split}
\forall g \in \caG: [p,g]=0 \qquad \Rightarrow \qquad p \in \caG.
\end{split}\end{equation}
\textbf{TQO-Ic} follows directly from the definition of $d_{\rm sym}$.

\item[{\textbf{TQO-IIc}}] 
For every connected set $S$ in the check graph $\caE$ where $|S|  < \widetilde{d}$, we have 
\begin{equation}\begin{split}
\caG(S) \subseteq \caG_{B_{\ell|S|}(S)}
\end{split}\end{equation}
where $\ell = \Theta(1)$, and
$B_{\ell|S|}(S)=\{\alpha : \mathrm{dist}(\alpha,S)\leq \ell |S|\}$.
% Note that here \textbf{TQO-IIc} differs only slightly from the statement of \textbf{TQO-II} in Sec.~\ref{subsec:TQOPauli}, in that we require this for all connected sets $S$ without restrictions on its size, but this difference is not essential.
\end{itemize}
We define $d_{*, \rm sym} = \min(d_{\rm sym}/w_c, \widetilde{d})$. Theorems \ref{thm: main_classical} and \ref{thm: locality_classical} are meaningful i.e. they give a stable gap and ground state degeneracy if $d_{*, \rm sym} \geq c \log(N)$ for $c>0$. 

We are interested in the stability of the $\mathbf{X}$-symmetric $H_0$ under an $\mathbf{X}$-symmetric perturbation $Z^{(0)}$.
Viewing $Z^{(0)}$ as an operator-collection, we write
\begin{equation}
    Z^{(0)} = \sum_{\bfs} (Z^{(0)})_\bfs.
\end{equation}
By symmetry, we may choose an operator-collection such that each $(Z^{(0)})_\bfs$ is $\mathbf{X}$-symmetric.
Similarly, we can choose an operator-collection such that each $(H_0)_\bfs$ is $\mathbf{X}$-symmetric.
These choices can be made following Sec.~\ref{sec:initial_choice_operator_collection}.
We additionally require 
\begin{equation}
    ||Z^{(0)}||_{\mu_0} < \epsilon_0,
\end{equation}
where we recall that the word norm $||\ldots||_{\mu_0}$ is defined in \eqref{eq:def_mu_norm}.

We can apply the iterative KAM procedure from Sec.~\ref{sec:KAM_scheme} to $H_0 + Z^{(0)}$.
By induction, each of the operator-collections $A^{(n)}, D^{(n)}, M^{(n)}, V^{(n)}, E^{(n)}$ and $Z^{(n)}$ is $\mathbf{X}$-symmetric at all scales $n$.
We note that \textbf{TQO-I} and \textbf{TQO-II} are invoked in Lemma~\ref{lem: bound on ghosts} and in Sec.~\ref{sec:proof_rel_bounded} for characterizing $M_\bfs^{(n)}$, and they can be traded for \textbf{TQO-Ic} and \textbf{TQO-IIc} for $\mathbf{X}$-symmetric $M^{(n)}_\bfs$.
The same Propositions and Lemmas can therefore be used to control their word norms under the flow equations.

With these, we can state our results for classical LDPC code Hamiltonians.
They are in close parallel to Theorems~\ref{thm: main} and \ref{thm: locality}.

\begin{theorem}\label{thm: main_classical}
Assume that $\mu_0 > \mu_{*}$.
Then, there is an $\epsilon(\mu_0)$ that can depend on $\mu_0$ and our $\Theta(1)$-model parameters, such that, for any $\epsilon_0 \leq \epsilon(\mu_0)$, there is a $\mu_{\infty}>\mu_*$ such that the following holds:  There is a number $b$ such that the spectrum of the operator 
\begin{equation}
H=H_0+Z^{(0)} -b
\end{equation}
is contained in a union of intervals centered on $k=0,1,2,\ldots$, where the $k$-th interval is
\begin{equation}
    I_k =[k-(kC' \epsilon_0 + \delta), k+(kC' \epsilon_0 + \delta)], \quad \delta = C{N_c} \epsilon_0 e^{-\mu_\infty d_{*, \rm sym}}.
\end{equation}

% lengths  
% $C'\epsilon_0 m + C{N_c} \epsilon_0 e^{-\mu_\infty d_{*, \rm sym}}$
% \begin{enumerate}
%     \item $C{N_c} \epsilon_0 e^{-\mu_\infty d_{*, \rm sym}}$ for $m=0$.
%     \item $C'\epsilon_0 m + C{N_c} \epsilon_0 e^{-\mu_\infty d_{*, \rm sym}}$ for $m=1,2,\ldots$
% \end{enumerate}
\end{theorem}

\begin{theorem}\label{thm: locality_classical}
Under the same assumptions as in Theorem \ref{thm: main_classical} above, we have
\begin{equation}
\widetilde P = U P U^{-1}
\end{equation}
for a unitary $U$ that is locality-preserving and locally close to identity, in the following sense: For any operator $O$ supported in a connected set $Y$, we can write
\begin{equation}
U^{-1}OU= O+\sum_{r=1}^{\mathrm{diam}(\Lambda)} O_r +O_{\mathrm{bg}}
\end{equation}
where $\mathrm{diam}(\Lambda)$ is the diameter of the graph $\Lambda$, and
\begin{enumerate}
    \item  $O_r$ is supported in $\{x: \mathrm{dist}(Y,x) \leq r\}$.
    \item  $O_r$ satisfies \begin{equation} ||O_r|| \leq \epsilon_0 C(\mu_0,Y) e^{-cr}||O||\end{equation}
    for some $c>0$ and constant $C(\mu_0,Y)$ that can depend on $\mu_0$ and $Y$.
    \item  $O_{\mathrm{bg}}$ (with $\mathrm{bg}$ standing for `background') is bounded as 
    $|| O_{\mathrm{bg}} || \leq     C{N_c} \epsilon_0 e^{-\mu_\infty d_{*, \rm sym} } ||O||$.
\end{enumerate}
The same statement holds true for $UOU^{-1}$. 
\end{theorem}

% \begin{itemize}
%     \item Symmetry breaking/ classical codes
%     \item paramagnetic

%     \item 
% \end{itemize}

\section{Discussion of some physical consequences}

We have proven that classical and quantum LDPC stabilizer codes define robust gapped phases of matter, in the sense that their ground state degeneracy and gap are both robust to sufficiently small local perturbations. In this section we comment on some further physical consequences of our proof, discussing what physical properties of the phase remain robust in the vicinity of the unperturbed LDPC Hamiltonian. We will keep the presentation more qualitative in this section. 

\subsection{Order parameters and emergent symmetries}

Let us start by considering classical LDPC codes. These, when embedded into quantum Hamiltonians, are generalizations of spontaneous symmetry breaking phases; the logical operators of the code constitute a set of symmetries, which are broken in the ground states. These symmetry breaking patterns can be characterized by a set of local order parameters. In particular, for a classical code with $K$ logical bits, we can find a (non-unique) set of sites $\{x_i\}_{i=1}^{K}$ such that the eigenvalues $\mathsf{Z}_{x_i} = \pm 1$ of the corresponding local Pauli matrices can be used to label a full basis set for the $2^K$ degenerate ground states. If we think of classical codes as a special case of quantum codes, then we can think of these $\mathsf{Z}_{x_i}$ as a choice for the $K$ logical $Z$ operators of the code, conjugate to the $K$ logical $X$ operators that correspond to the classical logicals / symmetries. This means that there is a basis choice for the logicals $\{\mathbf{X}_i\}_{i=1}^{K}$ such that $\mathsf{Z}_{x_i}$ anti-commutes with $\mathbf{X}_i$ and commutes with $\mathbf{X}_{j}$ for $j \neq i$. 

Now consider the perturbed Hamiltonian $H_0 + V$, where $V$ are perturbations commuting with all the symmetries (classical logicals). By Theorem~\ref{thm: locality}, we can relate the ground state subspaces as $P = U P_0 U^\dagger$. We can use this to define a new, ``smeared out'' order parameter $\tilde{\mathsf{Z}}_{x_i} = U \mathsf{Z}_{x_i} U^\dagger$ for every $i=1,\ldots,K$. By properties of $U$, these are quasi-local operators, localized to an arbitrary precision $\epsilon$ within the ball $B_r(x_i)$ around $x_i$ with radius $r=O(\log(1/\epsilon))$, independent of $n$. By construction, these have the same expectation values in the new ground states as $\mathsf{Z}_{x_i}$ in the originals:
\begin{equation}
    \langle \tilde{\psi} | \tilde{\mathsf{Z}}_{x_i} |\tilde{\psi} \rangle = \langle \tilde{\psi} | U \mathsf{Z}_{x_i} U^\dagger |\tilde{\psi}\rangle = \langle \psi | \mathsf{Z}_{x_i} |\psi \rangle,
\end{equation}
where $\ket{\tilde{\psi}}$ is a ground state of $H_0+V$ and $\ket{\psi}$ is a ground state of $H_0$. Thus, the $2^K$ symmetry breaking product state ground states of $H_0$ labeled by the eigenvalues of $\mathsf{Z}_{x_i}$ map onto a set of $2^K$ ground states labeled by eigenvalues of $\tilde{\mathsf{Z}}_{x_i}$. 
Moreover, as $U$ is $\mathbf{X}$-symmetric by construction, the smeared out operators have the same symmetry properties as the old ones. Thus $\tilde{\mathsf{Z}}_{x_i}$ serves as a local order parameter diagnosing the spontaneous breaking of the symmetry $\mathbf{X}_i$. 

In fact, we expect that one should also be able to measure the same symmetry breaking using the original, un-smeared order parameters, which is more practically relevant. By the analyticity and hence continuity of the ground state projector $\tilde{P}$, the expectation values of local observables are continuous functions of the local perturbation strength. This implies that for any constant $c < 1$ we can make the perturbation strength sufficiently small (but finite) such that $|\langle\tilde{\psi}|\mathsf{Z}_{x_i}|\tilde{\psi}\rangle| > c$ in the above ground states, which is a more practical way of diagnosing the spontaneous breaking of $\mathbf{X}_i$, without the need to measure the perturbation-dependent smeared order parameters $\tilde{\mathsf{Z}}_{x_i}$. 

An advantage of the former approach, using smeared operators, is that it generalizes better to quantum codes. In that case, there are no local order parameters, but we can still consider the set of all logical operators, $\mathbf{X}_i$ and $\mathbf{Z}_i$, $i=1,\ldots,K$, that come in conjugate anti-commuting pairs. Now there is no symmetry requirement on the perturbations so both species of logicals get smeared out by $U$, to $\tilde{\mathbf{X}}_i = U \mathbf{X}_i U^\dagger$ and $\tilde{\mathbf{Z}}_i = U \mathbf{Z}_i U^\dagger$, but their commutation relations are preserved. These constitute a set of \emph{emergent symmetries} characterizing the phase of matter. These symmetries are also deformable: we can multiply them with the operators $U C_\alpha U^\dagger$ (with which they commute) to find other operators that have the same action on the ground state subspace. By the cleaning lemma~\cite{bravyi2009no}, for any region $Y$ with $|Y| < d$, all logicals can be deformed in such a way as to avoid $Y$; for example, we can do this for any ball $B_{r}(x)$ with $r < \log{d}/\kappa$. Since the smearing out is only by an $O(1)$ factor, one can then also clean out the new logicals, $\tilde{\mathbf{X}}_i, \tilde{\mathbf{Z}}_i$ from any ball with radius $r=o(\log{d})$. In this sense, the emergent symmetries of perturbed qLDPC codes have the characteristics of higher-form symmetries~\cite{gaiotto2015generalized,mcgreevy2023generalized}.

Similarly, if one of the original logical is localized to some sub-region $Y$ of size $o(n)$ of the lattice $\Lambda$, then its smeared out version can be approximated, with exponential precision, by an operator supported on a region whose size is proportional to $|Y|$. One example where this is relevant are hypergraph product codes~\cite{tillich2013quantum}: these come equipped with an abstract ``square grid'' structure where qubits are labeled by a pair of labels $(x,y)$. The code distance is $d = \Theta(\sqrt{n})$ and one can choose a canonical basis set of logical operators that are supported on ``rows'' and ``columns'', i.e. sets of qubits where either the first or the second label is held constant (see Ref.~\cite{rakovszky2024product} for details), analogous to the familiar logical operators of the toric code. Our results imply that these canonical logicals maintain their ``one-dimensional'' nature and only get smeared out to an $O(1)$ number of neighboring rows / columns. 

%\tibor{To comment on good codes here, we would need to know something about the behavior of $U$, in particular the constant $c$ in Thm.~\ref{thm: locality}, in the limit where the perturbation is turned off but I'm not sure if we have that.}

%The situation of \emph{good} quantum codes [Breuckmann,PanteleevKalachev,etc.] is more difficult. In this case, $d = \Theta(n)$, so that the original logical operators are already supported on a finite fraction of all qubits. Nevertheless, this fraction might be quite small. Let us say that $d < c n$ where $c$ is some constant. Take a logical with support of size $d$; the corresponding smeared-out logical can be approximated

\subsection{Excitations}

So far we focused on ground state properties. Another way to characterize the phase is in terms of excitations. We will now argue that the local excitations defined by the unperturbed Hamiltonian remain well-defined for small enough perturbations, building on the ideas of BHM~\cite{Bravyi_2010}. 

For simplicity, let us consider a stabilizer code whose checks are all independent. In that case, for any check $C_\alpha$, we can find an appropriate excitation operator $\mathcal{D}_\alpha$ which anti-commutes with $C_\alpha$ and commutes with all other checks. For classical codes one can think of $\mathcal{D}_\alpha$ as a set of spin flips creating the analog of a domain-wall, while for quantum codes it is the analog of an operator creating an anyon. They create localized excitations that form the lowest band, with energy $1$ in $H_0$.  One can also argue that in many cases the corresponding operator has to be non-local, with support that is comparable to the code distance~\cite{rakovszky2023gauge,hong2024quantum}. 

For the perturbed Hamiltonian, we can again use $U$ from Thm.~\ref{thm: locality} to define smeared out excitation operators $\tilde{\mathcal{D}}_\alpha = U \mathcal{D}_\alpha U^\dagger $. Then applied to a ground state $\ket{\tilde{\psi}}$, these create excitations localized near the support of check $C_\alpha$. To see this, consider measuring some local operator $O$, supported on a ball $B_r(x)$, which one can think of as measuring the local energy density near $x$. We take the ball to be far from the location of $C_\alpha$, i.e. $r \ll R$ where $R \equiv \text{dist}(x,\text{supp}(C_\alpha))$. We want to argue that if the distance $R$ is large, then $O$ cannot distinguish $\tilde{\mathcal{D}}_{\alpha}\ket{\tilde{\psi}}$ from $\ket{\tilde{\psi}}$. In particular, 
\begin{equation}\label{eq:LocalExcitation}
\langle \tilde{\psi} | \tilde{\mathcal{D}}_\alpha O \tilde{\mathcal{D}}_\alpha |\tilde{\psi} \rangle = \langle \tilde{\psi} | O | \tilde{\psi} \rangle + h(R),    
\end{equation}
where $\lim_{R\to\infty}h(R) = 0$. 

To show this, we note that $\langle \tilde{\psi} | O | \tilde{\psi} \rangle = \langle \psi | \tilde{O} | \psi \rangle$ and $\langle \tilde{\psi} | \tilde{\mathcal{D}}_\alpha O \tilde{\mathcal{D}}_\alpha |\tilde{\psi} \rangle = \langle \psi | \mathcal{D}_\alpha \tilde{O} \mathcal{D}_\alpha |\psi \rangle$, where $\ket{\psi}$ is a ground state of $H_0$ and $\tilde{O} = U^\dagger O U$. Up to precision $e^{-O(c \tilde{r})}$ we can approximate $\tilde{O}$ by an operator supported on the ball $B_{\tilde{r}}(x)$. It is thus enough to argue that the presence of the excitation at the check $C_\alpha$ cannot be picked up by any such operator, as long as $\tilde{r}$ is not too large. This follows from the condition \textbf{TQO-II}\footnote{The argument is more straightforward for classical codes and it is worth spelling out. In that case, the probe operator $O$, standing in for local energy density, should be chosen to be symmetric under all the classical logicals. All such operators are generated by the checks $C_\alpha$ together with single site pauli $X$ operators. $\mathcal{D}_\alpha$ commutes with all elements of this generating set with the exception of $C_\alpha$. Thus, the only question is whether $O$ contains this check or not. Then, by \textbf{TQO-IIc}, the answer is no if $O$ is supported sufficiently far away from $C_\alpha$.}. The reduced density matrix of $\ket{\psi}$ on $B_{\tilde{r}}(x)$ involves elements of the local stabilizer group $\mathcal{G}(B_{\tilde{r}}(x))$. The same is true of $\mathcal{D}_\alpha \ket{\psi}$, and the two only differ in the sign of terms involving $C_\alpha$. However, by \textbf{TQO-II} and Eq.~\eqref{eq:def_kappa}, $\mathcal{G}(B_{\tilde{r}}(x)) \subseteq \mathcal{G}_{\tilde{r} + \ell e^{\kappa \tilde{r}}}$. As long as this larger ball does not contain $C_\alpha$, the two reduced density matrices are the same. We can thus choose $\tilde{r} \propto \log(R/\ell) / \kappa$, which makes $h(R)$ in Eq.~\eqref{eq:LocalExcitation} decay as $h(R) \propto R^{-c/\kappa}$,
where $c$ is the constant controlling the locality of $U$ in Thm.~\ref{thm: locality}. We note that the slow decay is due to the very weak TQO condition being used; if the radius appearing in \textbf{TQO-II} could be strengthened from $\ell |S|$ to $\ell \log{|S|}$, which we expect in many cases of interest, then we could choose $\tilde{r} = O(R)$, resulting in $h(R) \propto e^{ -\frac{c}{\kappa \ell} R}$.

Since $\tilde{\mathcal{D}}_\alpha\ket{\tilde{\psi}}$ only contains excitations within some finite neighborhood of $\text{supp}(C_\alpha)$, it should be spanned, up to good approximation, by low-lying eigenstates of $H_0 + V$. In BHM, a stronger statement is made, using the fact that in their case, there exists a quasi-local unitary $U$ that relates not only the ground state subspaces, but also low-lying excited states, between perturbed and un-perturbed Hamiltonians. In paricular, let $Q$ be the projector onto eigenstates with energy $1$ in $H_0$ and $\tilde{Q}$ be the projector onto the corresponding to eigenstates with energies in the interal $I_1$ defined in Thm.~\ref{thm: main}, which remains well-separated from all the other intervals for sufficiently weak perturbations. Now assume that there exists $U$ such that both $\tilde{P} = U P U^\dagger$ and $\tilde{Q} = U Q U^\dagger$, with the same locality properties as in Thm.~\ref{thm: locality}. We then have, following BHM, that
\begin{equation}
\tilde{Q} \tilde{\mathcal{D}}_\alpha \ket{\tilde{\psi}} = U Q \mathcal{D}_\alpha \ket{\psi} = U \mathcal{D}_\alpha \ket{\psi} = \tilde{\mathcal{D}}_\alpha \ket{\tilde{\psi}}.    
\end{equation}
This would mean that the excited state created by $\tilde{\mathcal{D}}_\alpha$ lies entirely within the lowest band of $H_0 + V$. We can thus think of this band as being spanned by the localized excitations created by $\tilde{\mathcal{D}}_\alpha$ that are adiabatically connected to the ones of $H_0$. Whether an appropriate choice of $U$ satisfying this condition exists is an interesting question for future work. 

Finally, we note that while here we considered the case when we initially excite a single check, the above considerations should also generalize to other low-lying excitations. For example, instead of $\mathcal{D}_\alpha$, we can consider an operator that excites a $\Theta(1)$ number of checks, all spatially far separated from each other. We can then smear out this operator, to get a new operator that creates a set of far-separated localized excitations in the perturbed Hamiltonian. This is relevant for e.g. the case when the code exhibits some redundancies, preventing one from creating a single excitation. One would then want to find $U$ that also rotates the projector of the appropriate set of excitations between unperturbed and perturbed Hamiltonians.  

\section{Summary and outlook}

In this paper, we generalized the stability proof of Ref.~\cite{Bravyi_2010} to perturbations of a large set of quantum LDPC stabilizer codes, without the restriction of Euclidean locality, including constructions of good qLDPC codes and other examples living on expander graph geometries. We proved the robustness of ground state degeneracy (up to a splitting exponentially suppressed in the code distance), the stability of the spectral gap, and the existence of a unitary relating perturbed and unperturbed ground states that is quasi-local with respect to the interaction distance of the unperturbed Hamiltonian. These results provide a rigorous starting point for defining and studying phases of matter on generic graphs, away from the limit of Euclidean lattices, which is relevant in light of the ongoing experimental efforts that are currently pushing quantum many-body physics into this new territory~\cite{Kollar2019,Periwal2021,Bluvstein2022,LukinLDPC}.

Our results raise a number of questions, some of which we have already flagged above. Some of these involve the physical consequences stemming from the quasi-adiabatic continuation between ground states. In general, some physical consequences might require strengthening the form of the \textbf{TQO-II} condition, which we expect should be possible in many cases of interest. As we also noted, BHM showed a stronger quasi-adiabatic continuation result~\cite{Bravyi_2010}, where the unitary $U$ relates not only ground states, but also low-lying excited states, which was needed to prove statements about the form of low-lying excitations in the perturbed model. Is there a similar adiabatic continuation of excitations in our case? 

Another set of questions involves the set of allowed perturbations. Can the locality conditions on $V$ be relaxed while maintaining stability, e.g. to generic $k$-local perturbations without reference to any graph structure? Our TQO-II condition is easy to satisfy and asks that connected stabilizers can be constructed out of stabilizer generators that are not too far away. If we instead ask for a stricter condition (i.e. that TQO-II must also apply also to disconnected stabilizers), our proof shows stability to $k$-local perturbations. The main change to the proof is in the construction of the initial operator collections from the perturbations. In particular, we can construct the initial operator collections via sandwiching the qubit support (rather than a minimal connected set covering the qubit support) with projectors. These words can be disconnected (as can the resulting words generated over the course of the KAM steps); we had previously avoided this, as we used connectivity of all words to apply TQO-II in the proof of relative form boundedness of $M^{(n)}$ in Section \ref{sec:proof_rel_bounded} and in bounding commutators of ghosts with other operators in Lemma~\ref{lem: bound on ghosts}. However, disconnected words are fine under the stricter TQO-II, and the corresponding word norms of long-range, few-body perturbations are much smaller than if we asked for connected words. As an example, under this looser construction of the intial operator collections, a two-body nonlocal operator separated by a distance $r$ has word norm weighted by $e^2$ rather than $e^r$. This broadens the class of perturbations that have bounded word norm in the thermodynamic limit to include cases like $\frac{1}{L}\sum_{ij} Z_i Z_j$. Finally, we had stated some of our results in terms of a Pauli norm for simplicity in calculation; the Pauli norm can now be chosen to be exponentially weighted by qubit support, rather than exponentially weighted in the minimal connected covering. The main change to the results is that there is no longer graph locality of the unitary $U$ that maps the perturbed spectral projector back to the unperturbed one; it is only locality-preserving in the sense of word norms. In all, the stricter TQO-II condition gives stability to $k$-local perturbations rather than just graph local ones; which models satisfy this stricter condition?

Finally, it would be very intriguing to generalize our results from stabilizer codes to more general families of Hamiltonians, either to all commuting projector models, as was done by BHM, or even to the more general setting of frustration-free Hamiltonians as in Ref.~\cite{Michalakis_2013}. Other avenues of generalization would involve considering the stability of finite temperature phases~\cite{hastings2007quantum,capel2023decay} in the non-Euclidean setting. 
%This is also relevant to a more rigorous understanding of the spin-glass order that is present in classical~\cite{mezardmontanari2009book} and quantum \tibor{forward-cite ourselves?} LDPC codes.

% \tibor{Local vs global gap?}

\section*{}
\emph{Note Added---} In work that appeared in the same arXiv posting, other authors have independently investigated the stability of LDPC codes~\cite{yinlucas_stability_2024}. 

\section{Acknowledgements}
We are grateful for helpful discussions with Israel Klich, Christopher Laumann, Ali Lavasani and Sagar Vijay. V.K. and T.R. especially thank Benedikt Placke and Nikolas Breuckmann for many insightful discussions and close collaboration on related work. 

% \vedika{grants. Oxford conference for hospitality. Chris Laumann. Niko and Benedikt. Especially Tibor. } \textcolor{red}{Nick: Nick and Yaodong thank Israel Klich for a conversation on methods for proving gaps. Yaodong thanks Ali Lavasani and Sagar Vijay.}

WDR was supported by the FWO-FNRS EOS research project G0H1122N EOS 40007526 CHEQS, the KULeuven Runner-up grant iBOF DOA/20/011, and the internal KULeuven grant C14/21/086.
V.K. acknowledges support from the Packard Foundation through a Packard Fellowship in Science and Engineering and the Office of Naval Research Young Investigator Program (ONR YIP) under Award Number N00014-24-1-2098 (support for N.O.D. and V.K.). 
Y.L. was supported in part by the Gordon and Betty Moore Foundation’s EPiQS Initiative through Grant GBMF8686, in part by the US Department of Energy, Office of Science under Award No DE-SC0019380, and in part by a Q-FARM Bloch Postdoctoral Fellowship at Stanford University.
T.R. was supported in part by Stanford Q-FARM Bloch Postdoctoral Fellowship in Quantum Science and Engineering and by the HUN-REN Welcome Home and Foreign Researcher Recruitment Programme 2023.

% \bibliography{refs}
\printbibliography

\appendix

\section{Arguments for TQO-II conditions for specific examples}\label{app:TQO}

In Sec.~\ref{subsec:TQOPauli} we claimed that the condition \textbf{TQO-II} is satisfied by a number of relevant examples of qLDPC codes, such as known construction of good qLDPC codes and hypergraph products of classical expander codes. We will argue for this now. 

The codes we want to consider are all CSS codes. We can thus separately consider the parts of the stabilizer group generated by ${X}$ and $Z$ checks, which we denote $\mathcal{G}^{X}$ and $\mathcal{G}^{Z}$ and similarly for $\mathcal{G}(S)$ and $\mathcal{G}_S$. It is then enough to show \textbf{TQO-II} for the $X$ and $Z$ parts of the stabilizer group separately in order for it to to be satisfied for the entire stabilizer group.

We will now argue that the claim holds for $\mathcal{G}^X$ (the argument for $\mathcal{G}^Z$ is analogous). The $X$ checks define some classical code with a parity check matrix $\mathbb{H}_X$. It is more intuitive to consider the question in terms of the transpose code, $\mathbb{H}_X^\text{T}$ (which, in the terminology of Refs.~\cite{rakovszky2023gauge,kubica2018ungauging} is \emph{gauge dual} to the quantum code in question), which exchanges the roles of bits and checks. For this classical code, \textbf{TQO-II} becomes a question of \emph{energy barriers}. $\text{supp}(S)$ becomes, in the transpose code, a set of checks. \textbf{TQO-II} then states that that any excitation contained in this set can be created by flipping a number of bits with a ball of radius $\ell |S|$ where, in defining the distance, we are now only using the checks of $\mathbb{H}_X^\text{T}$. 

\subsection{Good qLDPC codes}

The key observation, made in Refs.~\cite{rakovszky2023gauge,rakovszky2024product} is that in all known constructions of good qLDPC codes, the transpose code $\mathbb{H}_X^\text{T}$ (and similarly, the code $\mathbb{H}_Z^\text{T}$) is \emph{locally testable}, meaning that there exists a constant $\gamma = O(1)$ such that 
\begin{equation}\label{eq:LTC}
|\mathbb{H}_X^\text{T} \mathbf{x}| \geq \gamma |\mathbf{x}|    
\end{equation}
for any bit-string $\mathbf{x}$. In our case, $\mathbb{H}_X^\text{T} \mathbf{x}$ is contained within $\text{supp}(S)$ so we have $|\mathbf{x}| \leq \frac{w_c}{\gamma} |S|$. At the same time, we can assume that $\mathbf{x}$ forms a connected set, otherwise we would deal with each connected component separately. Thus, by definition, it fits within a ball of radius $r = |\mathbf{{x}}|/2$. This gives \textbf{TQO-II} with $\ell = w_v / (2\gamma)$.

\subsection{Hypergraph products of expander codes}

In the case of hypergraph product codes, the classical code $\mathbb{H}_X^\text{T}$ in question is the tensor product of two classical input codes. Such tensor product are not locally testable. Nevertheless, if the two input codes are chosen to be classical expander codes (e.g., good classical LDPC codes) than the resulting product will have sufficiently large energy barriers that it still satisfies \textbf{TQO-II}. The key points is that for a classical expander code, while Eq.~\eqref{eq:LTC} is not true in general, it is true when restricted to bit-strings $\mathbf{x}$ such that $|\mathbf{x}| < \nu N$ ($N$ being the number of bits in the classical code) for some constant $\nu = O(1)$. This property is then satisfied by the product of two such expander codes: i.e., let $\mathbb{H}_A$ and $\mathbb{H}_B$ be the parity check matrices defining the two input codes that make up the product, with expansion parameters $\nu_A,\nu_B$ and $\gamma_A,\gamma_B$. In the product, we can restrict ourselves to $|\mathbb{x}| \leq \text{min}(\nu_A N_A, \nu_B N_B)$; for these, we satisfy Eq.~\eqref{eq:LTC} with $\gamma = \text{min}(\gamma_A,\gamma_B)$ (to see this, note that $\mathbf{x}$ can be decomposed into either its rows or its columns; we can count the number of excited checks in each, which amounts to counting checks in $\mathbb{H}_A$ and $\mathbb{H}_B$, respectively). Now, consider bit-strings with $d_A d_B/2 > |\mathbf{x}| > \nu_A N_A$. We want to argue that in this case where exists some constant $C$ such that $|\mathbb{H}_X^\text{T}| \geq C d = \Theta(N^{1/2})$; in this case, we can choose $\tilde{d}$ in the in the \textbf{TQO-II} condition to be $C d$, such that all allowed sets $S$ still satisfy Eq.~\eqref{eq:LTC}. 

To argue this last point, consider all the non-empty rows $R$ of $\mathbf{x}$ and divide it into two sets, $R = R_< \cup R_>$, where $R_<$ contains rows with $\leq \nu_A N_A$ bits flipped and $R_>$ contains those with more. We now try to construct an $\mathbf{x}$ with energy $|\mathbb{H}_X^\text{T} \mathbf{x}| = o(N^{1/2})$ and find that this is impossible when $|\mathbf{x}| > \nu_A N_A = \Theta(N^{1/2})$. First, due to expansion, it must be that $|\mathbf{x}|_{R_{<}}| = o(N^{1/2})$, otherwise the energy would already be too high. Thus, almost all non-empty rows have many bits flipped. The same is true when we consider a decomposition into columns, rather than rows. This means that the total number of non-empty rows and columns has to be $|R|,|C| = O(N^{1/2})$. But, by the assumption that $|\mathbf{x}| \leq d_A d_B / 2$, most of these rows and columns must contain at least one excitation, bringing the energy up to $O(N^{1/2})$ again. 

\section{Proof of relation between Pauli norm and word norm}
\label{sec:proof of norm relationship}

\subsection{Proof of Proposition~\ref{prop:tripleboundsdouble}}
Consider an operator that is a sum of Paulis, $O = \sum_{p} c_p p$, and consider re-expressing it in terms of an interaction $O$ via the procedure in Sec.~\ref{sec:setup}.
We show that the resulting norm $||O||_{\mu}$ can be bounded above by $|||O|||_{\mu'}$ for a $\mu' = a \mu + b$, where $a,b$ are $\Theta(1)$ constants $a=w_q$ and $b=(w_q\log(2) + \kappa + \log(2 w_c))$. We will define $\kappa' = \kappa + \log(2 w_c)$ for ease.

Before directly bounding $||O||_{\mu}$, it is useful to note a few preliminary inequalities which we will use freely in the following. For a given $p$, there are at most  $2^{|\ext(p)|}$ $\mathbf{S}$ such that $p_{\mathbf{S}}$ is nonvanishing, and all of them satisfy $||p_{\mathbf{S}}|| \leq 1$. Since each qubit touches at most $w_q$ stabilizers, $|\ext(p)| \leq w_q |\mqs(p)|$. Because $\mqs(p)$ is connected, every $x,y \in \mqs(p)$ satisfies $\mathrm{dist}(x,y) \leq |\mqs(p)|$. Given $\alpha \in \ext(p)$, then there must be an $x \in \supp(p)$ such that $\mathrm{dist}(x,\supp(C_\alpha)) \leq |\mqs(p)|$. For any $\alpha$, $\sum_{x: \mathrm{dist}(x,\supp(C_\alpha)) < r} 1 \leq w_c e^{\kappa r}$.
\begin{equation}
\begin{split}
 ||O||_{\mu} &= \sup_{\alpha} \sum_{\mathbf{S}: \alpha \in S} || \sum_{p: \ext(p) = S} c_p p_{\mathbf{S}}|| e^{\mu |S|}
 \\&\leq \sup_{\alpha} \sum_{\mathbf{S}: \alpha \in S} \sum_{p: \ext(p) = S} |c_p| e^{\mu |S|}
\\&= \sup_{\alpha} \sum_{p: \alpha \in \ext(p)} \sum_{\mathbf{S}: S = \ext(p)} |c_p| e^{\mu |\ext(p)|}
\\&\leq \sup_{\alpha} \sum_{p: \alpha \in \ext(p)}  |c_p| 2^{|\ext(p)|}e^{\mu |\ext(p)|}
\\&\leq \sup_{\alpha} \sum_{p} \left(\sum_{x: \mathrm{dist}(x,\supp(C_\alpha)) \leq |\mqs(p)|} I_{x \in \supp(p)} \right) |c_p| e^{w_q (\mu + \log(2)) |\mqs(p)|}
\\&= \sup_{\alpha}  \sum_{r=1}^\infty e^{-\kappa' r} \sum_{p: |\ext(p)|=r} \left(\sum_{x: \mathrm{dist}(x,\supp(C_\alpha)) \leq r} I_{x \in \supp(p)} \right) |c_p| e^{(\kappa' + w_q(\mu + \log(2))) |\mqs(p)|}
\\&= \sup_{\alpha}  \sum_{r=1}^\infty e^{-\kappa' r}  \sum_{x: \mathrm{dist}(x,\supp(C_\alpha)) \leq r}  \sum_{p: |\ext(p)|=r, x \in \supp(p)} |c_p| e^{(\kappa' + w_q(\mu + \log(2))) |\mqs(p)|}
\\&\leq \sup_{\alpha}  \sum_{r=1}^\infty e^{-\kappa' r}  \sum_{x: \mathrm{dist}(x,\supp(C_\alpha)) \leq r}  |||O|||_{\kappa' + w_q(\mu + \log(2))}
\\&\leq  \sum_{r=1}^\infty e^{-\kappa' r}  w_c e^{\kappa r} |||O|||_{\kappa' + w_q(\mu + \log(2))}
\\& \leq \frac{w_c e^{-(\kappa' - \kappa)}}{1-e^{-(\kappa' - \kappa)}} |||O|||_{\kappa' + w_q(\mu + \log(2))}
\\& \leq |||O|||_{w_q\mu + (w_q\log(2) + \kappa + \log(2 w_c))}
\end{split}
\end{equation}

\subsection{Proof of Proposition~\ref{prop:doubleboundstriple}}

Similarly, we can bound $|||O|||_{\mu}$ in terms of $||O||_{\mu'}$ for a $\mu' = \widetilde{a} \mu + \widetilde{b}$ where $\widetilde{a},\widetilde{b}$ are $\theta(1)$ constants $\widetilde{a}=w_c$ and $\widetilde{b}=\log(4w_q)$. 

We always have connected $S$, and we will use that here. For notational convenience, define $\supp(S) := \cup_{\alpha \in S} \supp(C_\alpha)$. We will collect a few useful inequalities to start, which we will use freely in the following. If $p \in \supp(S)$, then $|\mqs(p)| \leq w_c|S|$. There are at most $4^{|\mqs(p)|} \leq 4^{w_c|S|}$ Pauli strings $p$ such that $p \in \supp(S)$. If $x \in \supp(p)$ and $p \in \supp(S)$, then there is an $\alpha \in S$ such that $x \in \supp(C_\alpha)$. If $O = \sum_{p} c_p p$, then $\forall p, \, |c_p| \leq ||O||$. 

We can first expand each term $O_S$ in the Pauli basis as $O_\mathbf{S} = \sum_{p: p \in \supp(S)} c_{\mathbf{S}, p} p$, so that $O = \sum_{p} \sum_{\mathbf{S}: p \in \supp(S)} c_{\mathbf{S}, p} p$. Then
\begin{equation}
\begin{split}
|||O|||_{\mu} &= \sup_x \sum_{p: x \in \supp(p)} \left|\sum_{\mathbf{S}: p \in \supp(S)} c_{\mathbf{S}, p}\right| e^{\mu |\mqs(p)|}
\\&\leq \sup_x \sum_{p: x \in \supp(p)} \sum_{\mathbf{S}: p \in \supp(S)} ||O_{\mathbf{S}}|| e^{\mu |\mqs(p)|}
\\&= \sup_x \sum_{\mathbf{S}} \sum_{p: x \in \supp(p), p \in \supp(S)} ||O_{\mathbf{S}}|| e^{\mu |\mqs(p)|}
\\&\leq \sup_x \sum_{\mathbf{S}} \sum_{p: x \in \supp(p), p \in \supp(S)} ||O_{\mathbf{S}}|| e^{\mu w_c |S|}
\\&\leq \sup_x \sum_{\mathbf{S}} \sum_{p: x \in \supp(p), p \in \supp(S)} (\sum_{\alpha \in \supp_c(x)} I_{\alpha \in S})  ||O_{\mathbf{S}}|| e^{\mu w_c |S|}
\\&= \sup_x \sum_{\alpha \in \supp_c(x)} \sum_{\mathbf{S: \alpha \in S}} \sum_{p: x \in \supp(p), p \in \supp(S)}   ||O_{\mathbf{S}}|| e^{\mu w_c |S|}
\\&\leq \sup_x \sum_{\alpha \in \supp_c(x)} \sum_{\mathbf{S: \alpha \in S}} 4^{w_c |S|}   ||O_{\mathbf{S}}|| e^{\mu w_c |S|}
\\&\leq \sup_x \sum_{\alpha \in \supp_c(x)} ||O||_{w_c \mu + \log(4)}
\\&\leq w_q ||O||_{w_c \mu + \log(4)}
\\&\leq ||O||_{w_c \mu + \log(4w_q)}
\end{split}
\end{equation}

\end{document}